\DocumentMetadata{testphase=new-or-1}
\documentclass[12pt]{article}
\usepackage{graphicx}
\usepackage{amsmath}
\usepackage{amssymb}
\usepackage{amsthm}
\usepackage[margin=1in]{geometry}
\usepackage{babel}
\usepackage{natbib}
\usepackage{comment}
\usepackage{enumitem}
\usepackage{booktabs}
\usepackage{caption}
\usepackage{adjustbox}
\usepackage[dvipsnames]{xcolor}
\usepackage[breaklinks=true,colorlinks=true]{hyperref}

\usepackage{thmtools}
\usepackage{xparse}
\usepackage{float}

\usepackage{algorithm}
\usepackage{algpseudocode}

\declaretheoremstyle[
    spaceabove=10pt,
    spacebelow=10pt,
    headfont=\itshape,
    bodyfont=\normalfont,
    qed=$\diamond$
]{myremark}

\declaretheoremstyle[
    spaceabove=10pt,
    spacebelow=10pt,
    headfont=\bfseries,
    bodyfont=\normalfont,
    qed=$\triangle$
]{myexample}

\declaretheorem[name=Corollary]{corollary}
\declaretheorem[name=Lemma]{lemma}
\declaretheorem[style=definition,name=Assumption]{assumpx}
\declaretheorem[style=myremark,name=Remark]{remark}
\declaretheorem[style=myexample,name=Example]{example}
\NewDocumentEnvironment{assumption}{m o}
  {
    \renewcommand{\theassumpx}{\textup{#1}}
    \IfNoValueTF{#2}{\assumpx}{\assumpx[#2]}
  }
  {\endassumpx}

\newtheoremstyle{manual}
  {}
  {}
  {}
  {}
  {\bfseries}
  {.}
  {.5em}
  {\thmnote{#3}}

\theoremstyle{manual}
\newtheorem*{exampleinner}{}

\newenvironment{examplecont}[1]{%
  \ifhmode\par\fi\vspace{0pt}%
  \pushQED{\qed}%
  \begin{exampleinner}[Example~\ref{#1} (Continued)]%
  \ignorespaces
}{%
  \popQED
  \end{exampleinner}%
  \ignorespacesafterend
}

\newcommand{\ind}{\perp\!\!\!\perp} 

\newcommand{\kmax}{k\text{-}\!\max}

\title{Interpretable Discovery from Unstructured Data: A High-Dimensional Approach}

\date{First Draft: October 2025\\This Draft: November 2026}
\author{Jacob Carlson\\Harvard University\thanks{Email: \texttt{jacob\_carlson@g.harvard.edu}. I am grateful to Isaiah Andrews, Melissa Dell, Claudia Goldin, Shengwu Li, Neil Shephard, Rahul Singh, Elie Tamer, Davide Viviano, and participants in the Harvard econometrics workshop for their feedback and suggestions. I thank Harvey Barnhard, Andrew Kao, and Karthik Tadepalli for early feedback and encouragement. There are no relevant financial relationships or other potential conflicts of interest to disclose that relate to the research described in this paper. A previous version of this paper was circulated under the title ``Making Interpretable Discoveries from Unstructured Data: A High-Dimensional Multiple Hypothesis Testing Approach.'' I thank Claude for research assistance in literature review, proof checking, and code generation.}\\
}

\begin{document}

\maketitle
\begin{center}
\vspace{-2\topsep}
\href{https://jscarlson.github.io/assets/carlson_discovery.pdf}%
       {\fbox{Updated frequently: please click here for the latest version}}
\vspace{\topsep}
\end{center}

\begin{abstract}
I propose an automatic, general-purpose framework for making discoveries from unstructured data (e.g., text data from open-ended surveys of economic beliefs). The framework leverages recent methods from the literature on AI interpretability to transform unstructured datasets into high-dimensional, structured datasets of interpretable concept measures; specifies parameters and null hypotheses of interest based on these concept measures; tests these hypotheses using algorithms validated by new results in high-dimensional multiple testing, producing a selected set (``discoveries''); and automatically generates and evaluates natural language descriptions of the concepts associated with these discoveries. The proposed framework has few researcher degrees of freedom, is robust to data snooping, mitigates under-exploration, and facilitates fast and inexpensive sensitivity analysis and replication. I revisit applications to recent descriptive and causal analyses of unstructured data in empirical economics, and find this framework is able to automatically replicate existing discoveries, add nuance to others, and make entirely new discoveries as well.
\end{abstract}

\section{Introduction}

Unstructured data provide opportunities to measure economic phenomena that would be otherwise unmeasurable. In some settings, unstructured data enable high quality, scalable measurement of an economic feature that theory or prior empirical research suggested was important. For example, satellite data may be used for high-frequency measurement of changes in local economic activity \citep{khachiyan_using_2022}. In other settings, unstructured data enable researchers to measure \textit{novel} features of economic interest---previously uncaptured by both theory and structured data collection---which can serve as inputs to downstream econometric analysis and facilitate new economic \textit{discoveries}.\footnote{For some other recent examples of analysis of unstructured data in economics, consider: speech recordings from FOMC meetings used to better understand monetary policy \citep{gorodnichenko_voice_2023}; videos of start-up pitches used to study entrepreneurship and investment \citep{hu_persuading_2025}; mug shots used to study judicial decisionmaking \citep{ludwig_machine_2024}; visual art and written narratives used to infer long-run living standards \citep{gorin_state_2025,lagakos_american_2025}; and qualitative interviews used to better understand the impacts of RCT treatments \citep{bergman_creating_2024,krause_impact_2025,ahmadi_leveraging_2026}.} For example, consider the following three settings, motivated by real applications in recent empirical economics research. 

\begin{example}[Information treatment RCT, e.g., \cite{bursztyn_justifying_2023,andre_narratives_2026}]\label{ex:rct}
A researcher is interested in running a randomized controlled trial (RCT) with an information treatment. A key outcome of the experiment is the participant's reaction to the information treatment, and the researcher would like to discover any  systematic differences in reactions across the treatment and control groups. The full space of economically meaningful reactions is unknown to the researcher ex-ante. As such, for all participants, the researcher collects text responses to open-ended questions eliciting reactions to the experimental intervention, and is interested in learning any relevant causal effects on the presence of concepts surfaced in these text responses. 
\end{example}

\begin{example}[Descriptive survey analysis, e.g., \cite{andre_subjective_2022,stantcheva_why_2024}]\label{ex:desc}
A researcher is interested in understanding the attitudes of some population of interest towards some particular (economic) policy. To learn about diverse population attitudes without priming respondents to discuss particular topics, the researcher runs a survey on a representative sample from the population of interest that asks many open-ended questions about attitudes and beliefs, to which survey respondents reply freely in text. The researcher is interested in describing any frequently discussed concepts that are surfaced from this open-ended survey, as well as how concept frequencies differ across various subgroups in the sample.
\end{example}

\begin{example}[Labor RCT on human-AI productivity, e.g., \cite{noy_experimental_2023}]\label{ex:hci}
A researcher is interested in understanding the causal effects of AI adoption on worker productivity and creativity. As such, the researcher runs an experiment in which participants are asked to complete a short professional writing task, with half of participants randomly receiving access to a frontier LLM (the treatment group) and the other half receiving no aid (the control group). The researcher is interested in discovering any interpretable features of the writing that differ systematically across treatment and control.
\end{example}

To address these examples---among many others---this paper proposes an automatic, general-purpose framework. At a high level, the framework consists of transforming an unstructured dataset into a high-dimensional, structured dataset of interpretable concept measurements and then analyzing this transformed dataset using new methods for multiple hypothesis testing. More specifically, the framework can be broken down into the following four steps, which are also depicted in the schematic in Figure \ref{fig:schematic}.
\begin{figure}[!t]
    \centering
    \adjustbox{center, width=0.4\textwidth}{\includegraphics{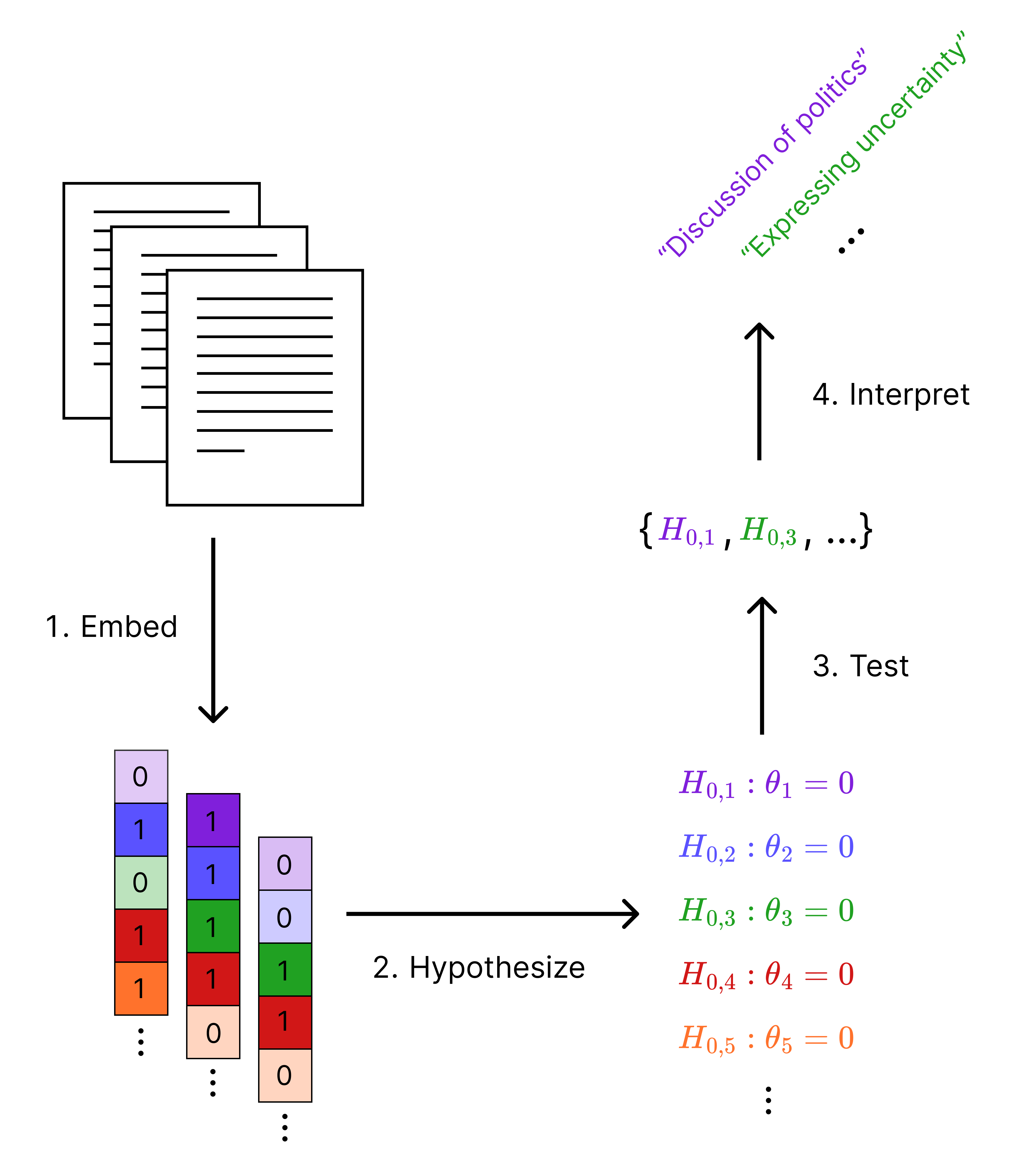}}
    \caption{A schematic for the proposed framework for discovery from unstructured data.}
    \label{fig:schematic}
\end{figure}

1. \underline{Creating concept embeddings}: each unstructured data point is mapped to a ``concept  embedding,'' which is a binary vector.\footnote{The name ``concept embedding'' is chosen to be consistent with other parts of the CS and AI interpretability literatures that make use of similar objects, e.g., \cite{bhalla_interpreting_2024,jiang_interpretable_2025}.} Each element of the vector is intended as an indicator for the presence of a particular human-interpretable feature---or, informally, a ``concept''---expressed in the unstructured data point. Though many choices of concept embedding are compatible with the framework, a default recommendation is to use a pretrained \textit{sparse dictionary learning} model (newly popular in the AI interpretability literature) as a concept embedding method. By construction, concept embedding vectors are \textit{high-dimensional}, aiming to catalog a vast number of human-interpretable concepts and thus broaden the scope for making novel, economically meaningful discoveries later on.
    
\begin{examplecont}{ex:rct}
    Each experimental participant's post-treatment text response is passed through an LLM equipped with pretrained sparse autoencoder (SAE), a popular implementation of a sparse dictionary learning model \citep{bricken2023monosemanticity}.\footnote{See, e.g., the blogpost \url{https://www.goodfire.com/blog/a-practical-guide-to-saes} for a more thorough introduction.} The SAE's feature activations are pooled and thresholded at zero to serve as a concept embedding.  
    There is one concept embedding of each text response (and one text response per experimental participant), each of which is of dimension $p\approx 10^4$. 
    The set of concepts measured by this SAE-derived concept embedding approximates the huge set of interpretable text features that the LLM is capable of ``thinking about'' upon reading a text. 
    For example, concept number 1 in the concept embedding may indicate the presence of exclamation points, concept number 2 may indicate discussion of politics, concept number 3 may signal discussion of health issues, and so on.
     \end{examplecont}

2. \underline{Specifying concept-level parameters and hypotheses}: the researcher, using their new dataset of concept embeddings, specifies $p$ concept-specific hypothesis tests based on $p$ concept-level parameters of interest---one for each concept cataloged in the concept embedding. The researcher has great flexibility in designing meaningful concept-level parameters, hypotheses, and test statistics at this stage of the framework, which may employ data sources beyond the concept embeddings themselves (e.g., covariates from structured data sources).
    \begin{examplecont}{ex:rct}
    The researcher computes a  difference in means across treatment and control groups for each \textit{concept indicator} (a binary entry in the concept embedding) to test a total of $p$ null hypotheses of zero average treatment effect (ATE) on the probability of a given concept appearing in a text response. The $p$ parameters of interest are the $p$ concept-by-concept ATEs.
     \end{examplecont}

3. \underline{High-dimensional multiple testing}: all $p$ concept-specific hypotheses are tested using statistical procedures with high-dimensional, generalized familywise error guarantees validated by novel results in high-dimensional multiple hypothesis testing and central limit theory. In particular, these results prove that the algorithms of \cite{romano_control_2007} for controlling the $k$ familywise error rate ($k$-FWER)---the probability of making $k$ or more false rejections---are valid in high-dimensions under mild regularity conditions. Control of the $k$-FWER for $k>1$ is less conservative than simple FWER control---control over the probability of making \textit{any} false rejections\footnote{Note that FWER is equivalent to $k$-FWER when $k=1$.}---leading to more powerful tests. Moreover, the researcher may invert these tests to form \textit{generalized simultaneous confidence intervals}: the probability that $k$ or more parameters fall outside the simultaneous random intervals is less than or equal to a desired level $\alpha$. The researcher can also ``augment'' the $k$-FWER controlling procedure to achieve high-dimensional control of false discovery exceedance (FDX)---a high probability guarantee in large samples that the ratio of false discoveries to total rejections does not exceed a prespecified threshold.%
\begin{examplecont}{ex:rct} An algorithm based on the procedures of \cite{romano_control_2007} is deployed to simultaneously test each of the $p$ null hypotheses of no average treatment effect, guaranteeing $k$-FWER control at level $\alpha$ (in large samples) for a researcher chosen value of $k$, even though $p \gg n$ (inference is ``high-dimensional''). 
    Critical values produced by the algorithm may also be used to form generalized simultaneous confidence intervals for the difference-in-means point estimates of the $p$ concept-by-concept ATEs.
     \end{examplecont}

4. \underline{Local automatic interpretation}: the output of the framework so far is a selected set of rejected, concept-level null hypotheses (as well as corresponding generalized simultaneous confidence intervals and point estimates), a.k.a., ``discoveries.'' However---much like the ``topics'' learned by LDA models \citep{blei2003latent}---the ``concepts'' encoded by sparse dictionary learning methods do not come with natural language descriptions at first. That is, we know, e.g., the first entry of the concept embedding encodes a human-interpretable concept, and that we rejected the null hypothesis associated with it, but we do not yet have an English description for what that concept is. As such, in the final stage of the framework, the researcher generates natural language descriptions of the concepts in the concept embedding using a variant of a popular \textit{automatic interpretation} (or ``autointerp'') method that is tailored to the unstructured data distribution being explored.\footnote{Many pretrained sparse dictionary learning models come with ``off-the-shelf'' feature descriptions, but they are usually of poor quality when transported to DGPs the SAE was not trained on, and thus benefit from refinement based on information in the dataset of interest. This is discussed further in Section \ref{sec_auto}.} A new statistical formalization of a workhorse autointerp evaluation method from \cite{paulo_automatically_2024} is used to infer the quality of these ``local'' autointerp descriptions for the data distribution being analyzed. 
    \begin{examplecont}{ex:rct}
    The researcher now has in hand some subset of the $p$ null hypotheses that have been rejected: a set of indices associated with particular elements of the concept embeddings. In order to provide an English language description for each of the selected concepts, a dedicated ``explainer LLM'' (which may be closed- or open-source) reads examples of texts in the researcher's dataset for which a given feature is (highly) present  and generates an English description for what concept is encoded. The credibility of this LLM-derived description of the concept is then evaluated with a specially designed classification task on a held-out evaluation sample of the researcher's unstructured dataset.  
     \end{examplecont}

The final output of the proposed framework is a set of concept-level discoveries (rejections of nulls) that come with generalized familywise error rate guarantees (or, point estimates for associated parameters of interest and generalized simultaneous confidence intervals); natural language descriptions of those discoveries; and scores measuring the quality of those descriptions.
\begin{examplecont}{ex:rct}
The researcher discovers that the null of no average treatment effect is rejected for concepts indexed $\{2, 101, 5030, ...\}$. Concept 2 is described by an explainer LLM as activating on ``mentions of politics'' and this description is estimated as having a high ``quality'' score, i.e., the information treatment credibly causes increased discussion of politics; and so on.  
 \end{examplecont}

The proposed framework for discovery has several distinct scientific benefits. Firstly, the proposed framework is fully \textit{automatic} in execution, requiring no ``human-in-the-loop'' to guide discovery. As such, it has very few \textit{researcher degrees of freedom} and facilitates inexpensive\footnote{All empirical examples in this paper were computed on Google Colab notebooks using a single A100 GPU, access to which only costs \$10 per month (and is free for students and academic faculty), and takes approximately one hour to run using this hardware.} and practical \textit{replication and sensitivity analysis}.\footnote{The few choices a researcher must make in the proposed framework, such as a choice of concept embedding, can also easily be preregistered, further hampering the possibility that a researcher applies this framework with many different models to select on favorable results. Even without preregistration, cherry-picking across sparse dictionary learning models for the purposes of reverse engineering a desired result is unlikely to succeed given that, e.g., SAE features tend to exhibit some degree of universality across architectures (c.f., \cite{wang2025towards})---another reason I advocate their default adoption as concept embeddings.} This stands in contrast to conventional approaches for discovery from unstructured economic data (including both ``hand-coding'' approaches \citep{haaland_understanding_2024} and more modern ``hypothesis generation'' approaches \citep{ludwig_machine_2024}), which typically require researcher discretion and intervention over which measurements to take (hand-coding) or test (hypothesis generation). Such researcher degrees of freedom, in turn, open up room for cherry-picking ``desirable'' features to measure and/or test; create opportunities for conscious or unconscious data snooping (breaking the ``firewall'' of statistically independent generation and testing \citep{mullainathan_machine_2017}); and make reproducibility and sensitivity analysis challenging (or infeasible, when discretionary choices are never documented or formalized). As new advances in AI tools (e.g., AI interviewers \citep{geiecke_conversations_2024}) make it easier than ever for researchers to collect unstructured datasets, controlling researcher degrees of freedom and reproducibility in econometric methods for analyzing unstructured data may only become more important.

Secondly, the proposed framework suggests methods for \textit{high-dimensional} measurement and testing. This also stands in contrast to conventional approaches to discovery in economics, which usually artificially delimit the scope of interpretable measurements to take or hypotheses to test, leading to \textit{under-exploration} of the unstructured data and ``incompleteness'' of economic insights or explanations  \citep{fudenberg_measuring_2022,modarressi_causal_2025}. 
That is, by keeping measurement and testing high-dimensional, the proposed framework is better able to help researchers see past the ``streetlight'' of their own economic intuitions for what might be important to measure and test with their unstructured datasets, broadening the possibilities for making truly surprising discoveries.\footnote{The ``streetlight effect'' describes the cognitive bias that occurs when researchers only think to measure that which is easiest or most obvious to measure: idiomatically speaking, the researcher is ``only searching for their lost keys where the light is shining.'' The same problem is sometimes cast as failing to measure or discover important ``unknown unknowns.''} The ability to take many interesting and interpretable measurements and test them all with meaningful power under a generalized Type I error guarantee---even with access to only a small dataset of several hundred observations---is a direct result of the contributions made within this framework.

More specifically, the proposed framework is supported by distinct contributions to econometric theory, computational economics, and AI interpretability, which themselves may be of independent interest, including:
\begin{enumerate}
    \item Extensions of two frontier econometric toolkits---extending the high-dimensional ``many approximate means'' framework of \cite{belloni_high-dimensional_2018} to the study of the $k$-th largest coordinate of sums of independent random vectors, and extending the application of the $k$-FWER controlling procedures of \cite{romano_control_2007} to the high-dimensional setting. These results further unlock high-dimensional FDX control via the augmented $k$-FWER testing procedures outlined in \cite{genovese_exceedance_2006}.
    \item A new approach to (and code implementation for) doing \textit{high-dimensional} hypothesis generation with \textit{small} datasets, which leverages pretrained sparse dictionary learning models of LLM residual streams and a ``local'' automatic interpretation procedure, complementing frameworks such as \cite{movva_sparse_2025}.
    \item A new statistical inference-based formalization of a popular framework for evaluating automatically generated SAE feature interpretations from \cite{paulo_automatically_2024}, and associated estimation and inference results.
\end{enumerate}


The remainder of this paper is structured as follows: Section 2 discusses related literature; Section 3 describes the framework in detail; Section 4 applies the framework to three recent papers making use of unstructured data for discovery \citep{andre_narratives_2026,bursztyn_justifying_2023,noy_experimental_2023}; and Section 5 concludes.

\section{Related Literature}\label{sec_lit}

The framework proposed in this paper is related to recent works from literatures spanning economics, statistics, and computer science.

\paragraph{Econometric methods for unstructured data.}
Motivated by the capabilities of modern AI/ML methods for extracting structure from unstructured economic data (see, e.g., \cite{mullainathan_machine_2017,gentzkow_text_2019,ash_text_2023,dell_deep_2025}), new econometric and statistical frameworks have been developed to facilitate principled statistical inference on low-dimensional features (predictions) learned from unstructured datasets under supervision, e.g., \cite{angelopoulos_prediction-powered_2023,ludwig_large_2024,carlson_unifying_2025,rambachan_program_2024}. This most recent wave of econometric literature emphasizes nonparametric frameworks compatible with black-box AI models and centers debiasing methods, as opposed to model-based approaches, e.g.,  \cite{gentzkow_measuring_2019,battaglia_inference_2024}. Though the framework presented in this paper is similarly concerned with rigorous inference on unstructured data without making parametric assumptions, its goal is primarily to pursue high-dimensional inference on machine learned features to facilitate unsupervised discovery, as opposed to pursuing low-dimensional inference on known features of interest measured with AI predictions.

The proposed framework is most similar to existing works such as \cite{ludwig_machine-learning_2017,modarressi_causal_2025}, both of which provide formal econometric frameworks for discovery from high-dimensional economic data. Both of these works focus specifically on causal inference in RCTs, unlike the proposed framework, which generally treats any inference question that can be posed as a hypothesis test (including causal inference based on RCT data). Both \cite{ludwig_machine-learning_2017} and \cite{modarressi_causal_2025} employ the intuition of a ``reverse regression'' to formally test whether a distribution of (high-dimensional) outcomes is the same across treatment and control groups as a method for discovery. However, both papers note that such a coarse test is often insufficient for developing scientific insight into the causal effects of treatment, motivating methods for probing \textit{what} changed across the outcome distributions. \cite{modarressi_causal_2025} specifically explores this idea for (unstructured) text data, suggesting a semi-automated framework based on a human-LLM hybrid pipeline that hypothesizes a low-dimensional set of  ``causal themes'' that differ across treatment and control and then scores texts per these themes to evaluate ATEs in a held-out sample.

\paragraph{High-dimensional multiple testing.} There is an extensive literature on multiple hypothesis testing (or ``selective inference'' more broadly) in statistics, biostatistics, and economics (see, e.g., \cite{romano_hypothesis_2010} for a review in econometrics). The methods in this literature span asymptotic and finite sample valid frameworks, and low- and high-dimensional settings. In particular, for the problem considered in this paper, the literature on asymptotically valid, high-dimensional selective inference is most relevant, as such methods permit making very few assumptions about the data (c.f., assuming independence or a known flavor of dependence \citep{benjamini_controlling_1995,fithian_conditional_2022}); readily incorporate resampling methods that improve power (relative to guarding against worst-case dependence, c.f., Bonferroni corrections or the method of \cite{benjamini_control_2001}); and do not prescribe particular hypothesis tests of interest or make assumptions about distributions of individual random variables (c.f., the literature on ``knockoffs'' \citep{barber_controlling_2015,candes_panning_2018,barber_knockoff_2019,janson_familywise_2016}).

The new results for high-dimensional central limit theory and multiple hypothesis testing appearing in this paper are derived by combining and extending recent work on high-dimensional central limit theory and Gaussian multiplier bootstraps for order statistics \citep{belloni_high-dimensional_2018,chernozhuokov_improved_2022,ding_gaussian_2025}. The formal results on high-dimensional central limit theory in this paper are best viewed as extensions of some core results of the ``many approximate means'' framework of \cite{belloni_high-dimensional_2018}, moving the framework from the study of the maximal order statistic to the study of the \textit{largest} order statistics by leveraging new results from a recent high-dimensional probability literature \citep{ding_gaussian_2025,feng_high_2026}. 

These high-dimensional CLTs in turn support new high-dimensional ($p\gg n$) analysis of the $k$-FWER controlling step-down procedures of \cite{romano_control_2007,romano_formalized_2008}---which were originally studied in a low-dimensional (fixed $p$) setting---leading to theorems that show such procedures remain valid for asymptotic $k$-FWER control in high-dimensions under several important technical conditions (chiefly, that $k$ is fixed, or ``sufficiently small''). These results in asymptotic high-dimensional $k$-FWER control complement existing methods for asymptotic high-dimensional control of the FWER \citep{belloni_high-dimensional_2018} and asymptotic high-dimensional control of the FDR under weak dependence \citep{liu_phase_2014,belloni_high-dimensional_2018}.

\paragraph{Hypothesis generation.}
There is a growing literature in economics and computer science on ``hypothesis generation'' \citep{ludwig_machine_2024,zhong_goal_2023}---that is, using AI/ML frameworks to systematically generate and measure research-relevant, interpretable features from unstructured data \citep{dell_measurement_2026}. From the perspective of this literature, the proposed framework can be seen as advocating for \textit{high-dimensional} hypothesis generation, paired with high-dimensional hypothesis evaluation (testing). More specifically, concept embeddings can be viewed as the result of applying a hypothesis generator to some unstructured evaluation sample, the result of which is subjected to formal statistical analysis. Correspondingly, then, the new methods in high-dimensional multiple testing in this paper can also be seen as a new technique for evaluating many generated hypotheses.

Sparse autoencoders have recently been proposed and applied as hypothesis generators for social science research questions \citep{movva_sparse_2025,peng_use_2025,wang_your_2026}. From the perspective of this hypothesis generation sub-literature, one can consider that the hypothesis generation procedure associated with the proposed SAE-based concept embeddings was that of the original self-supervised training of the SAE model in an external corpus. The application of SAE as hypothesis generator in this paper is based on a model \textit{pretrained} on an LLM's residual stream, eschewing the need for a large SAE training corpus and thereby accommodating analysis with small datasets. That said, the proposed framework would likewise be compatible with SAEs trained from scratch (based on dense embeddings or otherwise, c.f., \cite{movva_sparse_2025}). Furthermore, unlike other works in this space, SAE features are not filtered by their ability to predict some outcome or treatment. Filtering is not pursued in the proposed framework because filtering through supervision can: (1) artificially reduce the breadth of interpretable measurements taken by the SAE; (2) introduce another lever for researcher discretion/intervention; and (3) inappropriately prune features that, e.g., have large associated causal effects but do not explain much variation in treatment.\footnote{``Many true [causal] hypotheses need not be visible in raw correlations... Other procedures would be required to surface those. Our focus here is on searching for true hypotheses that are visible in raw correlations'' \citep{ludwig_machine_2024}.} That said, using structured covariates to supervise the dimensionality reduction of a concept embedding is theoretically compatible with the proposed framework (under appropriate sample splitting) and may have important applications in settings considered beyond this paper's empirical examples.

\paragraph{AI interpretability.}
There is a large literature in computer science on machine learning interpretability (see, e.g., \cite{doshi-velez_towards_2017}). The machine learning interpretability methods leveraged by the present framework originate from a nascent though highly active literature known as ``mechanistic interpretability,'' which seeks to develop methods for interpreting LLM behavior via quantitative analyses of model internals (e.g., activations, weights). In particular, the sparse dictionary learning methods for LLM interpretability implemented in this framework are called ``sparse autoencoders'' (SAEs), and were developed in \cite{bricken2023monosemanticity,templeton2024scaling}.\footnote{SAEs are very wide autoencoders inserted at various layers of an LLM that are trained to intercept and reconstruct activations from model internals (such as the residual stream, MLPs or the like) under a sparsity-inducing penalization on a reconstruction loss. Under the ``linear representation hypothesis'' and ``superposition hypothesis'' (see, e.g., \cite{bricken2023monosemanticity} for more discussion), these sparse autoencoders are thought to act as an overcomplete basis of the space of concepts in text leveraged by an LLM to make next token predictions, encouraging learned autoencoder latents to have monosemantic interpretations.} \!\footnote{There has been much debate about the success of SAEs in pursuit of AI interpretability (e.g., \cite{leask_sparse_2025}), though most criticisms still support the notion that SAEs and other sparse dictionary learning methods have a comparative advantage in discovering as opposed to detecting concepts of interest in text \citep{peng_use_2025}---which is how they are applied in this framework.} Other sparse dictionary learning methods that have been proposed as competitors to SAEs (e.g., transcoders \citep{paulo_transcoders_2025}) may similarly serve as concept embedding methods. The proposed framework is also related to recent independent efforts in ``data-centric interpretability,'' such as \cite{jiang_interpretable_2025}, who leverage SAE activations to create interpretable concept embeddings of text in order to facilitate (purely descriptive) analysis of datasets relevant to LLM behavior.

SAEs and similar methods learn quantitative measures of interpretable concepts, but do not on their own generate natural language (e.g., English) descriptions for these learned concepts. As such, a related literature on ``automatic interpretability'' (or ``autointerp'') methods has become active, which seeks to coherently use LLMs to describe the features discovered by sparse dictionary learning methods at scale. Recent important papers in this literature include: \cite{bills2023language,shaham_multimodal_2024,paulo_automatically_2024,rajamanoharan_jumping_2024}. Specifically, the autointerp framework proposed in this paper adapts and  formalizes insights from the ``detection scoring'' method of \cite{paulo_automatically_2024}.

Though the empirical examples considered in this paper handle one of the most popular unstructured data types---text---for which sparse dictionary learning-based interpretability methods are the most mature, sparse dictionary learning techniques have been successfully applied to many other modalities, including audio and images \citep{abdulaal_x-ray_2024,bhalla_interpreting_2024,fry2024towards,daujotas2024case,pluth_sparse_2025,fel_structuring_2026}. Future work would seek to apply the proposed framework to unstructured datasets in these other data modalities.

\section{Framework}

We now consider each step of the proposed framework in detail. In what follows, I define $[p] := \{1, \dots, p\}$, and let $x \mapsto \log x$ be the natural logarithm. Further, for any vector $s \in \mathbb{R}^p$ and a subset of indices $K \subseteq[p]$, let $s_K \in \mathbb{R}^{|K|}$ denote the subvector indexed by $K$, such that $\left(s_K\right)_q=s_{j_q}$ for $j_q \in K$ and $j_1<j_2< \cdots<j_{|K|}$. I use $s \mapsto s_{[k]}$ as notation for the function that selects the $k$-th \textit{largest} coordinate of the vector $s$.

\subsection{Creating Concept Embeddings}\label{sec_featdict}

Consider that the researcher has access to a dataset of size $n$, $\{(W_i, Z_i)\}_{i=1}^n$, which is sampled i.i.d. from some population of interest represented by distribution $P$.\footnote{That is, the fundamental source of uncertainty in this framework is sampling uncertainty, driven by an underlying sampling experiment in which units are sampled i.i.d. from $P$, and each unit is associated with a measurement $(W_i, Z_i)$.} The $Z_i \in \mathcal{Z}$ are unstructured data points and the $W_i \in \mathcal{W}$ are any other observed covariates of interest. The space $\mathcal{Z}$ is typically high-dimensional but semantically shallow, e.g., if each $Z_i$ was a $224 \times 224$ pixel image, $\mathcal{Z}$ might then be the space of all $3 \times 224 \times 224$ arrays of RGB values, or if $Z_i$ was text that was truncated up to some maximum length, $\mathcal{Z}$ might be the space of all binary matrices of a certain dimensionality, which are concatenated one-hot encodings for each word, term, or token in the text with respect to a specific pre-defined vocabulary.

\begin{examplecont}{ex:rct}
The $Z_i$ is a text response elicited for experimental participant $i$; the $W_i \in \mathcal{W} = \{0,1\}$ is a treatment status indicator for experimental participant $i$. 
 \end{examplecont}

The first stage of the framework relies on the researcher having access to a function which creates high-dimensional, human-interpretable concept embeddings from unstructured data points, which I denote as $h: \mathcal{Z} \to \{0,1\}^p$.\footnote{Note that, consistent with the hypothesis generation literature, $h$ can be thought of as the measurement function resulting from some learned, \textit{high-dimensional} ``hypothesis generator.'' It can also be thought of as a \textit{high-dimensional} ``codebook function'' in the conceptual framework of \cite{egami_how_2022}.}
The researcher has freedom to implement any high-dimensional concept embedding strategy of interest, so long as it is constructed in a way that is statistically independent of the dataset $\{(W_i, Z_i)\}_{i=1}^n$. However, as a computationally practical and widely applicable default, I suggest computing concept embeddings using pretrained \textit{sparse dictionary learning} models.\footnote{Alternative examples of concept embeddings include: a concept embedding where each index records the presence or absence of particular $n$-gram (cheap computation, but not very rich measurement) or one where each index indicates absence or presence of a Wikipedia category (\url{https://en.wikipedia.org/wiki/Wikipedia:Contents/Categories}) operationalized with zero-shot LLM classification \`a la \citet{asirvatham_gpt_2026} (computation is expensive, but measurement is very rich).}
In practice, this is accomplished by passing an unstructured data point to a deep neural network equipped with a pretrained sparse dictionary learning model (e.g., an LLM equipped with pretrained SAE when $Z_i$ are texts); recording the feature activations from the SAE; and then aggregating these SAE feature activations (across tokens) to form a vector of indicators for the presence of each feature measured by the SAE, indexed $j \in [p]$.

There are many possible ways for a researcher to design a concept embedding based on a sparse dictionary learning model, though an especially straightforward choice is to record the SAE activations at a single layer of the deep neural network, setting the $j$-th entry of the concept embedding equal to 1 if the $j$-th SAE feature activated on any token of the input text, and 0 otherwise:
\[
h(z)_j := \begin{cases}
&1 \text{ if SAE feature $j$ at layer $l$ activates on any token of } z, \\
&0 \text{ otherwise.}
\end{cases}
\]
Concrete choices for the SAE architecture and layer $l$ will be discussed in the empirical examples in Section \ref{sec_empirical}.

Let $Y_i := h(Z_i) \in \{0,1\}^p$ denote the concept embedding of an unstructured data point $Z_i$. Each ``concept indicator'' $Y_{ij} \in \{0,1\}$ has the natural interpretation that a particular concept $j$ appeared in the unstructured data point $Z_i$.\footnote{SAE feature activations are typically positively valued scalars, where magnitude is thought to correspond to some notion of intensity of activation. Other works in mechanistic interpretability consider using these raw activation values directly, perhaps with max or average pooling across tokens. However, presently, magnitude of activation is not a well-understood or highly interpretable quantitative property of sparse dictionary learning methods, and as such the default $h$ function advocated for in the proposed framework does not incorporate this information. That said, leveraging feature activation intensity is empirically shown to be useful for deriving good autointerpretation generation strategies; see Section \ref{sec_auto} for more on this.}

\begin{examplecont}{ex:rct}
Each text response $Z_i$ is passed through an open-source LLM equipped with a pretrained SAE. The SAE activations for each token of $Z_i$ at the $l$-th layer of the LLM are max-pooled (i.e., the maximum token activation is used as the activation summary for the entire text) and then thresholded at zero, yielding a $p\approx 10^4$-dimensional binary vector $Y_i$ for each experimental participant $i$. Each $Y_{ij} \in \{0,1\}$ is a binary indicator of a particular concept being expressed in experimental participant $i$'s text response, e.g., $Y_{i7}$ could indicate whether or not experimental participant $i$'s text response contains phrases indicating uncertainty.
 \end{examplecont}

Application of the $h$ function can be thought of as effectively producing a new, semantically rich, high-dimensional, and interpretable dataset of concept measurements with associated covariates $\{(Y_i,W_i)\}_{i=1}^n$ from the original unstructured dataset $\{(Z_i,W_i)\}_{i=1}^n$. The crux of the proposed framework is that simultaneous inference on parameters derived from each concept indicator $Y_{ij}$---or each $(Y_{ij}, W_i)$ pair---is desirable because:
\begin{enumerate}
    \item Each $Y_{ij}$ is intended to indicate the presence of a rich, \textit{human-interpretable} concept, and therefore parameters derived from $Y_{ij}$ for any given $j$ (and associated hypotheses) inherit this interpretability, e.g., $E_P[Y_{ij}]$ is the probability of human-interpretable concept $j$ appearing in an unstructured data observation in population $P$.
    In practice, for concept embeddings based on pretrained SAEs, the concepts measured might not capture entirely mutually exclusive concepts in unstructured data---leading to some redundancy in hypotheses---but still provide meaningful \textit{coverage} of distinct concepts, i.e., jointly, SAE features can be shown to tile manifolds that constitute the meaningful concepts an LLM ``thinks about'' upon reading a text, in a relevant activation space \citep{bhalla_sparse_2026}. 
    \item Collectively, the $p$ indicators in $Y_{i}$ for any given observation $i$ catalog a vast space of interpretable concepts. Again, for SAEs, this catalog is intended to span the huge set of concepts that a frontier LLM uses for reasoning while doing next token prediction (or other pre-/mid-/post-training tasks) in internet-scale datasets. As such, simultaneous inference on parameters derived from all of these concept indicators greatly reduces the possibility that a researcher misses out on measuring important ``unknown unknowns'' in their unstructured data, mitigating \textit{under-exploration} and expanding  the scope for making truly surprising discoveries. The fact that these concept measures were curated in a data-independent way further shields inference from scientific and statistical problems associated with \textit{data snooping}.
\end{enumerate}

Towards exploiting this wealth of interpretable information created through concept embedding, we turn to specifying formal hypothesis tests of parameters of interest in the next step of the framework.

\subsection{Specifying Concept-Level Parameters and Hypotheses}\label{sec_feathypo}

Based on the interpretable dataset of concept embeddings and associated covariates $\{(Y_i,W_i)\}_{i=1}^n$, consider defining $p$ ``one-sided'' concept-level hypotheses $\{\tilde H_{0,j}\}_{j \in [p]}$, where  
\[\tilde H_{0,j}: \theta_j(P) := E_P[X_{ij}] \leq 0, \quad X_{ij} = g(W_i, Y_{ij}),\]
for some measurable function $g$, or, alternatively, a set of ``two-sided'' hypotheses of interest $\{ H_{0,j}\}_{j \in [p]}$ given by
\[ H_{0,j}: \theta_j(P) := E_P[X_{ij}] = 0.\] As shorthand, I denote $X := \{X_{i}\}_{i=1}^n$. Letting the concept-level parameter of interest $\theta_j(P)$ be an expectation of some flexible function $g$ of a concept indicator and covariates accommodates many natural estimands of interest in the social sciences, though extensions to a wider class of parameters that are only associated with \textit{asymptotically linear} estimators will be discussed in detail in Section \ref{sec_mam}. 

\begin{examplecont}{ex:rct}
To infer average treatment effects on concept indicators, for treatment indicator $W_i$, fixed probability of treatment $\pi$ (bounded away from zero and one), and potential outcomes $Y_{ij}(w)$ under SUTVA,
\[
H_{0,j}: E[Y_{ij}(1) - Y_{ij}(0)] = E[X_{ij}]=E\left[ \frac{W_{i} - \pi}{\pi(1-\pi)}Y_{ij}\right] = 0.
\]
In other words, the researcher sets $X_{ij} = g(W_i, Y_{ij}) :=  \frac{W_{i} - \pi}{\pi(1-\pi)}Y_{ij}$, a Horvitz–Thompson transformation for computing a difference in means.
 \end{examplecont}

Test statistics for such hypotheses may be formed as
\[
T_{n,j} := \frac{1}{\sqrt{n}} \sum_{i=1}^n X_{ij}, \quad T_{n} := \frac{1}{\sqrt{n}} \sum_{i=1}^n X_{i},
\]
with absolute values taken as appropriate.

Though inference on all $p$ hypotheses facilitates rich, interpretable exploration of the unstructured dataset, inference in this setting is also challenging because:
\begin{enumerate}
    \item The $p$ is large, so simultaneous inference on (testing hypotheses for) all $j \in [p]$ parameters is a setting where we want \textit{multiple hypothesis testing} guarantees. To see this, for the running Example \ref{ex:rct}, consider that of the $p\approx 10^4$ concepts we measure, we may anticipate that the presence of only 100 or so are causally affected by treatment. However, we do not know the identities of these concepts ex-ante, so we must search through all $p$ concepts for them. If each concept-level hypothesis is tested at the standard 5\% level, even in the ideal case that we conclude all 100 truly affected concepts are significant, we also expect to label 5\% of the 9,900 unaffected concepts as significant, leading to a selected set of about $100 + (0.05 \times 9900) = 595$ concepts, of which $0.05 \times 9900/595 \approx 83\%$ are false discoveries. Multiple testing methods exactly aim to control error rates regarding the false discoveries a researcher makes after running many tests.
    \item Not only is $p$ large, but for most social science applications of interest $p \gg n$ (e.g., the experiment from Example \ref{ex:rct} contains a few hundred participants, but the concept embedding contains tens of thousands of concept indicators), so simultaneous inference on all parameters is a \textit{high-dimensional} multiple hypothesis testing problem. 
    \item Each estimator, test statistic, or p-value formed from the $Y_{ij}$ for each $j \in [p]$ for the purposes of inference is plausibly statistically \textit{dependent} on every other in a complicated way, ruling out multiple hypothesis testing approaches that assume independence or specific, known forms of dependence (e.g., PRDS \citep{benjamini_control_2001}).
    \item The intent of the analysis is discovery, so the desired form of selective (familywise) error control is ideally \textit{not too conservative}---similar to the original motivation for controlling selective error rates like the FDR \citep{benjamini_controlling_1995} (discussed further in Section \ref{sec_fdx}) instead of the FWER (probability of making \textit{any} false rejections).
\end{enumerate}

What is needed, then, for principled simultaneous inference on parameters based on all $Y_{ij}$ is a \textit{dependence-agnostic, high-dimensional multiple testing procedure with control over a generalized familywise error rate}. To be as well-powered as possible, we want testing procedures that employ resampling methods for estimating the true covariance across test statistics, such that conservative protection against worst-case dependence is not required. To keep things as general as possible, we also want to allow for only asymptotically valid test statistics and p-values, as is common in much of econometric analysis. In the following section, theory and corresponding statistical procedures are developed to achieve exactly these aims.

\subsection{High-Dimensional Multiple Hypothesis Testing}\label{sec_mht}

I now state and prove the validity of statistical procedures that provide high-dimensional $k$-FWER control asymptotically, i.e., in the large sample limit, control of the probability of making $k$ or more rejections of true nulls under distribution $P$:
\[
k\text{-FWER}_P := P\{\text{reject at least } k \text { hypotheses } H_{0,j}: j \in I(P)\}
\]
for $I(P)$ the set of indices of true null hypotheses under distribution $P$. Naturally, $k$-FWER control for $k=1$ is control of the FWER, and larger choices of $k$ power making more discoveries.

\subsubsection{Background and Motivation}

The fundamental challenge of high-dimensional statistical inference, characterized by settings where $p \gg n$, is that classic statistical theory no longer applies: classical asymptotic statistics treats $p$ as fixed in the relevant asymptotic thought experiment as $n$ grows large, whereas in high-dimensional settings the $p$ is growing (much) faster than $n$ in the relevant asymptotic thought experiment. As such, many classic results fail in high-dimensional settings, and new high-dimensional statistical theory must be developed---a literature that has grown popular in response to increasing interest in the statistical analysis of ``big data.''\footnote{Concretely, challenges with developing statistical methods in high-dimensions include the lack of explicit limit distributions and complex dependence structures (see, e.g., \cite{chernozhukov_high-dimensional_2023} for more).} Multiple hypothesis testing (or selective inference, more broadly) in high-dimensions has received considerable attention in recent years, with methods now available that provide asymptotic high-dimensional control of the FWER under mild conditions \citep{belloni_high-dimensional_2018}, and asymptotic high-dimensional control of the FDR under weak dependence across test statistics \citep{liu_phase_2014,belloni_high-dimensional_2018}. Though valid in high-dimensions, these methods are quite conservative, or involve assumptions about the dependence structure of the data that are implausible in many settings (including the present one), respectively.

These existing high-dimensional multiple hypothesis testing results are supported by advances in high-dimensional central limit theory. Most existing high-dimensional CLT results naturally handle the case of controlling the FWER, which is facilitated by learning the quantiles of the distribution of the \textit{maximum} coordinate of a scaled sum of independent random vectors in high-dimensions, but are invalid for studying the behavior of workhorse methods for controlling the $k$-FWER, which require learning the distribution of the \textit{$k$-th largest} coordinate of a scaled sum of independent random vectors in high-dimensions for $k>1$. Fortunately, recent work in high-dimensional probability and statistics (e.g., \cite{ding_gaussian_2025,feng_high_2026}) has developed new theoretical tools suited to understanding the behavior of exactly these objects in high-dimensions, under assumptions of varying strength, as well as related bootstrap procedures. I apply such tools to study and ultimately validate the asymptotic $k$-FWER control of the \cite{romano_control_2007} procedures in high-dimensions; I then use tools from \cite{ding_gaussian_2025} specifically to develop new, assumption-lean high-dimensional CLTs that support more general high-dimensional multiple hypothesis testing procedures, in the spirit of \cite{belloni_high-dimensional_2018}.

The following subsections state formal assumptions and results in high-dimensional central limit theory and multiple hypothesis testing. All formal results are stated as generally as possible, holding for data that is independent but not identically distributed (i.n.i.d.), which nests the setting of interest based on discovery from i.i.d. sampled unstructured data. Readers not interested in theoretical details and formal statements may skip to Section \ref{sec_auto}.

\subsubsection{Assumptions}\label{sec_assp}

Let $n \geq 3$ and $p \geq 3$. Further define
\[
S_{n,j} := \frac{1}{\sqrt{n}} \sum_{i=1}^n (X_{ij} - E[X_{ij}]) , \quad
S_n := \frac{1}{\sqrt{n}} \sum_{i=1}^n (X_i - E[X_i]).
\]
Denote $\Sigma:=E\left[S_n S_n^\mathtt{T}\right]=n^{-1} \sum_{i=1}^n E[(X_i-E[X_i])(X_i-E[X_i])^\mathtt{T}]$, which simplifies to $\Sigma = E[(X_i-E[X_i])(X_i-E[X_i])^\mathtt{T}]$ under i.i.d. data.

The following assumptions are then made.

\begin{assumption}{M}\label{ass_mom}
Let $b_1>0$ and $b_2>0$ be some constants such that $b_1 \leq b_2$. For all $i\in [n],\, j \in [p]$, assume there exists a sequence of constants $B_n>1$ (which may diverge) such that: (i) $E\left[\exp \left(\left|X_{i j}-E[X_{ij}]\right| / B_n\right)\right] \leq 2$; (ii) $b_1^2 \leq \frac{1}{n} \sum_{i=1}^n E[(X_{i j}-E[X_{ij}])^2]$; and (iii) $\frac{1}{n} \sum_{i=1}^n E[(X_{i j}-E[X_{ij}])^4] \leq B_n^2 b_2^2 $.
\end{assumption}

These are mild conditions on the tails and moments of the data $X_{ij}$, which are stated in a way that accommodates data that are i.n.i.d. Part (i) of Assumption \ref{ass_mom} simply requires that $X_{ij}$ be subexponential, or, equivalently, have an Orlicz norm in $\Psi_1$ bounded by $B_n$. The $B_n$ is indexed by $n$ to accommodate growing tail thickness within the subexponential regime as $p=p_n$ grows in the asymptotic thought experiment with i.n.i.d. data. 
Part (ii) of Assumption \ref{ass_mom} insists that the second (centered) moments of the data be bounded away from zero, appropriately stated for i.n.i.d. data. 
Part (iii) insists on bounded fourth (centered) moments in a similar fashion. Though these results are stated to handle i.n.i.d. data (in the interest of supporting applications beyond those considered in this paper), for i.i.d. data---the case considered in the applications of this framework---it is only required that, for some fixed $B < \infty$ and fixed $b_1, b_2$ with $b_1 \leq b_2$: (i) $E\left[\exp \left(\left|X_{i j}-E[X_{ij}]\right| / B\right)\right] \leq 2$; (ii) $b_1^2 \leq  E[(X_{i j}-E[X_{ij}])^2]$;
and (iii) $ E[(X_{i j}-E[X_{ij}])^4] \leq B^2 b_2^2 $. Note that, in many use cases of interest based on concept indicators, the $X_{ij}$ will be both bounded (almost surely) and studentized, making these assumptions highly plausible.

\begin{remark}[Degenerate features]\label{rmk_degen}
Some fraction of concept indicators may be degenerate at zero for any given dataset (they are ``dead'' \citep{bricken2023monosemanticity}). Assumption \ref{ass_mom} foreshadows that the high-dimensional CLTs underlying the upcoming high-dimensional multiple hypothesis testing results are only valid for understanding the $k$-th largest coordinate of the (in-population) non-degenerate coordinates of $S_{n}$, which can be labeled $S_{n}^+$. Of course, degenerate concepts in distribution $P$ are not of interest for the purposes of discovery, and so modeling the distribution of $S_{n,[k]}^+$ is precisely of scientific interest, in alignment with Assumption \ref{ass_mom}.

The researcher knows the in-sample degenerate coordinates, which is a superset of the in-population degenerate coordinates. It is recommended to drop the in-sample degenerate coordinates to form a vector $\check S_n$, and proceed with analysis based on  $\check S_n$ (noting that, by contrast, in finite samples, entries of $S_{n}^+$ may be zero, even if in the large sample limit they would not be). Whenever the $S_{n,j}^+$ for all non-degenerate $j$ are strictly positive with probability 1, then $\check S_{n,[k]} = S_{n,[k]}^+$ with probability 1 as well, and looking at the data to drop degenerate coordinates is both useful and benign. This commonly occurs: for example, when $|S_n^+|$ is of interest in two-sided testing and $X_{ij}$ has no atoms or $X_{ij}$ is binary (such as when $X_{ij}=Y_{ij}$).\footnote{So long as, technically, $nE[Y_{ij}]$ is not \textit{exactly} an integer.} The Gaussian multiplier bootstrap quantity that will be discussed later on, under two-sided testing, will also satisfy this property under the data-conditional law. Otherwise, whenever $k$ or more coordinates of $S_{n}^+$ are strictly positive with very high probability, so too is $\check S_{n,[k]} = S_{n,[k]}^+$ with very high probability, and this will almost always be the case in conventional datasets of concept embeddings (that feature many thousands of non-degenerate concepts).
\end{remark}

\begin{assumption}{R}\label{ass_rate}
Assume that $B_n^2 \log ^5(p n) = o(n)$.
\end{assumption}

This is the key rate condition needed for (uniform) Gaussian or bootstrap approximation error to go to zero asymptotically under the high-dimensional CLTs discussed in both the works of \cite{ding_gaussian_2025} and \cite{chernozhuokov_improved_2022}. Rewritten, it says that
\[
\frac{B_n^2 \log ^5(p n)}{n} = o(1)
\]
which permits $p$ growing very fast with $n$ in the asymptotic thought experiment. In fact, $p$ may be growing nearly exponentially in $n$, e.g., $p=e^{n^{1/6}}$ for fixed $B_n$. That this rate condition permits $p \gg n$ is an important positive result in high-dimensional central limit theory \citep{chernozhukov_central_2017,chernozhukov_high-dimensional_2023}.

\begin{remark}[State-of-the-art assumptions]
Both Assumptions \ref{ass_mom} and \ref{ass_rate} are the state-of-the-art in the high-dimensional central limit theory from which this framework draws \citep{chernozhuokov_improved_2022,ding_gaussian_2025}. Indeed, these conditions are the exact same as those used in both \cite{chernozhuokov_improved_2022} and \cite{ding_gaussian_2025}. Notably, these assumptions do not require non-degeneracy of $\Sigma$, another important positive result in this literature \citep{chernozhukov_high-dimensional_2023}. Under assumptions of non-degenerate $\Sigma$, improved rate conditions may be used in place of Assumption \ref{ass_rate}. 
\end{remark}


\subsubsection{High-dimensional $k$-FWER Control}\label{sec_control}

I focus on the Gaussian multiplier bootstrap as the high-dimensional bootstrap method, which is both computationally light and theoretically suited to the setting at hand. Define the Gaussian multiplier bootstrap quantity
\[
S_n^B:=\frac{1}{\sqrt{n}} \sum_{i=1}^n \xi_i\left(X_i-\bar{X}_n\right),
\]
where $\xi_i \overset{iid}{\sim} N(0,1)$ and $\bar X_n := n^{-1} \sum_{i=1}^n X_i$. For the purposes of bootstrapping, I denote the data-conditional probability measure $P^B(\cdot) := P(\cdot \mid X)$.


With these assumptions and notations in place, I now state results that certify the application of the step-down algorithms of \cite{romano_control_2007} (Algorithms 2.1 and 2.2) for controlling the $k$-FWER in large samples in high-dimensions. Recall that Algorithm 2.1 of \cite{romano_control_2007} begins by assuming that the researcher can compute appropriate critical values given by $\hat{c}_{n, K}(1-\alpha, k)$ for any $K \subset \{1, ..., p\}$, which in the present case will be given by the $1-\alpha$ quantile of $S_{n,K,[k]}^B$ under $P^B$. The algorithm proceeds by first rejecting all $H_{0,j}$ such that $T_{n,j} > \hat{c}_{n, \{1,...,p\}}(1-\alpha, k)$, and ends if there are no rejections for any $j$, in which case all hypotheses are ``accepted'' (so far, this is equivalent to the entirety of a single-step procedure). If there are less than $k$ rejections, then the algorithm halts; otherwise, heuristically, the next step of the algorithm uses information about the previously rejected hypotheses to inform a choice of $K$ that yields a $\hat{c}_{n, K}(1-\alpha, k)$ that is then compared to the test statistics for the remaining hypotheses, leading to more rejections and another step (just like the previous) or halting if there are no more rejections. Algorithms 2.1 and 2.2 only differ on how they construct $K$ after the first step, with Algorithm 2.2 offering a more computationally streamlined approach. (For completeness, both algorithms are reproduced in full in Appendix Section \ref{sec_rw21}.)

\begin{restatable}[High-dimensional $k$-FWER control for small $k$, one-sided]{theorem}{Smallk}\label{thm_smallk}
Consider the method of Algorithm 2.1 in \cite{romano_control_2007} with test statistics $T_n$ of hypotheses $\{\tilde H_{0,j}\}_{j \in [p]}$ and critical values $\hat{c}_{n, K}(1-\alpha, k)$ given by the $1-\alpha$ quantile of $S_{n,K,[k]}^B$ under $P^B$. Assume that $k$ is fixed (i.e., not growing with $n,p$). Then under Assumptions \ref{ass_mom} and \ref{ass_rate}, this method ensures: 
\begin{enumerate}[label=(\roman*)]
    \item $\limsup_{n,p\to\infty}
     k\text{-FWER}_P 
    \leq \alpha$.
    \item If $\tilde H_{0,j}$ is false and $\theta_j(P) \gg B_n \sqrt{\log p /n}$, then the probability that the step-down method rejects $\tilde H_{0,j}$ tends to 1.
    \item If
    $\min_{j \notin I(P)} \theta_j(P) \gg B_n \sqrt{\log p / n}$,
    then conclusions (i) and (ii) also hold when Algorithm 2.1 is replaced by Algorithm 2.2 of \cite{romano_control_2007}.
\end{enumerate}
\end{restatable} 

This result is the natural high-dimensional analog of Theorem 3.1 (as well as Theorem 3.3) of \cite{romano_control_2007}. Theorems 3.1 and 3.3 of \cite{romano_control_2007}---as well as all formal results in that paper---are implicitly analyzed with fixed $p$ asymptotics. Theorem \ref{thm_smallk} says that so long as $k$ is fixed (small in the relevant asymptotic sense), Algorithm 2.1 of \cite{romano_control_2007} may be used to asymptotically achieve control of the $k$-FWER at any desired level $\alpha$ even in high-dimensional settings, where $p$ can grow much faster than $n$, and moreover that Algorithm 2.1 is consistent against local alternatives that are not shrinking towards zero too quickly\footnote{Under Assumption \ref{ass_rate}, note that $B_n \sqrt{\log p /n}=o(1)$.} (i.e., intuitively, in sufficiently large samples, we are able to appropriately reject a false null so long as $\theta_j(P)=E_P[X_{ij}]$ is not very close to zero). Result (i) of Theorem \ref{thm_smallk} resembles result (i) of Theorem 3.1 of \cite{romano_control_2007}, but where now $p$ is growing as fast as allowed by Assumption \ref{ass_rate}; result (ii) of Theorem \ref{thm_smallk} resembles result (ii) of Theorem 3.1 of \cite{romano_control_2007}, but now $\theta_j(P) \gg B_n \sqrt{\log p /n}$ appears instead of $\theta_j(P) > 0$.\footnote{Formally, by $\theta_j(P) \gg B_n \sqrt{\log p /n}$ it is meant that $\theta_j(P) /(B_n \sqrt{\log p /n}) \to\infty$.} Result (iii) says that the streamlined Algorithm 2.2 of \cite{romano_control_2007} also provides high-dimensional $k$-FWER control and is consistent against local alternatives when a natural (though perhaps not always granted) separation condition holds.

The proof of this result is available in Appendix Section \ref{sec_app}, and, as previously alluded to, relies on the insight that the low-dimensional $k$-FWER controlling algorithms of \cite{romano_control_2007} can be adapted to and analyzed in the high-dimensional setting using new results in high-dimensional central limit theory for order statistics \citep{chernozhuokov_improved_2022,ding_gaussian_2025}. In the following theorem, I state the sibling procedure for high-dimensional control of the $k$-FWER in two-sided testing settings; it is the high-dimensional analog of Theorem 3.2 of \cite{romano_control_2007}.

\begin{restatable}[High-dimensional $k$-FWER control for small $k$, two-sided]{theorem}{Smallktwo}\label{thm_smallk_two}
Consider the method of Algorithm 2.1 in \cite{romano_control_2007} with test statistics $|T_n|$ of hypotheses $\{ H_{0,j}\}_{j \in [p]}$ and critical values $\hat{c}_{n, K}(1-\alpha, k)$ given by the $1-\alpha$ quantile of $|S_{n,K}^B|_{[k]}$ under $P^B$. Assume that $k$ is fixed (i.e., not growing with $n,p$). Then under Assumptions \ref{ass_mom} and \ref{ass_rate}, this method ensures: 
\begin{enumerate}[label=(\roman*)]
    \item $\limsup_{n,p\to\infty}
     k\text{-FWER}_P 
    \leq \alpha$.
    \item If $ H_{0,j}$ is false and $|\theta_j(P)| \gg B_n \sqrt{\log p / n}$, then the probability that the step-down method rejects $H_{0,j}$ tends to 1.
    \item If
    $\min_{j \notin I(P)} |\theta_j(P)| \gg B_n \sqrt{\log p / n}$,
    then conclusions (i) and (ii) also hold when Algorithm 2.1 is replaced by Algorithm 2.2 of \cite{romano_control_2007}.
\end{enumerate}
\end{restatable}

\begin{remark}[Small $k$]\label{rmk_small}
To comment explicitly on why the procedures of Theorems \ref{thm_smallk} and \ref{thm_smallk_two} are only theoretically valid for ``small'' $k$, note that the sup-norm bootstrap or Gaussian approximation error in the high-dimensional CLTs for the $k$-th largest coordinate introduced in \cite{ding_gaussian_2025} only goes to zero if
\[
\frac{k^8 B_n^2 \log ^5(p n)}{n} = o(1),
\]
meaning that $k$ must be fixed or must grow incredibly slowly. That is, $k$ needs to be quite ``small'' in theory.
\end{remark}

\begin{remark}[Choosing $k$ in practice]\label{rmk_practicalk}
The $k$ in a $k$-FWER controlling procedure is intended as an ex-ante hyperparameter choice. In practice, researchers might consider choosing a ``small'' $k$ based on the real scientific costs of investigating false discoveries, a back-of-the-envelope calculation that takes into account the researcher's prior for the count of true discoveries and/or total rejections, or calibrations based on external (independent) datasets. 

Researchers may also consider pre-specifying a small fixed grid of $k$ values, e.g., $k \in \{1, 5, 10\}$, and running the $k$-FWER controlling algorithm based on each element of the grid. Because the algorithms explored in this paper yield critical values that are monotonically decreasing in the value of $k$, researchers can simply report all the rejections for the largest $k$ in the grid, and, for each of those rejections, the first $k$ in the grid at which that rejection is made. All of this is subject to the largest $k$ in the grid still being sufficiently ``small''---see Remarks \ref{rmk_small} and \ref{rmk_placebo}. 
\end{remark}

\begin{remark}[Maximal $k$ and placebo tests]\label{rmk_placebo}
Though the high-dimensional central limit theory underlying Theorems \ref{thm_smallk} and \ref{thm_smallk_two} makes clear that $k$ must be ``small,'' this does not immediately suggest what a maximal choice of $k$ would look like in any finite sample. To evaluate if a researcher chosen $k$ is ``small'' enough, researchers may consider running tests of true global nulls in relevant settings. For example, in an experimental setting such as Example \ref{ex:rct}, researchers may generate placebo treatment assignments for units (independently of the data), and run the algorithm suggested by Theorem \ref{thm_smallk} at increasing $ k \in \{1, 2, ...\}$. The first $k$ in this sequence that yields many rejections of the null relative to that $k$---which the researcher knows to be false rejections by placebo design---suggests that $k$ has become so large that control of the $k$-FWER has begun to fail.
\end{remark}

\subsubsection{Many Approximate Means and $k$-FWER Control}\label{sec_mam}

So far, we have discussed validating the algorithms of \cite{romano_control_2007} in high-dimensions for hypotheses about means: the parameters of interest were written exactly as $\theta_j(P)=E_P[X_{ij}]$ for some interpretable $X_{ij}$, and estimators $\hat\theta_j$ were simply sample means (with $T_{n,j} = \sqrt{n}\hat\theta_j$). Though this is a common setting, researchers often want to test hypotheses for parameters which have estimators (and test statistics) that are not exactly expressible as averages of the observed data, yet ``approximately'' behave like means in large samples. Important examples include t-stats,\footnote{Multiple hypothesis testing with studentized statistics is often important for the purposes of ``balance'': that all hypothesis tests considered are similarly powered, contributing to selective error control similarly (see, e.g., \cite{romano_balanced_2010} for a more formal characterization of balance in this setting). For example, when $X_{ij} = Y_{ij} \in \{0,1\}$, more frequently occurring concept indicators are higher variance features, and we are typically interested in powering discoveries beyond just these frequently occurring concepts.} or when the parameters of interest are linear regression coefficients estimated with OLS. Thus, we also want $k$-FWER control for hypothesis testing based on ``many approximate means'' (MAM) in the parlance of \cite{belloni_high-dimensional_2018}. To achieve this, I will extend the high-dimensional CLTs introduced in \cite{ding_gaussian_2025} (which are only valid for the ``exact'' means previously studied) to the setting of MAM, and thereby unlock high-dimensional $k$-FWER control for tests based on a wider class of estimators of parameters of practical interest.

To begin, note that many estimators $\hat\theta_j$ of parameters  $\theta_j(P)$ will have the asymptotically linear representation
\begin{equation}\label{eqn_alr}
\sqrt{n}(\hat\theta_j - \theta_j(P)) = \frac{1}{\sqrt{n}}\sum_{i=1}^n \psi_{ij} + R_{n,j}, \quad \sqrt{n}(\hat\theta - \theta(P)) = \frac{1}{\sqrt{n}}\sum_{i=1}^n \psi_i + R_n,
\end{equation}
where $\hat\theta :=(\hat\theta_1, \hat\theta_2, ...,\hat\theta_p)^\mathtt{T}$, $\theta(P):=(\theta_1(P),\theta_2(P),...,\theta_p(P))^\mathtt{T}$, $\psi_i := (\psi_{i1},\psi_{i2},...,\psi_{ip})^\mathtt{T}$ is the collection of influence functions $\psi_{ij}$ at observation $i$ for $\theta_j(P)$, and the $R_{n,j}$ are small remainder terms (``linearization errors''). This setup nests the previous hypothesis testing setup based on ``exact'' means, for which $\theta_j(P) = E_P[X_{ij}]$ and 
\[\psi_{ij}=X_{ij} - E_P[X_{ij}], \quad R_n = 0.
\]
Based on this observation, going forward I will simply relabel the data $X_{ij}$ as the influence function $\psi_{ij}$, such that
\[
S_{n} = \frac{1}{\sqrt{n}}\sum_{i=1}^n \psi_{i}, \quad \tilde S_n := \sqrt{n}(\hat\theta - \theta(P)) = S_n + R_n,
\]
recalling that, by definition, influence functions are mean zero, i.e., $E[\psi_{ij}]=0$, and thus $\Sigma=n^{-1} \sum_{i=1}^n E[\psi_i \psi_i^\mathtt{T}]$. Assumption \ref{ass_mom} is treated as holding for data $\psi_{ij}$. With this new notation in place, I now state a high-dimensional CLT that accommodates the case of small linearization errors $R_n$, the so-called setting of ``many approximate means.'' 

\begin{restatable}[High-dimensional CLT for the small $k$-max coordinate, MAM]{theorem}{Hdapprox}\label{lem_hdkapprox}
Let $\tilde S_n := S_n + R_n$, and assume that $
\| R_n \|_\infty = o_P(1/\sqrt{\log (pn)})$. Further assume that $k$ is fixed (i.e., does not grow with $n,p$). If Assumptions \ref{ass_mom} and \ref{ass_rate} hold, then as $n,p\to\infty$
\[
\sup _{t \in \mathbb{R}}\left|P\left( \tilde S_{n,[k]} \leq t\right)-P\left(N(0,\Sigma)_{[k]} \leq t\right)\right| \to 0.
\]
\end{restatable}

This theorem is both a natural generalization of Lemma A.6 of \cite{ding_gaussian_2025}, which proves the special case for $R_n=0$, and an extension of Theorem 2.1 of \cite{belloni_high-dimensional_2018}, which concerns uniform approximations over (hyper)rectangular events (like the max statistic falling below $t$) for ``many approximate means.''\footnote{Note that the $k$-th largest coordinate falling below some $t$ for $k=1$ is an event where $\tilde S_n$ falls in a (hyper)rectangular set, and for $k>1$ is a \textit{union} of rectangles.} This theorem indicates that, for small $k$, in the large sample limit, the distribution of the $k$-th largest coordinate of a scaled sum of independent random vectors is (uniformly) well-approximated by the distribution of the $k$-th largest coordinate of a Gaussian random vector with variance matrix $\Sigma$ \textit{whether or not} there is a small linearization error $R_n$. In Appendix Section \ref{sec_hdclts}, I state an analogous theorem for the analogous (multiplier) bootstrap quantity under the relevant bootstrap law $P^B$.

Leveraging these high-dimensional CLTs for MAM, in Appendix Section \ref{sec_hdclts} I also state high-dimensional CLTs that explicitly allow for studentization, recognizing that studentization with estimated variances can be cast as analyzing a particular $\tilde S_n$. These results can in turn be used to prove the following high-dimensional $k$-FWER result; a two-sided version can similarly be stated, but is omitted for brevity.\footnote{It is also the case that we may safely only work on the event $\{\min _{j \leq p} \widehat{\Sigma}_{j j} \geq b_1^2 / 2\}$ in this theorem, so we incur no existence issues, as this event occurs with probability tending to 1 in the large sample limit given Assumption \ref{ass_mom}.}


\begin{restatable}[High-dimensional $k$-FWER control for small $k$, studentized, one-sided]{theorem}{Smallkstu}\label{thm_stu_kfwer}
Assume that $X_i \overset{\text{iid}}{\sim} P$. Define $\widehat{\Sigma}_{j j}:=n^{-1} \sum_{i=1}^n(X_{i j}-\bar{X}_{n, j})^2$ and $\widehat{\Lambda}:=\operatorname{diag}\{\widehat{\Sigma}_{11}, \ldots, \widehat{\Sigma}_{p p}\}$. Consider the method of Algorithm 2.1 in \cite{romano_control_2007} with test statistics $\hat\Lambda^{-1/2} T_n$ of hypotheses $\{\tilde H_{0,j}\}_{j \in [p]}$ and critical values $\hat{c}_{n, K}(1-\alpha, k)$ given by the $1-\alpha$ quantile of $(\hat\Lambda^{-1/2} S_{n}^B)_{K,[k]}$ under $P^B$. Assume that $k$ is fixed (i.e., not growing with $n,p$). Then under Assumptions \ref{ass_mom} and \ref{ass_rate}, so long as $B_n=O(1)$, this method ensures: 
\begin{enumerate}[label=(\roman*)]
    \item $\limsup_{n,p\to\infty}
     k\text{-FWER}_P 
    \leq \alpha$.
    \item If $\tilde H_{0,j}$ is false and $\theta_j(P) \gg \sqrt{\log p /n}$, then the probability that the step-down method rejects $\tilde H_{0,j}$ tends to 1.
    \item If
    $\min_{j \notin I(P)} \theta_j(P) \gg \sqrt{\log p / n}$,
    then conclusions (i) and (ii) also hold when Algorithm 2.1 is replaced by Algorithm 2.2 of \cite{romano_control_2007}.
\end{enumerate}
\end{restatable}

The evaluated influence functions $\psi_{ij}$ are typically not fully observed in the data, and instead must be estimated. This becomes important when bootstrapping, because the bootstrap must be based on quantities that we can actually compute with the observed data. The following theorem states that, when the influence functions are sufficiently well-estimated with observed data, the high-dimensional multiplier bootstrap remains effective.\footnote{Note that, even in the case of exact means, the influence function is not directly observed, because $E_P[X_{ij}]$ is not directly observed; nonetheless, the results proven earlier permit $E_P[X_{ij}]$ to be approximated with a sample average, as part of $S_n^B$. Theorem \ref{thm_hdbootinfluence} provides a more abstracted approximation condition to generalize this plug-in result.}

\begin{assumption}{A}\label{ass_approx}
Assume that 
$\max_{j \in[p]} \left(n^{-1}\sum_{i=1}^n(\hat{\psi}_{i j}-\psi_{i j})^2\right)^{1/2}=o_P(1 / \log(pn)).$ 
\end{assumption}

\begin{restatable}[High-dimensional bootstrap for the small $k$-max coordinate with estimated influence function]{theorem}{Hdbootinfluence}\label{thm_hdbootinfluence}
Define
$$\tilde S_n^B := \frac{1}{\sqrt{n}}\sum_{i=1}^n \xi_i\hat\psi_i.$$ 
Further assume that $k$ is fixed (i.e., does not grow with $n,p$). If Assumptions \ref{ass_mom}, \ref{ass_rate}, and \ref{ass_approx} hold, then as $n,p\to\infty$
\[
\sup _{t \in \mathbb{R}}\left|P^B\left( \tilde S_{n,[k]}^B \leq t \right)-P\left(N(0,\Sigma)_{[k]} \leq t\right)\right| \xrightarrow[]{P} 0.
\]
\end{restatable}

This theorem is the natural analog to Theorem 2.2 of \cite{belloni_high-dimensional_2018}. It uses the same sufficient condition for the asymptotic approximation quality of the influence function---Assumption \ref{ass_approx}---as \cite{belloni_high-dimensional_2018} (equation 2.3), which is derived from an implication of the Borell-Sudakov-Tsirelson inequality.

The following result then leverages Theorem \ref{thm_hdbootinfluence}, among others, to permit high-dimensional multiple hypothesis testing with $k$-FWER control in the MAM setting when influence functions must be estimated. Again, a two-sided result is omitted simply for brevity, and would introduce no new technical issues.

\begin{restatable}[High-dimensional $k$-FWER control for small $k$, MAM, one-sided]{theorem}{Smallkaif}\label{thm_aif_kfwer}
Consider the method of Algorithm 2.1 in \cite{romano_control_2007} with test statistics $ \sqrt{n}\hat\theta_j$ of null hypotheses $\tilde H_{0,j}: \theta_j(P) \leq 0$ for all $j \in [p]$, and critical values $\hat{c}_{n, K}(1-\alpha, k)$ given by the $1-\alpha$ quantile of $\tilde S_{n,K,[k]}^B$ under $P^B$. Assume that $k$ is fixed (i.e., not growing with $n,p$) and that for all $j \in [p]$ the asymptotically linear representation of equation (\ref{eqn_alr}) holds with $ \|R_{n}\|_\infty = o_P(1/\sqrt{\log (pn)})$. Then under Assumptions \ref{ass_mom}, \ref{ass_rate}, and \ref{ass_approx}, this method ensures: 
\begin{enumerate}[label=(\roman*)]
    \item $\limsup_{n,p\to\infty}
     k\text{-FWER}_P 
    \leq \alpha$.
    \item If $\tilde H_{0,j}$ is false and $\theta_j(P) \gg B_n\sqrt{\log p /n}$, then the probability that the step-down method rejects $\tilde H_{0,j}$ tends to 1.
    \item If
    $\min_{j \notin I(P)} \theta_j(P) \gg B_n \sqrt{\log p / n}$,
    then conclusions (i) and (ii) also hold when Algorithm 2.1 is replaced by Algorithm 2.2 of \cite{romano_control_2007}.
\end{enumerate}
\end{restatable} 
Theorem \ref{thm_aif_kfwer} embeds a practical bootstrapping procedure for estimators of parameters that are only asymptotically linear: derive the influence functions for the parameters of interest, estimate them if they are not directly observed, and then implement a multiplier bootstrap based on the data-conditional law of $\tilde S_n^B$. 

An important application of these results is the empirical setting where each $\theta_j(P)$ is a linear regression coefficient based on a separate regression of a concept indicator on some regressor(s) of interest, estimated with a corresponding OLS coefficient $\hat\theta_j$. Here, $\sqrt{n}(\hat\theta_j -\theta_j(P))$ will have only an asymptotically linear representation, and the data-evaluated influence functions must be estimated. In the following example, I walk through the derivation of the bootstrapping procedure for testing many regression coefficients.

\begin{example}[Bootstrap for many linear regression coefficients]\label{rmk_manylinear}
Consider the dataset given by $\{(Y_i,T_i,D_i)\}_{i=1}^n$, drawn i.i.d. from distribution $P$, where
$Y_i=(Y_{i1},\ldots,Y_{ip})^{\mathtt T}\in\mathbb R^p$ are concept indicators, $T_i\in\mathbb R$ is a treatment of interest, and $D_i\in\mathbb R^d$ are controls for some small (fixed) $d$ that includes a constant. The researcher is interested in making discoveries based on the coefficient on $T$ from the linear regression of $Y_j$ on $(T, D)$, denoted $\theta_j(P)$, and which is formally defined via the optimization problem
\[
(\theta_j(P),\vartheta_j(P))  := \arg\min_{b,g} E_P[(Y_j - T b - D^\mathtt{T} g)^2].
\]

By the Frisch-Waugh-Lovell theorem, one can write
\[
H_{0,j}: \theta_j(P) = \Omega(P)^{-1} E_P[\widetilde T\,\widetilde Y_j]=0
\]
where $L_P[\cdot\mid\cdot]$ denotes the population linear projection operator, and by defining $\widetilde T := T - L_P[T\mid D]$ and $\widetilde Y_j := Y_j - L_P[Y_j\mid D]$, letting $\Omega(P) := E_P[\widetilde T^2]$ for shorthand. Defining the population residual $\widetilde U_j := \widetilde Y_j - \theta_j(P) \widetilde T$, the influence function for $\theta_j(P)$ at observation $i$ is given by
\[
\psi_{ij} := \Omega(P)^{-1}\widetilde T_i \widetilde U_{ij}.
\]

For the expectation under the empirical distribution $E_n[\cdot]$ (i.e., the sample average), define feasible sample analogs as
\[
\begin{gathered}
\widehat{\gamma}_T:=E_n[D_i D_i^{\mathtt{T}}]^{-1}E_n[ D_i T_i],\qquad \widehat{\gamma}_{Y_j}:=E_n[D_i D_i^{\mathtt{T}}]^{-1}E_n[D_i Y_{i j}],\\
\widehat{\widetilde{T}}_i:=T_i-D_i^{\mathtt{T}} \widehat{\gamma}_T, \qquad
\widehat{\widetilde{Y}}_{ij}:=Y_{ij}-D_i^{\mathtt{T}} \widehat{\gamma}_{Y_j},\\
\widehat\Omega:=E_n[\widehat{\widetilde T}_i^2], \qquad
\widehat\theta_j:=\widehat\Omega^{-1}E_n[\widehat{\widetilde T}_i\,\widehat{\widetilde Y}_{ij}],\qquad
\widehat{\widetilde U}_{ij}:=\widehat{\widetilde Y}_{ij}-\widehat\theta_j\widehat{\widetilde T}_i, \qquad
\widehat\psi_{ij}:=\widehat\Omega^{-1}\widehat{\widetilde T}_i \widehat{\widetilde U}_{ij},
\end{gathered}
\]
noting that $\hat\theta_j$ is just the usual OLS coefficient. The asymptotically linear representation of each of the $j\in [p]$ regression coefficients can then be written as
\[
\sqrt{n}(\hat\theta_j - \theta_j(P)) = \frac{1}{\sqrt{n}}\sum_{i=1}^n \psi_{ij} + R_{n,j}.
\]
The appropriate law for facilitating $k$-FWER control via Theorem \ref{thm_aif_kfwer} is then given by
\[
\tilde S_n^B := \frac{1}{\sqrt{n}}\sum_{i=1}^n \xi_i\hat\psi_i = \widehat\Omega^{-1}\frac{1}{\sqrt{n}}\sum_{i=1}^n \xi_i \widehat{\widetilde T}_i \widehat{\widetilde U}_{i}
\]
under $P^B$, where $\widehat{\widetilde U}_{i}$ stacks all $\widehat{\widetilde U}_{ij}$ as a vector. That $\hat\psi_i$ is a sufficiently good approximation of $\psi_i$ and that $R_n$ is sufficiently small (in large samples) is proven in Propositions \ref{thm_hdreg} and \ref{thm_hdlin}, as stated in Appendix Section \ref{sec_manylinprop}, under certain sufficient conditions (e.g., all data is subgaussian tailed).
\end{example}

In Appendix Section \ref{sec_simulate}, I perform simulations calibrated to
real SAE data to assess the finite sample size control properties of the procedures based on
Theorems \ref{thm_stu_kfwer} and \ref{thm_aif_kfwer} as $k$, $n$, and $p$
vary. Consistent with the large sample theory, these simulations show
$k$-FWER control even when $p \gg n$, in all but the smallest samples.

\subsubsection{High-dimensional FDX Control}\label{sec_fdx}

In practice, researchers may also be interested in control of false discovery exceedance (FDX). Control of FDX at $(\gamma,\alpha)$ is defined as
\[
P(\operatorname{FDP} > \gamma ) \leq \alpha
\]
where $\operatorname{FDP} := F/(R \vee 1)$ and where $F$ is the number of false rejections made by some multiple testing procedure and $R$ is the total number of rejections made. Control of the FDX is sometimes preferred to control of the FDR (control that grants $E[\operatorname{FDP}] \leq \alpha$), as control of the FDR can still permit large deviations of the FDP from its mean with high probability, whereas with FDX control the probability of such large deviations is bounded \citep{genovese_exceedance_2006}. 

In practice, successful control of the FDX is usually implemented as a ``meta-procedure'' built on top of a valid $k$-FWER controlling procedure \citep{janson_familywise_2016}. This is also true of the large sample FDX controlling Algorithm 4.1 of \cite{romano_control_2007}, which is a wrapper around their $k$-FWER controlling Algorithms 2.1 or 2.2. However, the asymptotic argument that justifies using Algorithm 4.1 of \cite{romano_control_2007} would require that, in the high-dimensional asymptotic thought experiment, $k$ also grows\footnote{See the proof of Theorem 4.1 in \cite{romano_control_2007}.}---in other words, $k$ could be ``large,'' which is contradicted by the assumptions made to achieve \textit{high-dimensional} $k$-FWER control under the assumption-lean results of \cite{ding_gaussian_2025}.\footnote{In theory, one could make very strong signal sparsity-style assumptions to try to recover high-dimensional $k$-FWER control with growing (``large'') $k$, but such assumptions would be difficult to parse the validity of in practice.}

Alternatively, however, FDX control may still be achieved in high-dimensions with an ``augmentation'' approach, stated in the following two results.

\begin{restatable}[High-dimensional FDX control via augmentation]{theorem}{fdxcontrol}\label{thm_fdxcontrol}
Consider running the high-dimensional $k$-FWER controlling procedure of Theorem \ref{thm_aif_kfwer} at level $\alpha$, and then augmenting the output set of rejections $\mathcal{R}$ by the set of null hypotheses with next $q$ largest test statistics to form the new rejection set $\mathcal{R}_{\operatorname{aug}}$, where $0 \leq q \leq p-|\mathcal{R}|$ and $q$ may be data-dependent. Then this augmented procedure based on $\mathcal{R}_{\operatorname{aug}}$, under the assumptions of Theorem \ref{thm_aif_kfwer}, achieves
\[
\limsup_{n,p\to\infty}P\left\{\operatorname{FDP} > \frac{k-1+q}{(\left|\mathcal{R}\right|+q) \vee 1} \right\} \leq \alpha.
\]
\end{restatable}
\begin{proof}
Let $F$ be the number of false rejections made by the $k$-FWER controlling procedure at level $\alpha$, and $\mathcal{R}$ be the set of rejected nulls made by this procedure alone. Let $\tilde F$ be the number of false rejections made by the $k$-FWER controlling procedure at level $\alpha$ augmented by rejecting the previously unrejected null hypotheses with the next $q$ largest test statistics.

Consider the event $\{F \leq k-1\}$. On this ``good'' event, no matter how we choose $q$, $\tilde F \leq k-1+q$, which is the pessimistic case that all additionally rejected hypotheses are true nulls. That is
\[
\{F \leq k-1\} \subseteq \{\tilde F \leq k-1+q\} \subseteq \left\{\frac{\tilde F}{(\left|\mathcal{R}\right|+q) \vee 1} \leq \frac{k-1+q}{(\left|\mathcal{R}\right|+q) \vee 1} \right\} := \left\{\operatorname{FDP} \leq \frac{k-1+q}{(\left|\mathcal{R}\right|+q) \vee 1} \right\}.
\]
Then by Theorem \ref{thm_aif_kfwer} (i), we have that
\[
\limsup_{n,p\to\infty}P\left\{\operatorname{FDP} > \frac{k-1+q}{(\left|\mathcal{R}\right|+q) \vee 1} \right\} \leq \limsup_{n,p\to\infty}P \{F \geq k\} = \limsup_{n,p\to\infty} k\operatorname{-FWER}_P \leq \alpha,
\]
completing the proof.
\end{proof}

This is a natural extension of Theorem 3 of \cite{genovese_exceedance_2006} to the high-dimensional setting.\footnote{Note that this theorem can be generalized to consider rejecting \textit{any} disjoint set of previously unrejected hypotheses of size $q$, though the above implementation is resistant to data-mining and is plausibly the most ex-ante powerful choice.} Because $q$ can be chosen in any (even data-dependent) way, the corollary stated below follows immediately, which shows  FDX control can be achieved at a specific researcher chosen $(\gamma, \alpha)$ based on an underlying (high-dimensional) $k$-FWER controlling procedure.

\begin{restatable}[High-dimensional FDX control at $\gamma$ via augmentation]{corollary}{fdxcontrolgamma}\label{cor_fdxcontrolgamma}
Consider running the augmented high-dimensional $k$-FWER controlling procedure of Theorem \ref{thm_fdxcontrol}, but where if $\gamma|\mathcal{R}| < (k-1)$, set $\mathcal{R}_{\operatorname{aug}} = \emptyset$, and otherwise choose $q$ to be the largest integer such that
\[
\frac{k-1+q}{(|\mathcal{R}|+q)\vee 1} \leq \gamma.
\]
Then the augmented procedure based on $\mathcal{R}_{\operatorname{aug}}$ achieves
$$
\limsup _{n, p \rightarrow \infty} P\{\mathrm{FDP}>\gamma\} \leq \alpha .
$$
\end{restatable}
\begin{proof}
If $\mathcal{R}_{\operatorname{aug}}$ is the empty set, then FDX control at $(\gamma, \alpha)$ holds vacuously. When $\mathcal{R}_{\operatorname{aug}}$ is not the empty set, by the result of Theorem \ref{thm_fdxcontrol}, the conclusion is immediate.
\end{proof}

Note that, in practice, solving out the above, we may set $q$ to simply be the smaller of $p-|\mathcal{R}|$ and $\left\lfloor(\gamma |\mathcal{R}|-(k-1))/(1-\gamma)\right\rfloor$. Further note that (asymptotic) FDX control at $(\gamma,\alpha)$ implies (asymptotic) FDR control at level $\gamma(1-\alpha)+\alpha$ (e.g., if $\alpha = 0.05$ and $\gamma = 0.1$, then FDR is controlled at 0.145); see, e.g., \cite{romano_stepup_2006} for more on the connections between FDX and FDR control.

\subsection{Local Automatic Interpretation}\label{sec_auto}

By applying the statistical procedures discussed in Section \ref{sec_mht}, the researcher obtains a selected set of rejected null hypotheses, or ``discoveries.'' Concretely, these discoveries are recorded as a set of coordinates $\hat J \subseteq [p]$ indexing the rejected nulls. 

Though the machine learning theory that motivates sparse dictionary learning supposes that each hypothesis $\hat \jmath \in \hat J$ has a human-interpretable meaning, sparse dictionary learning methods themselves do not generate \textit{natural language descriptions} of what each human-interpretable meaning may be---analogous to the way in which ``topics'' discovered by LDA topic models do not come with natural language descriptions \citep{blei2003latent}. I now state a set of procedures for both automatically generating and evaluating the quality of natural language feature descriptions, emphasizing the importance of statistical inference with respect to the actual data distribution of interest, $P$.

\begin{remark}[Why do we need automatic interpretation?]
Human judgment alone could be deployed as a means of giving natural language descriptions to concept indicators, e.g., the researcher could inspect many of the text examples that activate (highly) for a given SAE feature in order to arrive at a plausible English language description for the concept signaled by that feature. However, such a process suffers from (1) problems related to motivated reasoning and data snooping (e.g., the researcher could simply come up with a concept description that is reverse engineered from a conclusion they wanted to find in the data); and (2) problems related to the scale of interpretation and associated interpretation effort/cost (e.g., there may be many hundreds of concepts discovered, all requiring interpretation; interpreting all $p$ concepts would require even more labor). Automatic interpretation pipelines solve both of these problems. They automate the process of creating and scoring natural language descriptions of concept indicators by leveraging a dedicated ``explainer LLM'' to reason about interpretations inexpensively at scale.\footnote{The terminology ``explainer LLM'' is chosen to be consistent with related works in the AI interpretability literature, e.g., \cite{jiang_interpretable_2025}.} Of course, the researcher can look directly at the highly activating unstructured data points to develop their own interpretations, in parallel.
\end{remark}

\subsubsection{Local Automatic Generation}

I adopt the canonical autointerp generation strategy pioneered in \cite{bricken2023monosemanticity} for this framework: the researcher passes text samples from their unstructured dataset that activate highly on a particular feature to an explainer LLM---annotating these exemplar texts with information about where and how strongly the feature activates---and asks for an English description of the possible meaning of the feature.\footnote{This exact generation strategy, which only looks at ``top activating'' examples, is also used in and argued for in, e.g., \cite{choi2024automatic}.} 

By drawing on text samples from the researcher's unstructured dataset of interest, this generation strategy is \textit{local}---it is tailored to the distribution $P$. Though off-the-shelf feature descriptions are sometimes available for pretrained sparse dictionary learning models, these descriptions were necessarily learned on a different distribution of data (e.g., some large and opaque swath of the internet). As such, these descriptions are liable to be of lower quality than those based on the researcher's own dataset. The specific prompts implementing this strategy for the empirical applications considered later are detailed in Appendix Section \ref{app_autointerp}.

\begin{examplecont}{ex:rct}
The null hypothesis of no average treatment effect is rejected for concept 4. Using the raw SAE feature activations (necessary to compute in the process of evaluating the $h$ function), the researcher is able to rank text responses from experimental participants by largest activations for concept 4 on any token (``max pooling''). The researcher takes the top 10 highest activating text examples and annotates each of them to indicate the highest activating token(s) in each response. These exemplar responses are then passed to an explainer LLM that reasons about and then outputs the concise English phrase ``discussion of politics'' to describe concept 4. 
 \end{examplecont}

\subsubsection{Local Automatic Evaluation}\label{sec_autoeval}

Autointerp pipelines ideally not only specify how to generate descriptions for concept indicators, but also \textit{evaluate} the quality of these descriptions. I now consider a new statistical formalization of a mainstream autointerp evaluation technique called ``detection scoring.'' 

The fundamental principle of detection scoring is that ``an interpretation should serve as a binary classifier distinguishing activating from nonactivating contexts'' \citep{paulo_automatically_2024}. In other words: description of a concept can be treated like a classification problem, one where the natural language description itself (harnessed in an LLM prompt) is the classifier, the input is a specific text example, and the label to be predicted is derived from the ground truth of the presence or absence of a SAE feature's activation in that text example (which is always known to the researcher using their $h$ function). Under this framing, binary classifier performance metrics like accuracy, precision, and recall characterize the quality of the description.

To formalize this autointerp strategy, consider a ``held-out'' unstructured dataset of size $m$ drawn independently from the same distribution of interest $P$. (This held-out dataset may be created, for example, by partitioning an original unstructured dataset drawn i.i.d. from $P$ into two disjoint datasets via indices chosen at random, independently of all the data.) Define the set of indices of unstructured data points in the main sample as $\mathcal{I}^\textrm{estim}$, and those in the held-out dataset as $\mathcal{I}^\textrm{eval}$.\footnote{In empirical settings featuring a randomized experiment, for the sake of preserving power, it may be useful for researchers to create the held-out dataset split from the larger treatment arm, if one exists.}

Using the data points with indices in $\mathcal{I}^\textrm{estim}$, the researcher extracts exemplar unstructured data points to be used to create an autointerp explanation of concept $j$, denoted $\hat\eta_j$. For all $i \in \mathcal{I}^\textrm{eval}$, I define the random variable $\hat Y_{ij} \in \{0,1\}$ as
\[
\hat Y_{ij} := \texttt{CLS}(Z_i, \hat \eta_j),
\]
where the function $\texttt{CLS}(\cdot)$ is a fixed choice of LLM and classification prompt fed both an unstructured data point $Z_i$ and the learned autointerp explanation $\hat\eta_j$. The $\hat Y_{ij}$ is the classifier output, a prediction of whether or not the unstructured data point $Z_i$ contains a concept with description $\hat \eta_j$.\footnote{Note that, even if a sample splitting implementation of this autointerp evaluation strategy is pursued, cross-fitting is inappropriate, because there is little reason to think, given the sample sizes available to researchers, that $\hat\eta_j$ would be stable across folds. As such, when sample splitting is employed, an uneven split is recommended, e.g., only 15\% of the data is used for the held-out evaluation sample.}

For the purposes of evaluating autointerp descriptions, I define the canonical ``detection score'' estimand as
\[
\theta_j^\textrm{acc}(\hat\eta_j;P) := E_P\left[S_{ij} \mid \hat\eta_j \right] = P(\{Y_{ij} = \hat Y_{ij}\} \mid \hat\eta_j) = P\left(\{Y_{ij} = \texttt{CLS}(Z_i, \hat\eta_j)\} \mid \hat\eta_j\right)
\]
where
$S_{ij} := \mathbf{1}\{Y_{ij} = \hat Y_{ij}\}$, i.e., $\theta_j^\textrm{acc}(\hat\eta_j;P)$ is the accuracy of the classifier. Conditioning on $\hat\eta_j$ in the estimand of interest should be thought of as evaluating a fixed autointerp description. Importantly, as the notation emphasizes, the $\theta_j^\textrm{acc}(\hat\eta_j; P)$ tells the researcher exactly the quality of the description $\hat\eta_j$ for the population of interest $P$, unlike other autointerp evaluation methods for which the underlying distribution that implicitly defines the evaluation score is highly opaque (c.f., simulation scoring \citep{bills2023language}). In this way, evaluation is also \textit{local}.

A plug-in estimator of the autointerp detection score estimand is easily constructed; its unbiasedness, consistency, and asymptotic normality under the $\hat\eta_j$-conditional law are established in the following proposition.

\begin{restatable}[Conditional inference on detection score]{proposition}{Ascoreest}\label{prop_ascore}
Let $m:=|\mathcal{I}^\textrm{eval}|$. Define the detection score estimator for concept indicator $j$ as
\[
\hat \theta_j^\textrm{acc} := \frac{1}{m} \sum_{i \in \mathcal{I}^\textrm{eval}} S_{ij} ,\quad S_{ij}= \mathbf{1}\{Y_{ij} = \hat Y_{ij}\} = \mathbf{1}\{Y_{ij} = \operatorname{\texttt{CLS}}(Z_i, \hat \eta_j)\}.
\]
Assume there exists a constant $\delta \in (0,1/2)$ such that $P(S_{ij}=1
\mid \hat \eta_j) \in [\delta, 1-\delta]$ almost surely.
Then we have that $E_P\left[\hat \theta_j^\textrm{acc} \mid \hat\eta_j \right]=\theta_j^\textrm{acc}(\hat\eta_j;P)$ almost surely, and that as $m \to \infty$
\[
\sup_{t \in \mathbb{R}} \left| P\left(\sqrt{m}( \hat \theta_j^\textrm{acc} - \theta_j^\textrm{acc}(\hat\eta_j;P)) \leq t \mid \hat\eta_j \right) - P(N(0, \operatorname{Var}(S_{ij} \mid \hat\eta_j)) \leq t \mid \hat\eta_j) \right| \xrightarrow[]{a.s.} 0.
\]
\end{restatable}
It is immediate from Proposition \ref{prop_ascore} then that the empirical variance estimator $\hat V_j^\text{acc}:= m^{-1}\sum_{i \in \mathcal{I}^\text{eval}} (S_{ij}-m^{-1}\sum_i S_{i j} )^2$ permits construction of valid large sample frequentist confidence intervals for $\hat\theta_j^\text{acc}$ under the $\hat\eta_j$-conditional law of the data.

\begin{examplecont}{ex:rct}
The null hypothesis of no average treatment effect is rejected for concept 4, and the autointerp generation strategy implemented by the researcher describes concept 4 as activating on ``discussion of politics.'' To evaluate the quality of this description, the researcher runs an LLM-based classifier over a held-out dataset of experimental participant text responses. The LLM classifier is given the prompt ``Does the following text example indicate discussion of politics? Return 1 or 0 for yes or no:'' followed by a text response from the held-out dataset. The researcher then uses the output of the LLM classifier to compute the detection score, which measures the accuracy of this autointerp description-based classifier at predicting the presence of concept 4 in the held-out dataset.
 \end{examplecont}

\begin{remark}[Oracle detection scores]
Towards better understanding the information communicated to the researcher by the detection score, it is useful to consider the function $\texttt{CLS}^*(\cdot)$, defined implicitly as
\[
\bar\theta_j^\textrm{acc}(\hat\eta_j;P):=\sup_{f \in \mathcal{F}} P\left(\{Y_{ij} = f(Z_i, \hat\eta_j)\} \mid \hat\eta_j\right) = P\left(\{Y_{ij} = \texttt{CLS}^*(Z_i, \hat\eta_j)\} \mid \hat\eta_j\right),
\]
where $\mathcal{F}$ is the space of all viable classification prompts and LLMs able to execute them. That is, $\texttt{CLS}^*(\cdot)$ is the ``oracle'' detection scoring model setup, which when evaluated produces the highest interpretative accuracy score possible under a given description $\hat\eta_j$ under distribution $P$. Arguably, $\bar\theta_j^\textrm{acc}(\hat\eta_j;P)$ is more informative to the researcher than $\theta_j^\textrm{acc}(\hat\eta_j;P)$ about the quality of $\hat\eta_j$, as $\bar\theta_j^\textrm{acc}(\hat\eta_j;P)$ purely quantifies the quality of $\hat\eta_j$, without the wedge introduced by classification failures stemming from $\texttt{CLS}(\cdot)$ itself. Because any feasible choice of $\texttt{CLS}(\cdot)$ yields
\[
P\left(\{Y_{ij} = \texttt{CLS}(Z_i, \hat\eta_j)\} \mid \hat\eta_j\right) \leq P\left(\{Y_{ij} = \texttt{CLS}^*(Z_i, \hat\eta_j)\} \mid \hat\eta_j\right)
\]
one can reasonably interpret $\hat \theta_j^\textrm{acc}$ as a conservative (under)estimate of the more informative quantity $\bar\theta_j^\textrm{acc}(\hat\eta_j;P)$. 
\end{remark}

Additional detection score estimands can also be defined with respect to either the precision or recall of the autointerp classifier:
\begin{align*}
    \theta_j^\textrm{prec}(\hat\eta_j) &:=P\left(Y_{ij}=1 \mid \hat Y_{ij}=1,\hat\eta_j\right) = \frac{E[\mathbf{1}\{Y_{ij}=1\}\mathbf{1}\{\hat Y_{ij}=1\}\mid \hat\eta_j]}{E[\mathbf{1}\{\hat Y_{ij}=1\}\mid \hat\eta_j]},
    \\
    \theta_j^\textrm{rec}(\hat\eta_j) &:=P\left(\hat Y_{ij}=1 \mid Y_{ij}=1,\hat\eta_j \right) = \frac{E[\mathbf{1}\{Y_{ij}=1\}\mathbf{1}\{\hat Y_{ij}=1\}\mid \hat\eta_j]}{E[\mathbf{1}\{Y_{ij}=1\}\mid \hat\eta_j]}.
\end{align*}
As above, these estimands may be consistently estimated with appropriate asymptotically valid confidence intervals under the $\hat\eta_j$-conditional law using the natural plug-in estimators and the delta method. That said, for the empirical applications in the following section, I focus on estimating the  detection score based on accuracy, following the common practice in the relevant literature \citep{paulo_automatically_2024}.

Further, recall that we are not typically interested in autointerp evaluation for a fixed $j$, but instead for a particular data-driven choice $\hat\jmath \in \hat J$ suggested by a high-dimensional selective inference procedure.\footnote{Autointerpretation typically incurs non-negligible compute costs, and, as such, researchers will often only be interested in forming meaningful autointerpretations for the subset of discovered concepts.} However, so long as discovery is also performed using the main estimation sample indexed by $\mathcal{I}^\textrm{estim}$, the researcher may treat each $\hat \jmath$ as fixed for the purposes of evaluation inference. Formally, for all $j \in [p]$,
\begin{align*}
     E_P\left[S_{i\hat \jmath} \mid \hat\eta_{\hat \jmath}, \hat \jmath = j \right] = E_P\left[S_{ij} \mid \hat\eta_{j}, \hat \jmath = j \right] =E_P\left[S_{ij} \mid \hat\eta_{j}\right]
\end{align*}
where the last equality follows because, under the proposed independently held-out data scheme,
\[
[\mathbf{1}\{Y_{ij} = \texttt{CLS}(Z_i, \hat\eta_j)\} \ind \hat \jmath] \mid \hat\eta_j.
\]

\begin{remark}[Inference on the best interpretation]
If a researcher is interested in inference on the best performing autointerp description across multiple generation prompts or strategies---say among a set $\{\hat\eta_{j,1}, \hat\eta_{j,2}, \dots\}$ for any $j$ of interest---then the evaluation performance should be estimated with appropriate ``inference on winners'' corrections such as those discussed in \cite{andrews_inference_2024}. 
\end{remark}

\begin{remark}[Evaluation of human-generated descriptions]
It is worth noting that this autointerp evaluation strategy can also serve as a general-purpose description evaluation technology, e.g., for carefully crafted human descriptions of concept indicators in settings where there is a small number of discoveries. The scientific benefit of evaluation in such a setting is also clear: even if there is concern that human-made descriptions are reverse-engineered to support a presupposed conclusion, any description (human made or machine made, motivated or not) that is not supported by the data will be surfaced through the evaluation scores. In this way, description evaluation can be seen as even more critical when automatic methods are not employed by the researcher to generate descriptions.
\end{remark}

\subsection{Pseudocode and Code}\label{sec_pseudo}

In Appendix Section \ref{sec_pseudocode}, I state a simple version of the entire framework for interpretable discovery from unstructured data as a single pseudocode algorithm. I focus on an implementation with text data, sample splitting, two-sided testing, single-step hypothesis testing (i.e., just the first step of Algorithm 2.1 of \cite{romano_control_2007}), and autointerp evaluation via accuracy-based detection scoring. Full implementation code is also available: see the link in Appendix Section \ref{sec_replication}.

\section{Empirical Applications}\label{sec_empirical}

As an illustration of the framework proposed in this paper, I now reanalyze several recent works in empirical economics related to discovery from unstructured data.

\subsection{Reanalyzing \cite{andre_narratives_2026}}\label{sec_narratives}

An inspiration for Example \ref{ex:rct}, \cite{andre_narratives_2026} study household narratives about the surge in inflation circa 2021-2022 in the United States. In particular, \cite{andre_narratives_2026} elicit household narratives about inflation using surveys with open-ended text responses. Among the many analyses they pursue, \cite{andre_narratives_2026} conduct several survey experiments in order to measure the causal effects of being exposed to ``expert'' narratives on respondents' personal narrative beliefs about  inflation. 

In their primary experiment, \cite{andre_narratives_2026} recruit a representative sample of hundreds of households and randomly expose one-third of respondents to a written ``expert'' narrative attributing rising inflation throughout 2021-2022 to ``pent-up demand'' and forced savings during the COVID-19 pandemic (with another third exposed to a different expert narrative on energy crises, and the last third acting as a control). In a follow-up survey the next day, the authors then elicit respondents' own inflation narrative through open text---one of the core outcomes of the experiment.\footnote{The authors' rationale for collecting open-ended text responses is explicitly to ``offer a lens into people’s spontaneous thoughts,'' as open-ended elicitation ``leaves individuals’ answers unrestricted and does not prime them on any particular issue.''} In order to quantitatively measure the narratives that appear in the open text, the authors hand-code 16 distinct ``narrative factor'' indicators, crafted per a discretionary, manual process.\footnote{See Table 1 of \cite{andre_narratives_2026} for a list of these hand-coded ``narrative factors.''} The primary finding of this experiment is that the presence of the ``pent-up demand'' hand-coded concept increases by 37.8 percentage points in the ``pent-up demand'' treatment group relative to the control---signaling a very high ``elasticity to the provision of new information''---and the hand-coded concept for ``energy crisis'' decreases by 7.9 p.p. in the pent-up demand treatment group relative to the control---a ``crowding out'' phenomenon that the authors do not interpret further.

Do we \textit{automatically} replicate these findings with the proposed framework? Do we expand on them? To investigate, I implement a concept embedding based on the Temporal Sparse Autoencoder (T-SAE) \citep{bhalla2026temporal}, itself based on layer 12 of a Gemma family LLM \citep{lieberum_gemma_2024}. I implement an autointerp pipeline detailed in Section \ref{app_autointerp} using as explainer LLMs Anthropic's Claude Sonnet (High) for generation and Google's lightweight Gemma 4 31B for evaluation.\footnote{Gemma 4 is open-source, and as such evaluations of generated descriptions remain readily reproducible, conditional on generated descriptions.} A small held-out evaluation dataset is created by sample splitting, using just 15\% of the original experimental dataset in \cite{andre_narratives_2026}. 

The hypotheses tested are of the form discussed in Example \ref{rmk_manylinear},
\[
H_{0,j}: \theta_j(P) = E_P[\widetilde T_i^2]^{-1} E_P[\widetilde T_i\,\widetilde Y_{ij}]=0,
\]
where $\widetilde Y_{ij}$ is an indicator for concept $j$ in the concept embedding and $\widetilde T_i$ is the ``pent-up demand'' treatment indicator, both residualized with the same controls used in the regression specification that produced the primary results of the original study (i.e., a linear regression specification with controls as documented in Table 4 of \cite{andre_narratives_2026}). This yields a dataset with $n=758$ and $p=5420$ (non-degenerate features---see Remark \ref{rmk_degen}). 

I begin by considering 5-FWER controlled multiple testing with $\alpha = 0.05$ based on the result of Theorem \ref{thm_aif_kfwer}, and make 162 interpretable discoveries---that is, the null of an average treatment effect equal to zero is rejected for 162 concepts encoded by the concept embedding. In Figure \ref{fig:narratives1}, I display the 10 largest (by $|\hat\theta_j|$) of the \textit{well-interpreted} discoveries, which I define as discoveries with a point estimate for the detection accuracy score greater than 0.75.\footnote{This is a benchmark derived from the median Claude detection score reported in \cite{paulo_automatically_2024} for autointerp on large-scale, open-source text corpora.} Furthermore, for presentational purposes, I suppress plotting discoveries with autointerp descriptions that are highly semantically similar to many others---the total number of similarly interpreted discoveries being indicated in a parenthetical when there is more than one.\footnote{In particular, I embed all autointerp generated descriptions with \texttt{all-mpnet-base-v2}, and perform agglomerative clustering with a threshold given by the 99th percentile cosine similarity of these dense embeddings; for each cluster, only the concept with the highest detection accuracy score is plotted.} Such a plot can be thought of as a concise summary of the outputs of the entirety of the discovery framework. The next 10 largest well-interpreted discoveries are plotted in Figure \ref{fig:narratives2}, and the remaining well-interpreted discoveries can be viewed in Appendix Section \ref{sec_addlplots}.

The discoveries made with the proposed framework automatically replicate the finding that the ``pent-up demand'' narrative increases prominently in the treatment group. Concept 1174, described as ``words like `consumer,' `demand,' `spending,' '' increases by nearly the same amount as the original, manually coded pent-up demand feature. Many other concepts (in the top 20 alone) reference demand, spending, or pandemic recovery, echoing highly specific elements of the treatment narrative\footnote{See Appendix Section H.3 of \cite{andre_narratives_2026} for the exact statement.} and making it appear plausible that respondents are acting in part due to some social desirability bias.\footnote{As speculated, but not confirmed, in footnote 21 of \cite{andre_narratives_2026}.} 

The original finding that a narrative about an ``energy crisis'' decreases in the treatment group is qualified by running the proposed framework: no specific SAE concepts regarding energy decrease in presence in the treatment group, but many concepts about the Russia-Ukraine War (1885, 2854, 8502, 1676) and international or global actors (1198, 532, 4906) do. As the authors note, ``At the time of the experiment, the energy crisis had just been exacerbated by Russia’s invasion of Ukraine.'' This suggests that the actual effect of the treatment narrative may have been to diminish respondents' beliefs in the importance of international forces in driving inflation (a ``global-to-local'' narrative shift).

Lastly, there are several discoveries made by the proposed framework that seem to characterize a new category of treatment effects missed by the original hand-coding scheme, which one might label ``narrative amplification.'' These include large treatment effects on concepts 5613 (``words expressing greed, excess desire, or entitlement'') and 1754 (``phrases about Biden taking/becoming office''), both of which are common inflation narratives held by non-experts.\footnote{See Section 4.1 of \cite{andre_narratives_2026}.} This is to say that a U.S.-specific treatment narrative may have reinforced some respondents' plausibly pre-existing narratives centered around domestic U.S. political factors.

\begin{figure}[ht]
\centering
\includegraphics[width=1.0\textwidth]{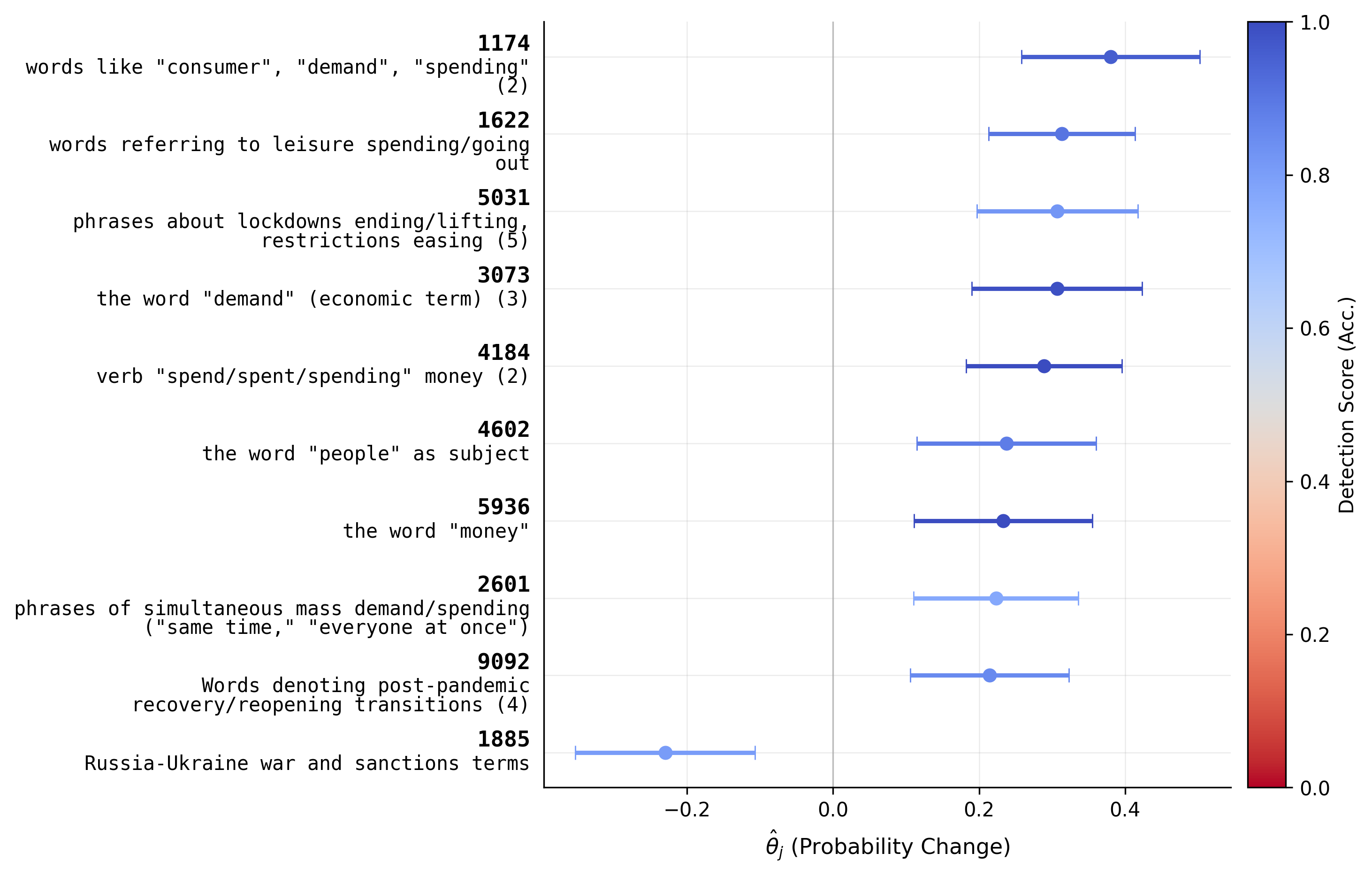}
\caption{Well-Interpreted Discoveries for the First Narrative Provision Experiment of \cite{andre_narratives_2026}, Rank 1-10.}
\label{fig:narratives1}
\end{figure}

\begin{figure}[ht]
\centering
\includegraphics[width=1.0\textwidth]{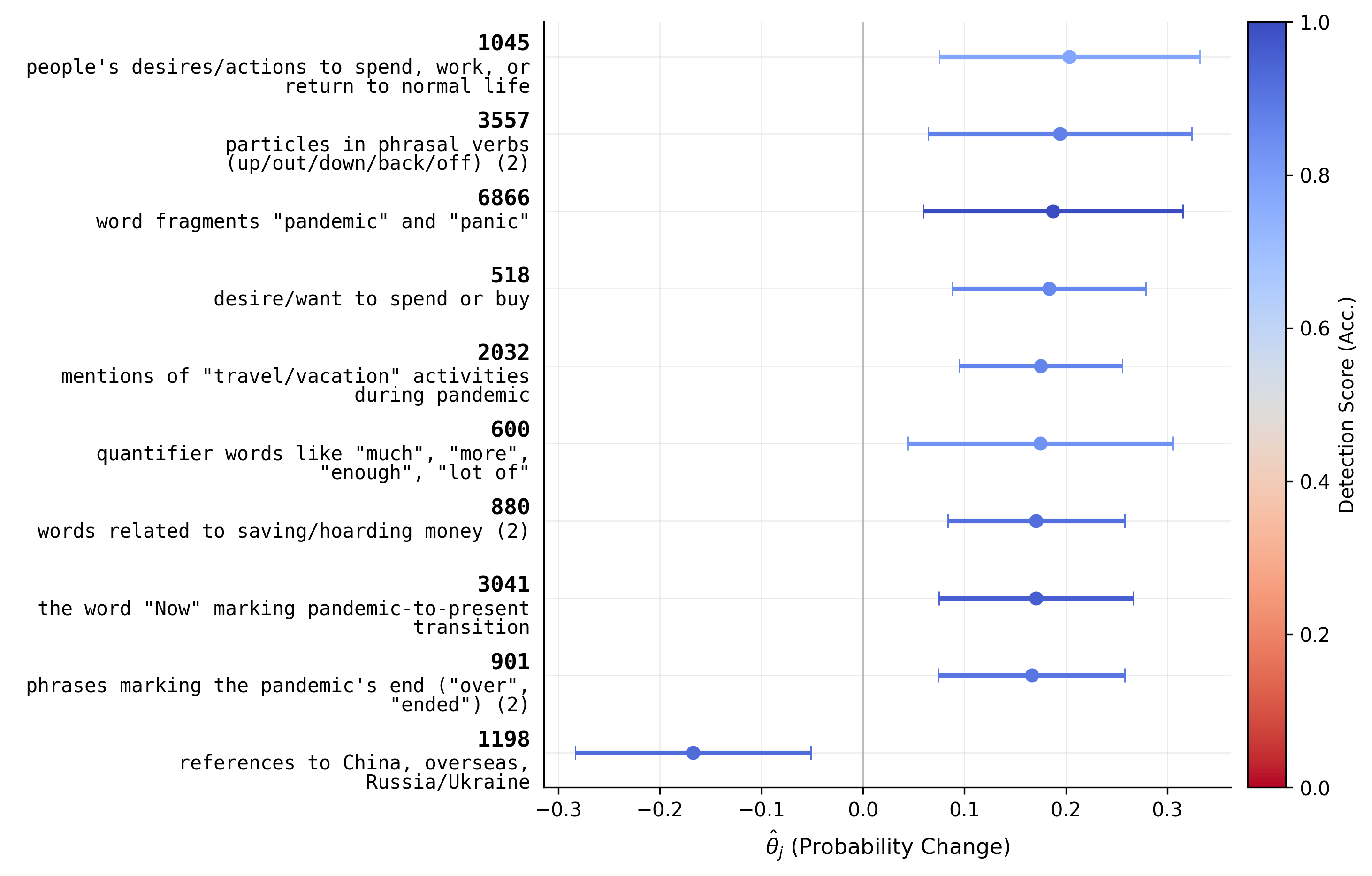}
\caption{Well-Interpreted Discoveries for the First Narrative Provision Experiment of \cite{andre_narratives_2026}, Rank 11-20.}
\label{fig:narratives2}
\end{figure}

If (1-)FWER controlled multiple testing is applied instead, at the same $\alpha = 0.05$, the framework only yields 65\% of these discoveries, in particular missing some of the ``narrative amplification'' discoveries discussed above.\footnote{That is, 106 discoveries are made with FWER control, which is only slightly better than a Bonferroni correction, which results in 97 discoveries being made.} 

A complementary survey experiment run by \cite{andre_narratives_2026} (detailed in Appendix G.1) studies to what extent respondents adopt narratives they encounter in ``endogenously chosen'' (self-selected) news media content. In this experiment, respondents in the treatment group are instructed to spend 5 minutes reading a newspaper article of their choice about U.S. inflation, whereas respondents in the control group are asked to do the same for an article on a topic unrelated to inflation (specifically, ``tourist attractions in Miami''), with both groups being provided a small monetary incentive to summarize the articles. Then, in a follow-up survey, respondents' narratives for the recent surge in inflation are elicited with an open-text response. The result is a nearly 7 p.p. increase in the presence of a hand-coded ``residual demand'' factor and nearly 10 p.p. increases in the presence of both the hand-coded ``pandemic'' narrative factor and ``residual supply'' narrative factor.\footnote{See Appendix Figure B.10 of \cite{andre_narratives_2026}.}

For this experiment, using the same concept embedding and explainer LLM setup as discussed above with a 10\% held-out evaluation set yields a dataset for analysis with $n=673$ and $p=5548$. Testing analogously formed null hypotheses, implementing a high-dimensional (1-)FWER controlled multiple testing procedure with $\alpha = 0.05$ only returns 3 discoveries, none of which yield any deep insight (e.g., concept 770 is ``the word `a/an' before `result,''' and none of these discoveries have descriptions that pass the formally described criterion for being well-interpreted).\footnote{Bonferroni correction applied in this setting returns the same 3 discoveries.} However, implementing 5-FWER control at $\alpha = 0.05$ upturns 17 discoveries. That is, by large sample guarantee of Theorem \ref{thm_aif_kfwer}, the probability of making 5 or more false rejections with this procedure is at most 5\% across repeated samples, and so, informally speaking, one would have high (frequentist) confidence that \textit{at least} 13 of these 17 discoveries are genuine. Well-interpreted discoveries are plotted in Figure \ref{fig:narrativesmedia}, and suggest that endogenous uptake of news-based inflation narratives significantly increases respondent narratives involving notions of supply and demand (an automatic replication of the discoveries from the original analysis) as well as consumer spending, specific economics jargon (suggesting parroting of the media content), and labor supply shortages (previously undocumented).

\begin{figure}[ht]
\centering
\includegraphics[width=1.0\textwidth]{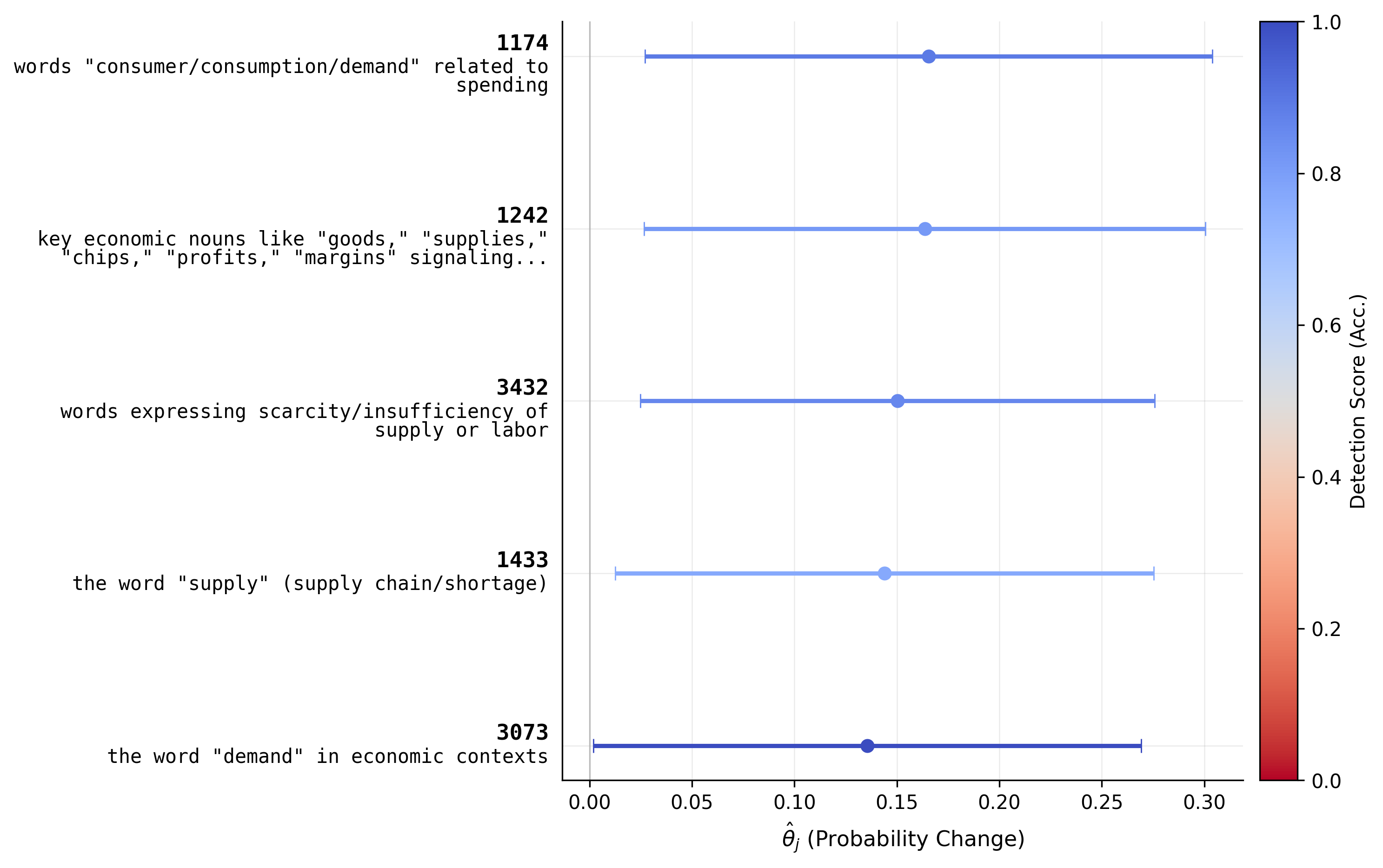}
\caption{Well-Interpreted Discoveries for the Media Narrative Provision Experiment of \cite{andre_narratives_2026}.}
\label{fig:narrativesmedia}
\end{figure}

All of this analysis was performed automatically, in a handful of minutes, on a single A100 GPU in Google Colab (accessed for free, as a student or university affiliate). The API calls to Claude cost only cents, and could also be replaced with open-source model generations for free. Though this analysis was not preregistered, it is possible to cheaply and automatically replicate these results, and assess their sensitivity to different autointerpretation strategies, concept embedding implementations, and testing hyperparameters. Notably, the space of possible concepts was not pruned in any way prior to this analysis. Discoveries and their interpretations were made without human intervention, with no room for data snooping.

\subsection{Reanalyzing \cite{bursztyn_justifying_2023}}\label{sec_dissent}

Another inspiration for Example \ref{ex:rct}, \cite{bursztyn_justifying_2023} study how the provision of ``social cover'' affects willingness to publicly dissent to socially stigmatized causes, and the perception of this dissent. As a key application of their theory, \cite{bursztyn_justifying_2023} run an  information treatment experiment (Experiment 2) in which participants are told they have been matched with another (fictional) respondent that chose to join a campaign to \textit{oppose defunding the police} (a plausible expression of dissent in liberal American politics at the time the experiment was conducted), and are then shown a tweet that the matched respondent is said to have posted. This tweet has been randomized to either provide or not provide some ``social cover'' for the fictional matched respondent's plausibly contentious statement, namely whether or not the tweet indicated that the matched respondent joined the campaign before reading an evidence-based article in support of it (the no social cover condition) or after (the social cover condition).\footnote{See Online Appendix Figure B.3 of \cite{bursztyn_justifying_2023} for a schematic of the experimental design.}

A key outcome collected by this experiment is the participant's open-ended text response to the question ``Why do you think your matched respondent chose to join the campaign to oppose defunding the police?'' This open-ended text response is meant to capture the causal effect of social cover on the perception of dissent. 
Towards this goal, \cite{bursztyn_justifying_2023} compute a Pearson's $\chi^2$ statistic for all phrases of up to three words per \cite{gentzkow_what_2010}, which they use as an index to rank the phrases that are most differentially expressed in each condition's open-ended responses.\footnote{Specifically, it is a Pearson's $\chi^2$ statistic for a null hypothesis that the propensity to use a given phrase is equal across conditions, per \cite{gentzkow_what_2010}. I interpret these statistics as simply forming an index, however, because the results of these hypothesis tests are not reported, and no multiple hypothesis testing corrections are implemented.} 
The richness of the results of this analysis is, naturally, hindered by the coarseness of the featurization of the text as three word phrases, as well as the fact that there are no obvious estimands or inferential guarantees. 
The only qualitative conclusion gleaned by \cite{bursztyn_justifying_2023} from this quantitative exercise is: ``we find that respondents in the Cover condition are more likely to use phrases related to the article or the associated evidence—for example, `article,' `read,' `convincing,' or `increase in crime.'\,''\footnote{The Online Appendix Table B.11 of \cite{bursztyn_justifying_2023} contains the top ten characteristic phrases in each condition based on the $\chi^2$ index.}

Do we make more (and more interpretable) discoveries with the framework proposed in this paper? To investigate, I apply the same concept embedding and explainer LLM setup as in Section \ref{sec_narratives}. A small held-out evaluation dataset is created by sample splitting, using just 15\% of the original dataset in \cite{bursztyn_justifying_2023}. 

The hypotheses of interest are, for all $j \in [p]$, 
\[
H_{0,j}: \theta_j(P) = E[Y_{ij}(1) - Y_{ij}(0)]=  E\left[\frac{W_{i} - \pi}{\pi(1-\pi)}Y_{ij}\right] = 0,
\]
for $\pi = 0.5$, and where the event $\{W_i = 1\}$ indicates the social cover condition. This yields a dataset for analysis where $n = 879$ and $p = 4697$. 

I again consider whether or not controlling the (1-)FWER would be too conservative for the purposes of meaningful discovery in this setting. To explore this, I set $k=1$ and $\alpha=0.05$,\footnote{Equivalently, I implement the high-dimensional FWER controlling procedure of \cite{belloni_high-dimensional_2018}.} and I obtain just 5 rejections of nulls from high-dimensional multiple testing. These include: a 14.8 percentage point causal effect on the presence of concept 1786; a 14.8 p.p. causal effect on the presence of concept 8565; an 8.9 p.p. causal effect on concept 3869; a 12.5 p.p. causal effect on concept 4335; and a 11.1 p.p. causal effect on concept 1503. The automatic interpretations of these concepts corroborate the original qualitative findings of \cite{bursztyn_justifying_2023}: concept 1786 is described as activating on mentions of ``references to the Washington Post as a cited source''; concept 8565 on ``the word `article' referring to a cited source'';  concept 3869 on ``reading or encountering an article with information''; concept 4335 on ``references to articles or editorials as persuasive sources''; and concept 1503 on ``verbs and phrases describing what an article says or implies.'' However, only 5 significant causal effects are discovered, which again yield little new insight relative to the original analysis---even if the discovery and interpretation of these concepts was transparent, principled, and automatic.

\begin{figure}[ht]
\centering
\includegraphics[width=1.0\textwidth]{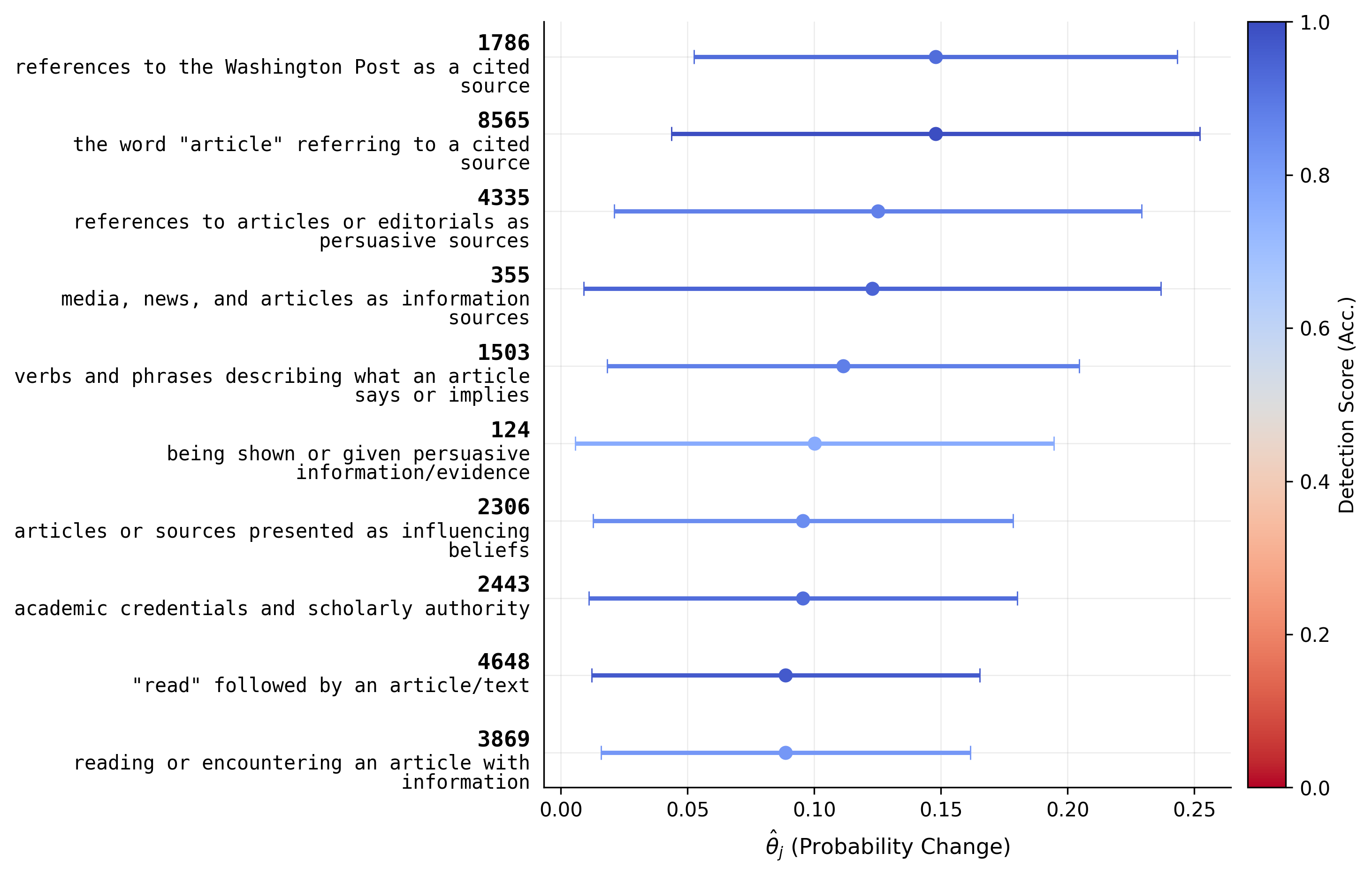}
\caption{Well-Interpreted Discoveries for Experiment 2 of \cite{bursztyn_justifying_2023}, Rank 1-10.}
\label{fig:dissent1}
\end{figure}

I now apply the discovery framework setting $k=5$, still with $\alpha = 0.05$, per the method of Theorem \ref{thm_stu_kfwer}.\footnote{In practice, placebo tests reveal that high-dimensional $k$-FWER control \textit{in this particular example} seems reliable even at fairly large $k$, i.e., $k$ can be over 1700 before encountering more rejections than should be permitted.} With 5-FWER control, 15 rejections of the null are obtained---that is, 15 discoveries are made. Again, by large sample guarantee of Theorem \ref{thm_stu_kfwer}, the probability of making 5 or more false rejections with this procedure is at most 5\% across repeated samples, and so, informally, in the frequentist sense, one would have high confidence that \textit{at least} 11 of these 15 discoveries are genuine.\footnote{To get a better sense for which discoveries are plausibly true, informally, one might consider re-running the algorithm of Theorem \ref{thm_stu_kfwer} at increasing values of $k$, and noting at which $k$ each hypothesis enters the selected set.}

In Figure \ref{fig:dissent1}, I again display the 10 largest well-interpreted discoveries. The remaining well-interpreted discoveries can be viewed in Appendix Section \ref{sec_addlplots}. In this plot, the previous concepts discovered with (1-)FWER control appear, as well as others absent from both that analysis and the original analysis in \cite{bursztyn_justifying_2023}. 

Though there is some conceptual redundancy in the discoveries---which we may have anticipated, based on earlier discussion related to, e.g., \cite{bhalla_sparse_2026}---in totality these discoveries recover many findings in \cite{bursztyn_justifying_2023} that otherwise required full additional experiments to learn. 
For example, concept 2443 responds to ``academic credentials and scholarly authority'' and increases by nearly 10 percentage points in the treatment group, indicating that---exactly as \cite{bursztyn_justifying_2023} posit, and run a separate experiment to further probe---the social cover mechanism plausibly relies on the perceived credibility of the article cited. Concept 7927 (the 13th largest well-interpreted discovery, viewable in Appendix Figure \ref{fig:dissent2}) responds to ``being persuaded or convinced by arguments,'' and its probability of activation is increased by just over 7 percentage points in the treatment group. This is direct validation of the authors' manipulation leading to ``the (implicit) justification that they joined the campaign because they were persuaded by the article’s claims''---exactly as intended, but something that \cite{bursztyn_justifying_2023} otherwise needed to run Auxiliary Experiment 4 to confirm.

\subsection{Reanalyzing \cite{noy_experimental_2023}}

An inspiration for Example \ref{ex:hci}, \cite{noy_experimental_2023} are interested in studying the causal effects of AI assistance on task-level productivity. Specifically, they run an RCT where midlevel professional writing tasks are assigned to hundreds of college-educated professionals, and access to ChatGPT (GPT 3.5) is randomly provided to half of the participants. In practice, in the ChatGPT-exposed treatment arm, there is 80\% adoption of the AI tool for completing the writing assignment. The headline findings of the original experiment (Figure 1) are that average time spent on the writing task decreases by 40\%, and a human-assessed quality ``grade'' (scored from 1-7) increases by 18\%.

Perceived quality of the writing increased, but why? Can we discover interpretable causal effects of AI assistance on writing using the written outputs themselves? To investigate this, I apply the discovery framework to a subset of the experimental sample consisting of HR professionals and managers, all of whom were given the same writing task: drafting a company-wide email on a sensitive topic. Again, I use the same concept embedding and explainer LLM setup as in Section \ref{sec_dissent} and a held-out evaluation sample consisting of 20\% of the original (email task) dataset. Task outputs are truncated to a maximum length of 1500 characters. 

The hypotheses tested are again of the form discussed in Example \ref{rmk_manylinear},
\[
H_{0,j}: \theta_j(P) = E_P[\widetilde T_i^2]^{-1} E_P[\widetilde T_i\,\widetilde Y_{ij}]=0,
\]
where $\widetilde Y_{ij}$ is an indicator for concept $j$ in the concept embedding and $\widetilde T_i$ is the treatment indicator, both residualized with the same controls used in the participant-level regression specification that produced the primary results of the original study (i.e., the specification S.1 in the supplemental appendix of \cite{noy_experimental_2023}). This yields a dataset with $n=138$ and $p=5435$ non-degenerate features. I again apply 5-FWER controlled multiple testing with $\alpha = 0.05$, based on the result of Theorem \ref{thm_aif_kfwer}, and make 128 interpretable discoveries.\footnote{With (1-)FWER control, only 46 discoveries are made.} In Figure \ref{fig:chatgpt1}, I again plot the 10 largest well-interpreted discoveries, and in Figure \ref{fig:chatgpt2} the next 10 largest, with the remaining well-interpreted discoveries viewable in Appendix Section \ref{sec_addlplots}. 

The discoveries made suggest that ChatGPT assistance has large causal effects on the inclusion and exclusion of specific phrases and words. Concepts 4179, 1060, and 3320 pick up on GenAI causing emails to use the specific phrase ``I hope this email finds you well.'' The causal effect on concept 5073 indicates the removal of ``Good...'' as an email salutation. A large causal effect on concept 3644 indicates the words ``finds'' and ``found'' go up in treatment; similarly for the phrase ``want to'' (concept 2861) and ``we believe'' (concept 4234) and ``concerned'' (concept 5208) and the transitional phrases ``in addition'' or ``additionally'' (concept 4465). Broadly, these findings corroborate a growing literature in human-computer interaction documenting that AI assistance on tasks homogenizes outputs, narrowing creative diversity \citep{jo_incentives_2026}.

Concepts 4299, 2689, 9037, 7200, and 1811 all increase significantly in the treatment group, and map to usage of language indicating understanding, (re)assurance, and collaboration---that is, positively valenced, empathetic language. Other discoveries, such as the negative causal effects on concepts 2147, 908, 6462, 4029, and 2431 suggest that ChatGPT exposure reduces rhetorical/linguistic/grammatical informality, e.g., reducing sentence-ending exclamation points and informal use of words like ``so,'' ``get,'' and ``go.'' Collectively, these insights hint at possible mechanisms for the perceived increase in email quality.

\begin{figure}[ht]
\centering
\includegraphics[width=1.0\textwidth]{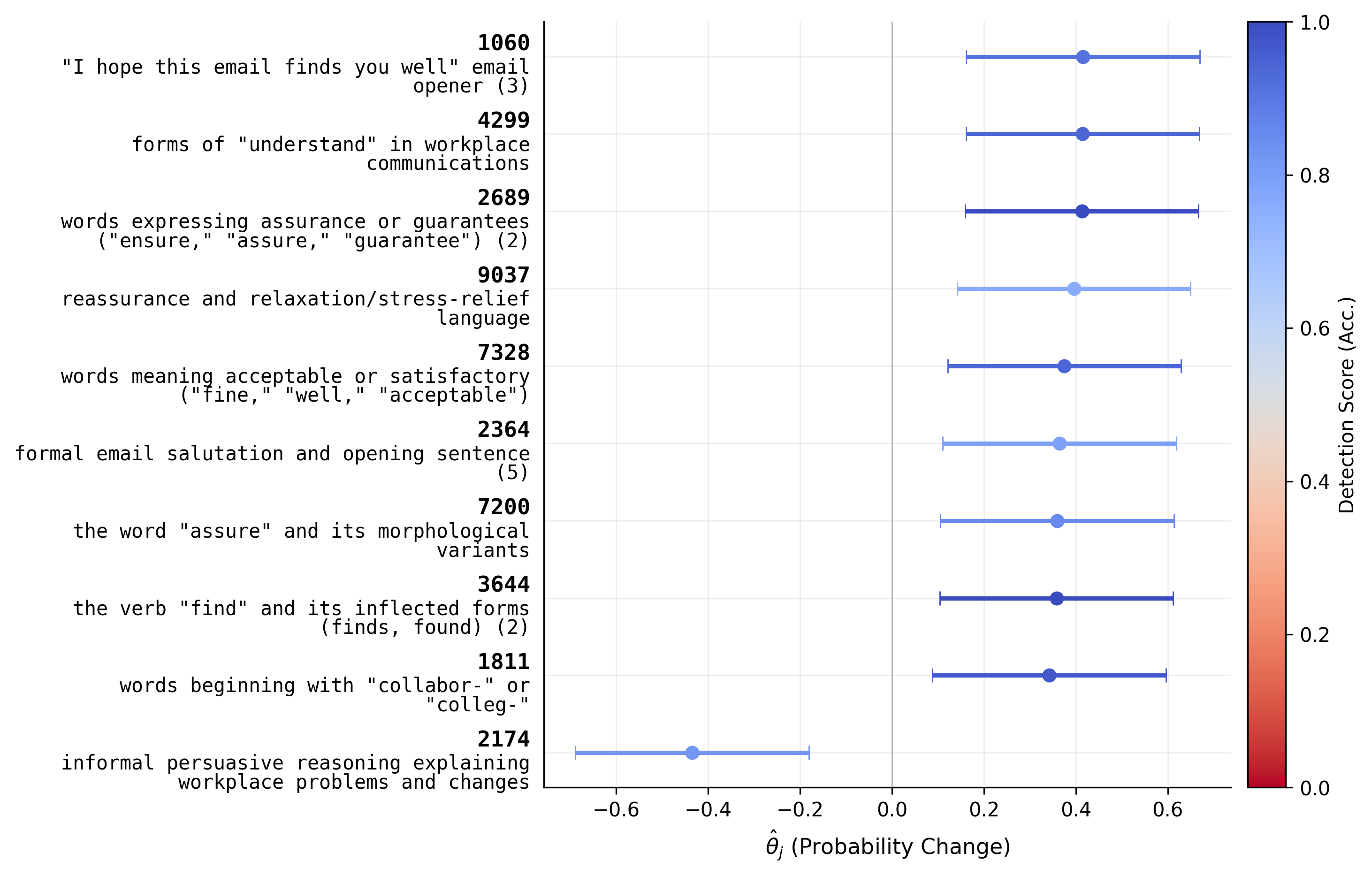}
\caption{Well-Interpreted Discoveries for \cite{noy_experimental_2023}, Rank 1-10.}
\label{fig:chatgpt1}
\end{figure}

\begin{figure}[ht]
\centering
\includegraphics[width=1.0\textwidth]{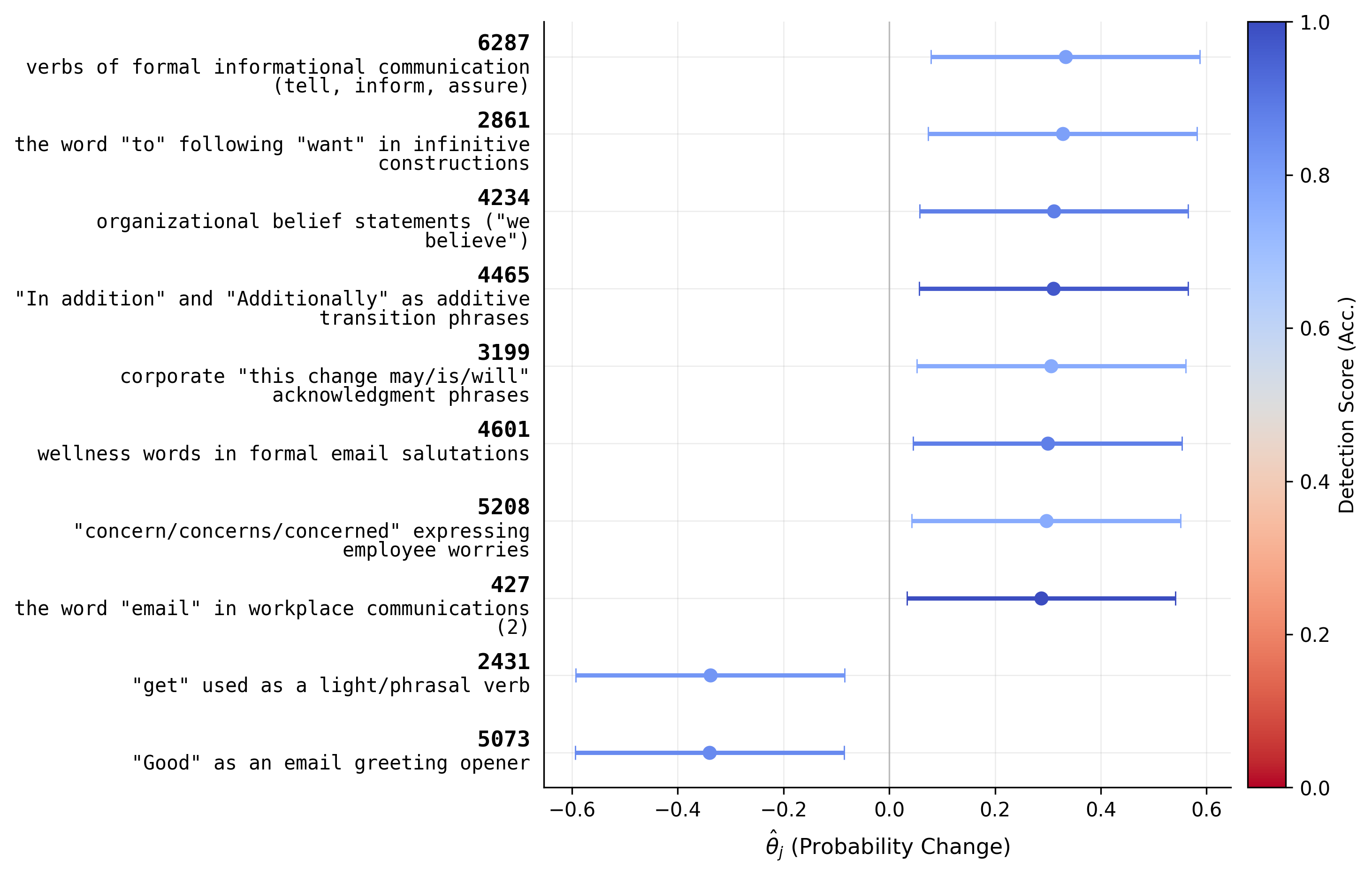}
\caption{Well-Interpreted Discoveries for \cite{noy_experimental_2023}, Rank 11-20.}
\label{fig:chatgpt2}
\end{figure}


\section{Conclusion}

Literatures across empirical economics and econometrics have increasingly suggested the importance of discovery from high-dimensional or unstructured data (e.g., \cite{ludwig_machine-learning_2017,ludwig_machine_2024,chernozhukov_fisherschultz_2025,haaland_understanding_2024,dell_measurement_2026}). The framework proposed in this paper shows that newly validated statistical procedures for high-dimensional multiple testing---when combined with measurement techniques based on the latest innovations in AI interpretability---can facilitate such discovery with higher practicality and  fidelity than previously possible.

It is worth noting that problems surrounding post-selection inference, researcher degrees of freedom, and motivated data mining have been taken very seriously in other literatures, such as the analysis of RCTs in economics, where these concerns motivated the creation of preanalysis plans (PAPs) for experiments. One analogy for the proposed framework, then, is that it provides similar value to a PAP for the setting of discovery from unstructured data. Just like a PAP, the proposed framework can be accompanied by contingent, complementary analyses of the unstructured dataset, so long as these analyses are evaluated by readers with the appropriate skepticism \citep{banerjee_praise_2020}. However, unlike writing a comprehensive PAP for an experiment, the proposed framework places very little burden on the researcher: the same default template for analysis (the concepts measured by an SAE) can be reused and shared across many diverse empirical settings.\footnote{When the unstructured data being analyzed literally comes from an experiment, preregistration of the proposed framework in an actual PAP is also easily declared (it simply involves stipulating a particular choice of explainer LLM and SAE under the default proposed concept embedding strategy).}

Having a well-defined and intuitive quantitative method for evaluating concept description quality is important in the context of the proposed framework, as it is these descriptions on which the scientific conclusions of discovery ultimately rest. By formulating autointerp evaluation as a statistical inference problem, the importance of evaluation with respect to the actual population of interest $P$ becomes salient (``locality''); the need for assembling an independent hold-out dataset (e.g., via sample splitting) becomes clear; the role for confidence intervals becomes legible by appeal to sampling uncertainty (c.f., \cite{miller_adding_2024}); and connections to the literature on post-selection inference are unlocked, such that settings evaluating a multiplicity of interpretations can be disciplined (c.f., \cite{andrews_inference_2024}).

It is important to reiterate that the proposed framework is compatible with many possible high-dimensional measurement strategies. The use of sparse dictionary learning models, and SAEs in particular, is viewed as a powerful \textit{default} recommendation for creating high-dimensional concept embeddings: even if prone to some noise and redundancy, it is both rich and broad as a measurement technology, as well as cost-effective. That said, domain-tailored alternative choices of high-dimensional concept embedding may be more appealing depending on the use case of interest.  An important future direction of work might also involve engineering a new ``social science default'' concept embedding that is tailored to a large space of highly social science relevant concepts.\footnote{I thank Jenny S. Wang for this suggestion.}


The proposed framework is most naturally viewed as one tool of many in the empirical researcher's toolkit for making discoveries from unstructured data. In practice, applying this framework alongside others could produce complementary scientific insights, or otherwise serve to enrich a researcher's understanding of their unstructured dataset.

\clearpage

\bibliography{ref.bib}

\clearpage

\section{Appendix}

\subsection{Replication}\label{sec_replication}

The results from the empirical applications section of this paper can be replicated using the open-source replication data from \cite{andre_narratives_2026}, \cite{bursztyn_justifying_2023}, \cite{stantcheva_why_2024}, and \cite{noy_experimental_2023}, respectively, in addition to the GitHub repository:  \url{https://github.com/jscarlson/hdmht-discovery}. 

\subsection{Reanalyzing \cite{stantcheva_why_2024}}

An inspiration for Example \ref{ex:desc}, \cite{stantcheva_why_2024} conducts surveys with representative samples of the United States population investigating attitudes towards inflation. These surveys include open-ended text responses to a variety of questions, with the goal of discovering attitudes and opinions that the researcher should not (for concern about priming respondents) or could not (for lack of imagination) manually pre-specify.\footnote{See, e.g., \cite{haaland_understanding_2024} for more discussion of the benefits of open-ended survey questions for understanding economic behavior.}


An important open-ended prompt that \cite{stantcheva_why_2024} solicits an answer to is ``High inflation is caused by...'' This  prompt is specifically highlighted by the Brookings Institution paper summary of \cite{stantcheva_why_2024}, signaling its policy relevance.\footnote{See \url{https://www.brookings.edu/articles/why-do-we-dislike-inflation/}.} In the original analysis, \cite{stantcheva_why_2024} codes open-ended responses with a discretionary keyword-based method introduced in \cite{ferrario_eliciting_2022}.\footnote{The method suggests that the researcher create a list of topics and associated keywords that denote their presence via anything from ``manual to semi-supervised or unsupervised'' means.} The result of applying this analytic technique to the open-ended responses to the above prompt yields Figure 3 in \cite{stantcheva_why_2024}, which finds, among other things, that mentions of ``Biden and the administration,'' ``Greed,'' ``Monetary policy,'' ``Fiscal policy,'' ``War and foreign policy,'' ``Demand vs supply,'' and ``Supply-side mechanisms (other than input prices)'' appear in more than 5\% of all responses. Discussion of ``input prices,'' ``energy prices,'' ``demand-side mechanisms,'' ``people earning higher incomes,'' ``government debt,'' and ``COVID-19'' appear in more than 2\% of responses.

Do we make the same discoveries automatically if we apply the proposed framework instead? Do we make more and new interpretable discoveries? To find out, I use the same concept embedding and explainer LLM setup as in Section \ref{sec_dissent}, but reduce the space of SAE features by half, filtering out all features that, in the corpus on which the SAE was trained, had greater than median empirical token activation frequency. This choice of dimensionality reduction is meant to screen out features that activate in many texts across domains, e.g., features related to common grammatical features of text, which was handled implicitly in the previous causal analysis by virtue of differencing. (Ideally, such a dimensionality reduction choice would be preregistered, to prevent cherry-picking by filtering, but could also easily be analyzed for sensitivity.) Again, I split the sample randomly to form a held-out evaluation dataset consisting of 15\% of the original dataset analyzed in \cite{stantcheva_why_2024}.

I now test hypotheses $H_{0,j}: \theta_j(P) =  E[Y_{ij}] = P(Y_{ij}=1)= 0$; naturally, in this setting where probabilities are being tested, of actual scientific interest is the dual of this testing procedure, i.e., the point estimates for the probability of discussion along with generalized simultaneous confidence intervals. This yields a dataset with $n=428$ and $p=1173$ non-degenerate features. I now apply 5-FWER control with $\alpha = 0.05$ based on the result in Theorem \ref{thm_stu_kfwer}, and 186 interpretable discoveries are made. In Figure \ref{fig:infl1}, as in Figure \ref{fig:dissent1}, I plot the 10 largest well-interpreted discoveries; in Figure \ref{fig:infl2} I plot the next 10 largest. The remaining well-interpreted discoveries can be similarly viewed in Appendix Section \ref{sec_addlplots}. 

Notably, among the full 186 well-interpreted discoveries made, many of the features coded from the original analysis in \cite{stantcheva_why_2024} are recovered without any human discretion or intervention: concept 2493 is ``blaming government or Biden for inflation''; concept 2189 is ```Greed' as a concept or word''; concept 3551 is ``Federal Reserve and monetary policy concepts''; concept 3029 is ``government fiscal mismanagement or overspending''; concept 1009 is ``excess demand outpacing supply'';  concept 3121 is ```supply chain' and physical goods/commodities''; concept 2105 is ``fuel, oil, and gas prices''; concept 453 is ```consumer' and `demand' in economic contexts''; concept 3672 is ``wages, labor costs, and worker pay increases''; concept 949 is ``national debt and government borrowing''; and concept 96 is ``global and international scale references.''\footnote{Concept 1138 corresponds to ``economic crisis, recession, and COVID-related economic causes''; it has a detection score point estimate of 0.723 and so just misses the cutoff for ``well-interpreted.''}

Many additional discovered concepts add nuance to these broad categories, e.g., for greed, concept 3498 is ``corporate greed and capitalist wealth,'' or for blaming the Biden administration, concept 697 is ``misplaced priorities or focus directed elsewhere.'' Furthermore, many of these well-interpreted discoveries contribute to some entirely new insights. Concept 221 is described as ``first-person expressions of uncertainty or ignorance'' and appears in over \textit{one-third} of the responses---a strong signal that is not coded in the original analysis, and what one might call an important ``unknown unknown.'' Discovered concepts 20, 925, 323, 574, 802, and 3904, among others, are all characterized by discussion of uncertainty or hedging language, further reinforcing the interpretation that many Americans are simply unsure of (or confused about) the causes of inflation, as opposed to subscribers of a particular pet theory. Concept 415 is ``commas and conjunctions separating listed causes or factors,'' which, among others, indicates that many surveyed have multiple factors in mind.


\begin{figure}[ht]
\centering
\includegraphics[width=1.0\textwidth]{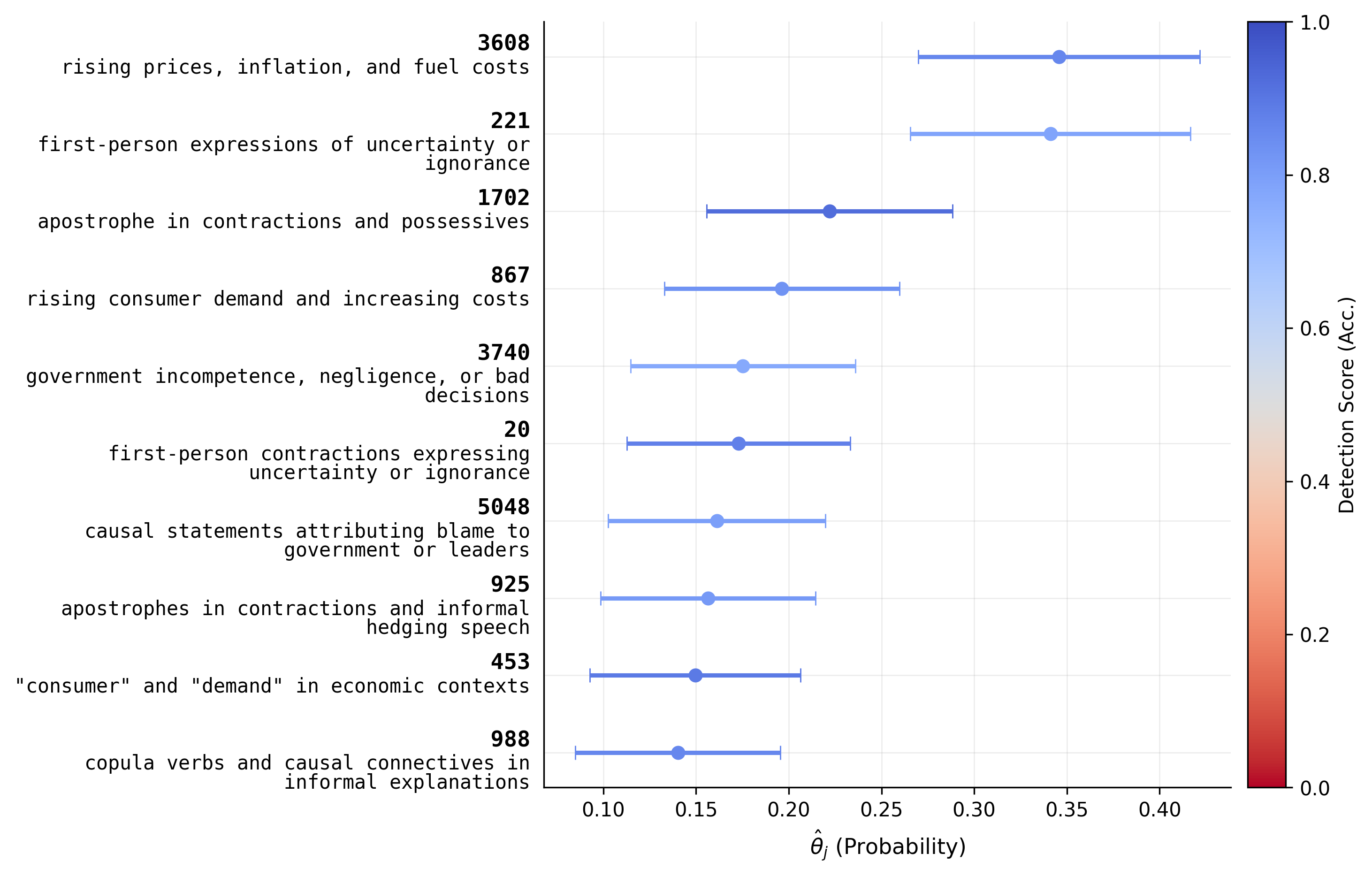}
\caption{Well-Interpreted Discoveries for Responses to Prompt ``High inflation is caused by...'' in \cite{stantcheva_why_2024}, Rank 1-10.}
\label{fig:infl1}
\end{figure}

\begin{figure}[ht]
\centering
\includegraphics[width=1.0\textwidth]{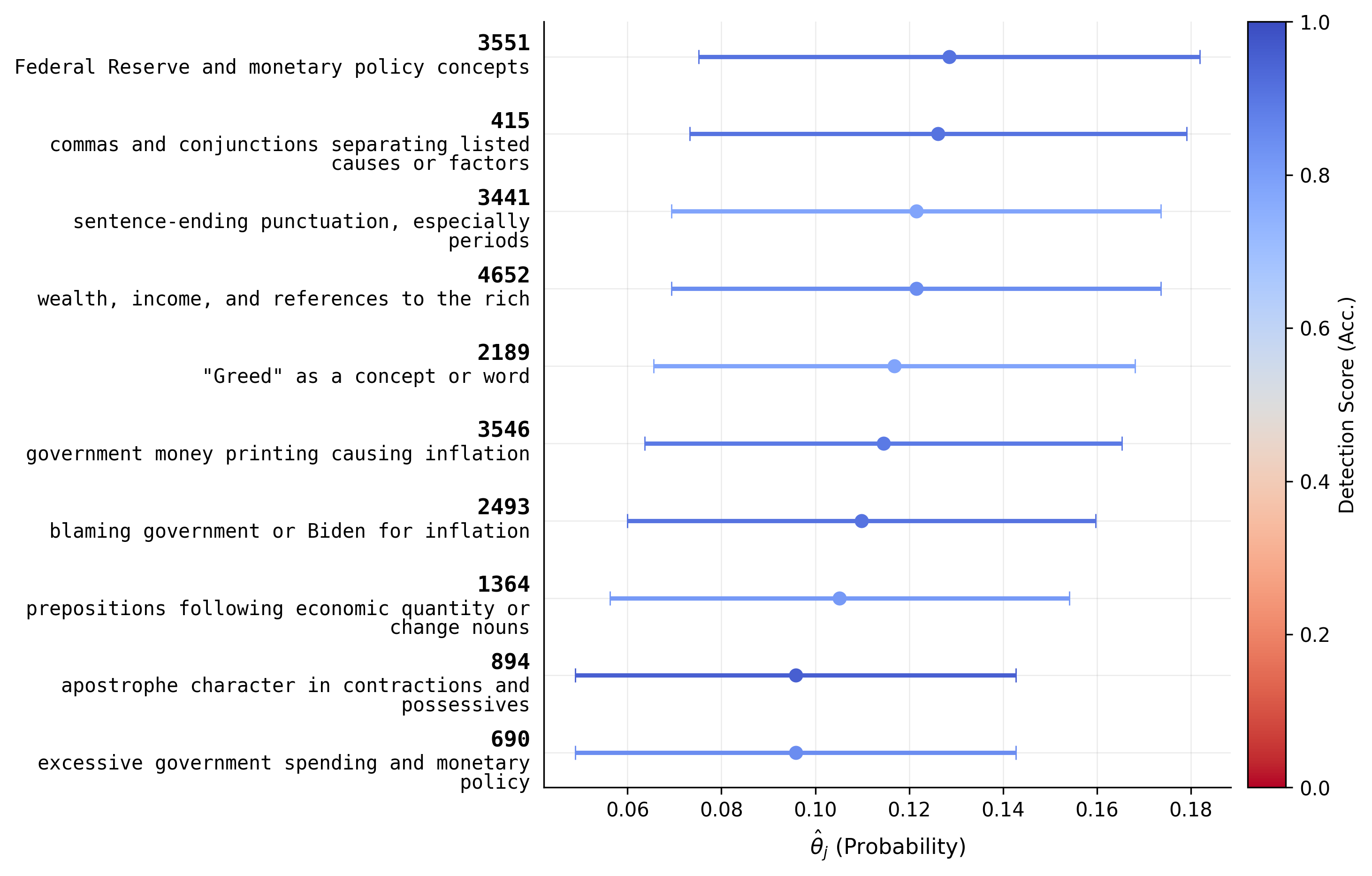}
\caption{Well-Interpreted Discoveries for Responses to Prompt ``High inflation is caused by...'' in \cite{stantcheva_why_2024}, Rank 11-20.}
\label{fig:infl2}
\end{figure}

Beyond filtering out greater than median activating concepts, no other choices were made to delimit the space of possible discoveries, or focus them in some way on the space of economic- or inflation-relevant concepts; the exact same explainer LLMs and SAE that yielded the results of Section \ref{sec_dissent} yielded the results in this section.

\clearpage
\subsection{Calibrated Size Control Simulations}\label{sec_simulate}

In order to validate the theory of Section \ref{sec_mht} and better understand the finite sample properties of $k$-FWER control when $p$ grows (much) larger than $n$, I perform Monte Carlo style simulations calibrated to data from the first narrative provision experiment of \cite{andre_narratives_2026}, which is further analyzed in Section \ref{sec_narratives}.

The simulations plotted in Figure \ref{fig:montecarlo} show empirical size control of single-step $k$-FWER controlled multiple testing as $k,n,p$ vary, based on tests of no average treatment effect under a true global null enforced by placebo assignment (see Remark \ref{rmk_placebo} for a more detailed description). Simulations for a specific $k,n,p$ combination are produced by choosing a random family of $p$ SAE features; drawing 2,000 datasets of size $n$ iid from the empirical distribution of the real concept embedding data and subsetting to these $p$ features; running $k$-FWER controlled multiple testing of treatment effect coefficients in $p$ regressions as described in Section \ref{sec_narratives} for each drawn dataset (with realized degenerate features' t-statistics set to 0); and computing the empirical share of datasets in which $k$ or more false rejections are made. The 2,000 simulated datasets are drawn once per $n$ and shared across every $p$ and $k$, and SAE features for larger $p$ nest all of those features from smaller $p$ for the sake of comparability. 

These simulations are consistent with the large sample size control properties of the theorems discussed in the main text. Control is violated empirically in only the smallest sample size simulated, $n=150$, for $k=1$. Otherwise, in these calibrated finite sample simulations, the procedure is
conservative. A possible explanation is that, in finite sample simulations based on the empirical distribution, there are many t-stats for very rare features that have tails that are meaningfully lighter than Gaussian,
but the Gaussian multiplier bootstrap assigns every coordinate Gaussian
tails, and so the resulting critical values are too large for the realized
statistics. Per the large sample theory, one could think of this as corresponding to a large effective $B_n$, and so convergence to exact size control is slow. 

To investigate this explanation further, in Figure \ref{fig:montecarloii} I perform simulations restricted to SAE features that appear in at least 1\% of (and no more than 99\% of) the original experiment responses, i.e., features that should have t-stats with distributions closer to Gaussians in these finite samples. We see that, as anticipated, size control appears noticeably less conservative in larger samples.

\begin{figure}[H]
\centering
\includegraphics[width=1.0\textwidth]{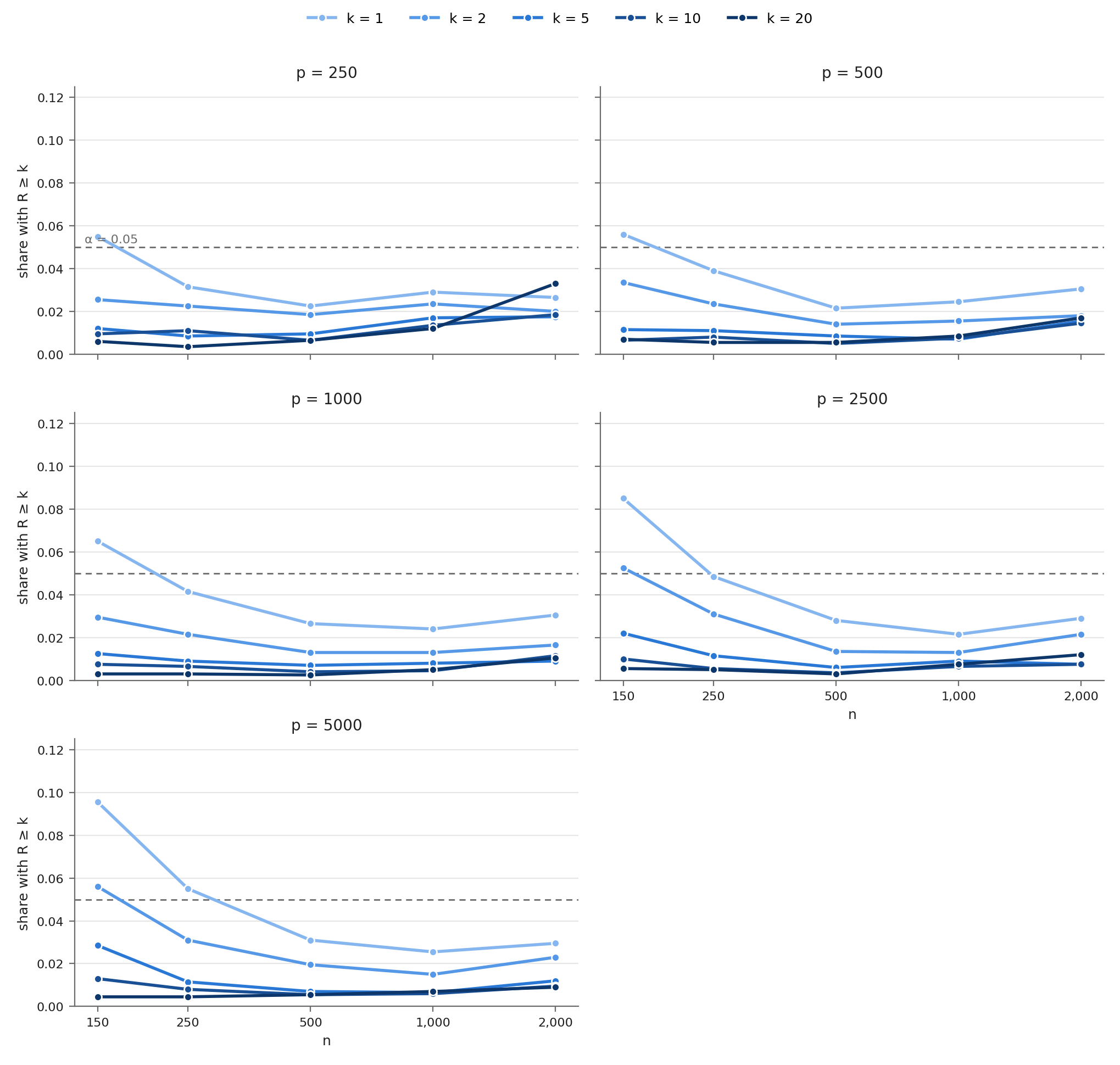}
\caption{Size Control Simulations Calibrated to First Narrative Provision Experiment of \cite{andre_narratives_2026}.}
\label{fig:montecarlo}
\end{figure}

\begin{figure}[H]
\centering
\includegraphics[width=1.0\textwidth]{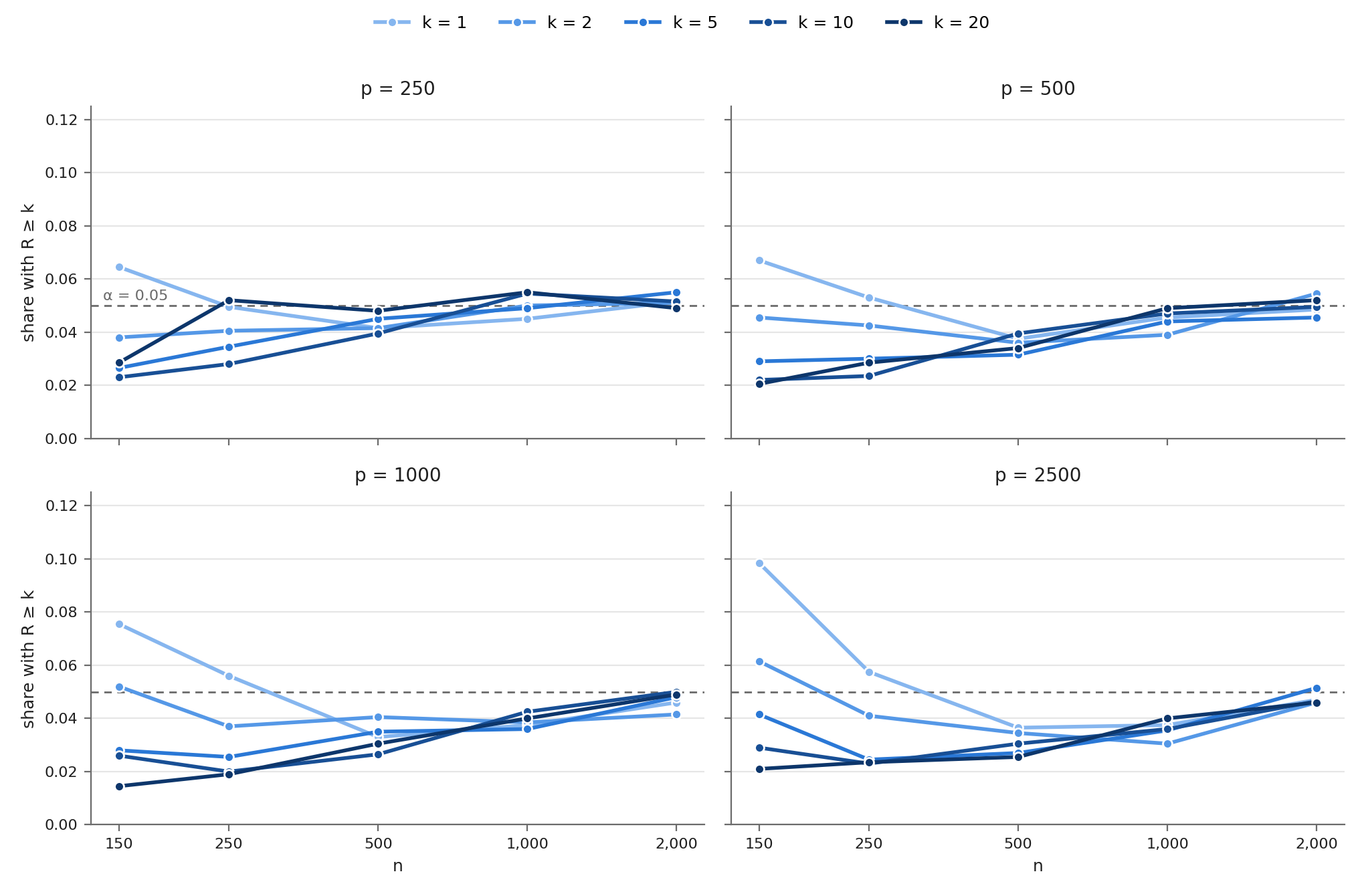}
\caption{Size Control Simulations Calibrated to Less Extreme SAE Features in First Narrative Provision Experiment of \cite{andre_narratives_2026}.}
\label{fig:montecarloii}
\end{figure}

\clearpage 
\subsection{Additional Plots}\label{sec_addlplots}

\begin{figure}[H]
\centering
\includegraphics[width=0.9\textwidth]{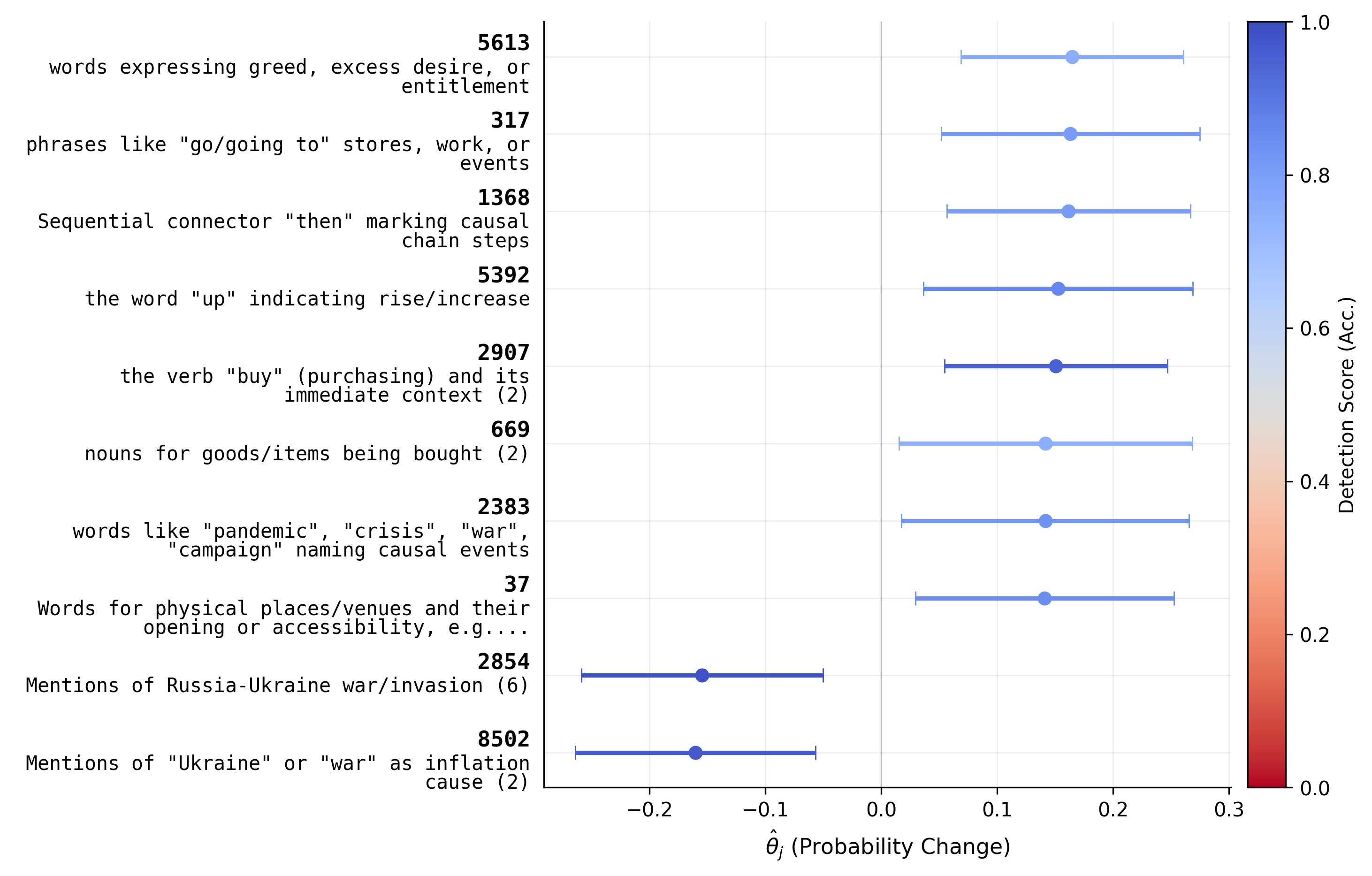}
\caption{Well-Interpreted Discoveries for First Narrative Provision Experiment of \cite{andre_narratives_2026}, Rank 21-30.}
\label{fig:narratives3}
\end{figure}

\begin{figure}[H]
\centering
\includegraphics[width=0.9\textwidth]{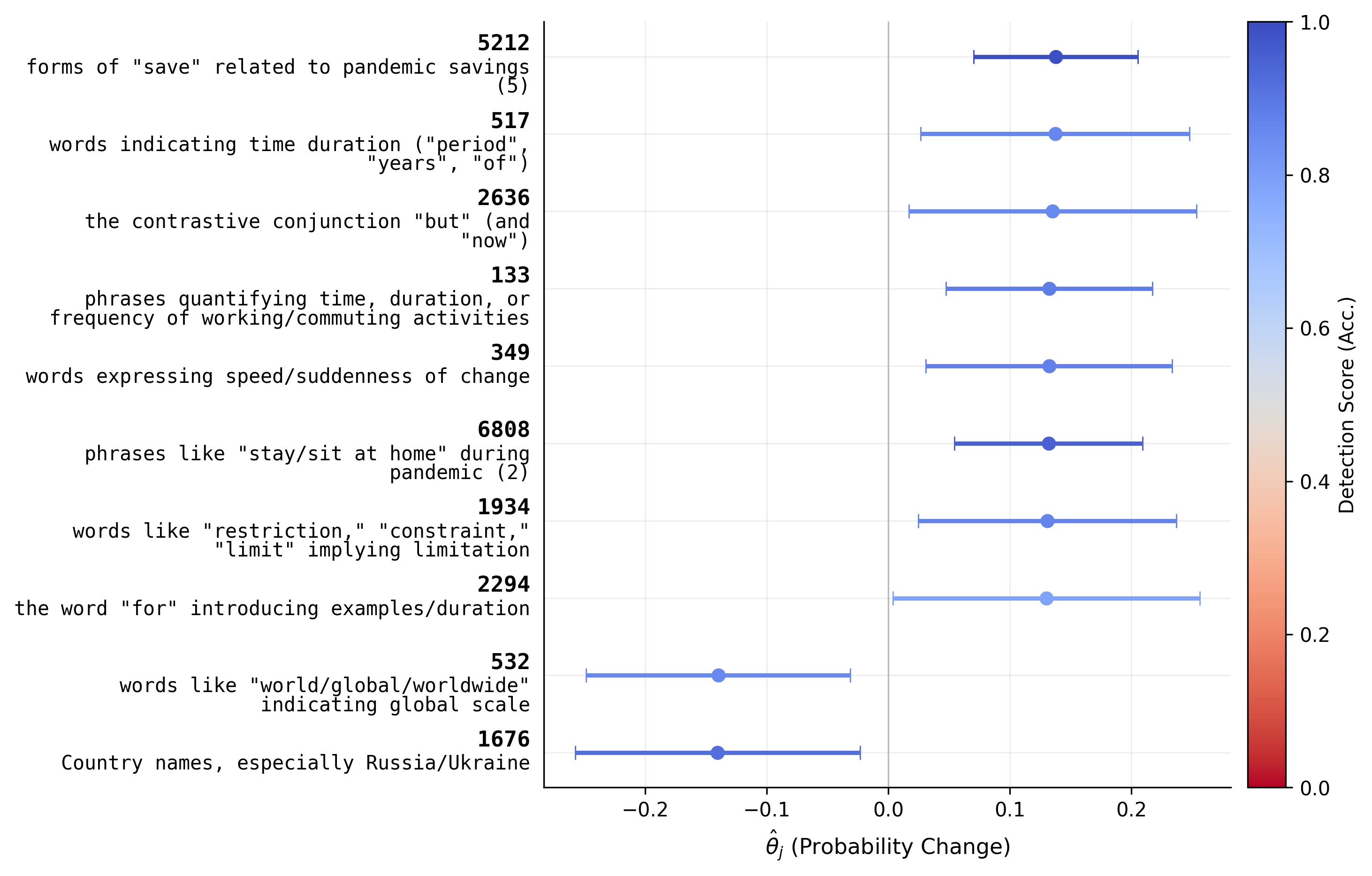}
\caption{Well-Interpreted Discoveries for First Narrative Provision Experiment of \cite{andre_narratives_2026}, Rank 31-40.}
\label{fig:narratives4}
\end{figure}

\begin{figure}[H]
\centering
\includegraphics[width=0.9\textwidth]{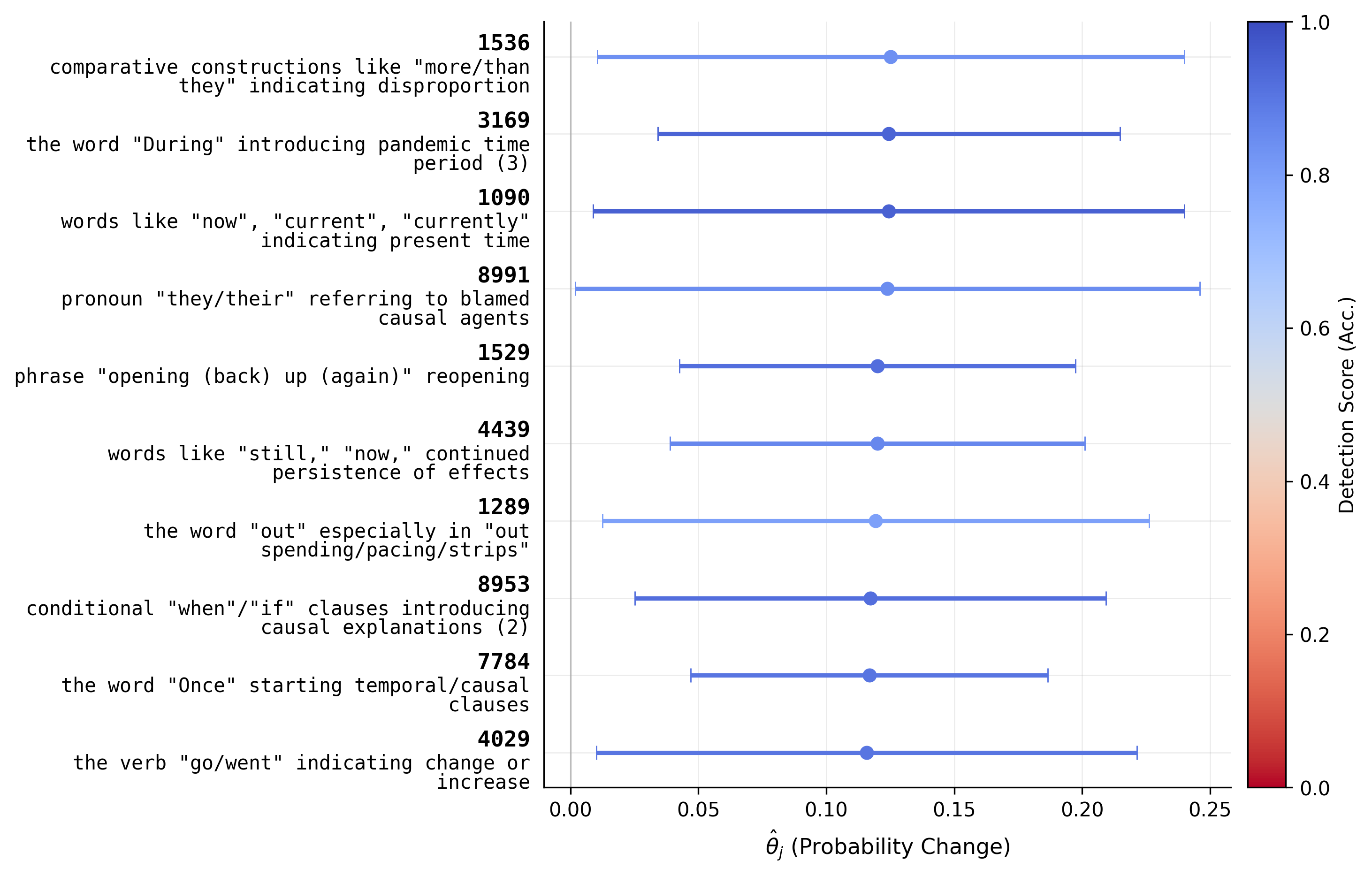}
\caption{Well-Interpreted Discoveries for First Narrative Provision Experiment of \cite{andre_narratives_2026}, Rank 41-50.}
\label{fig:narratives5}
\end{figure}

\begin{figure}[H]
\centering
\includegraphics[width=0.9\textwidth]{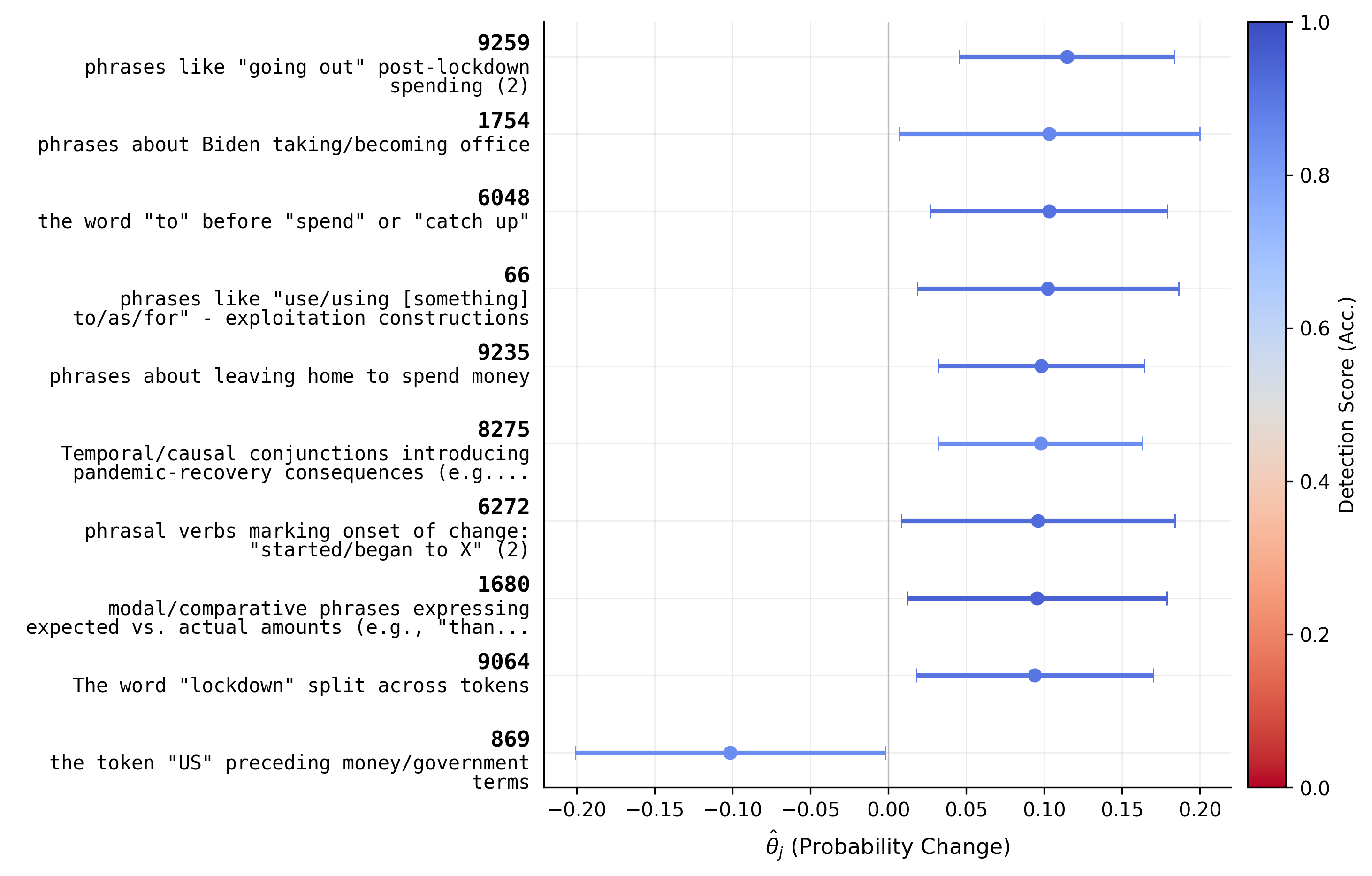}
\caption{Well-Interpreted Discoveries for First Narrative Provision Experiment of \cite{andre_narratives_2026}, Rank 51-60.}
\label{fig:narratives6}
\end{figure}

\begin{figure}[H]
\centering
\includegraphics[width=0.9\textwidth]{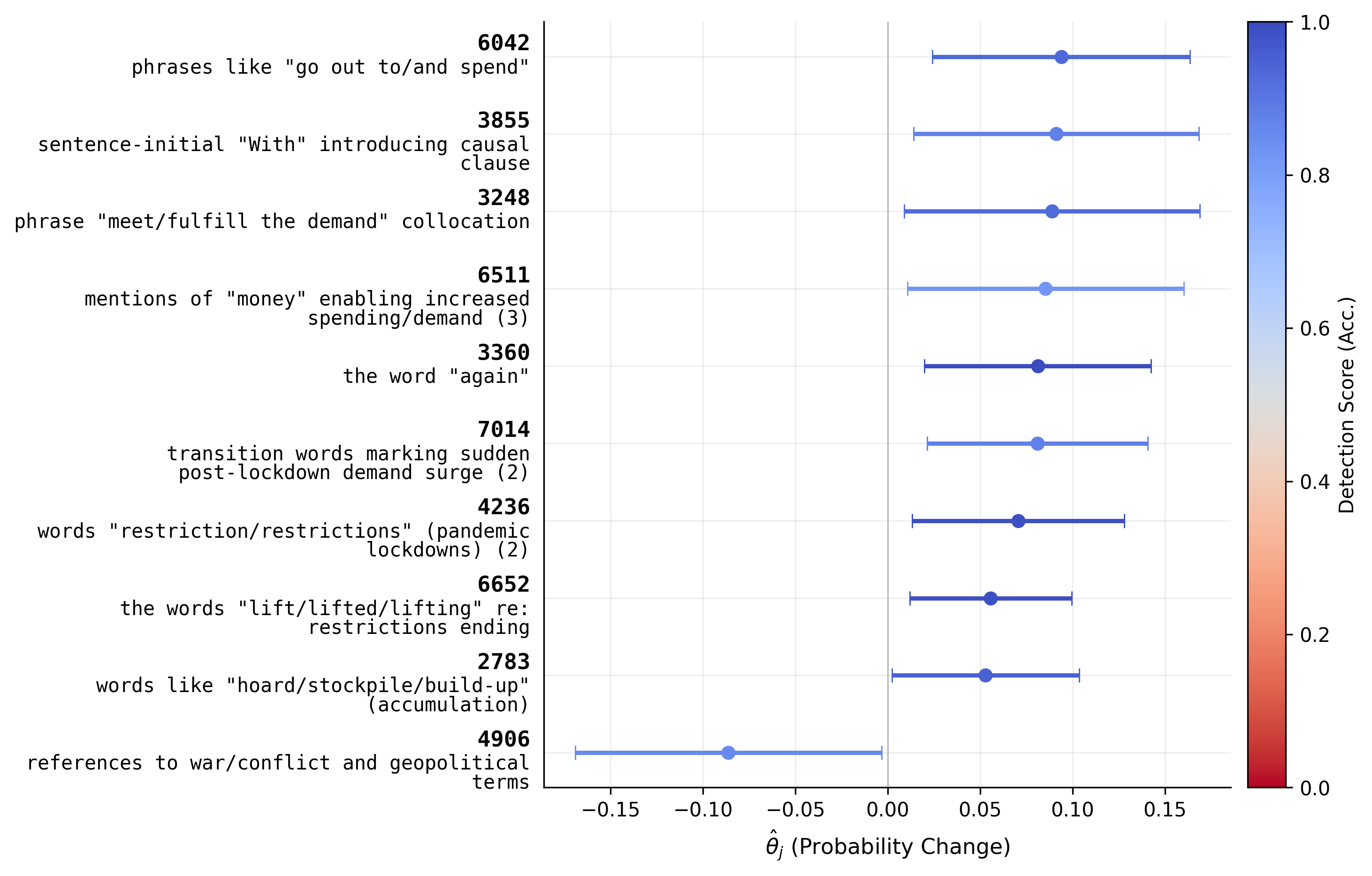}
\caption{Well-Interpreted Discoveries for First Narrative Provision Experiment of \cite{andre_narratives_2026}, Rank 61-70.}
\label{fig:narratives7}
\end{figure}

\begin{figure}[H]
\centering
\includegraphics[width=0.9\textwidth]{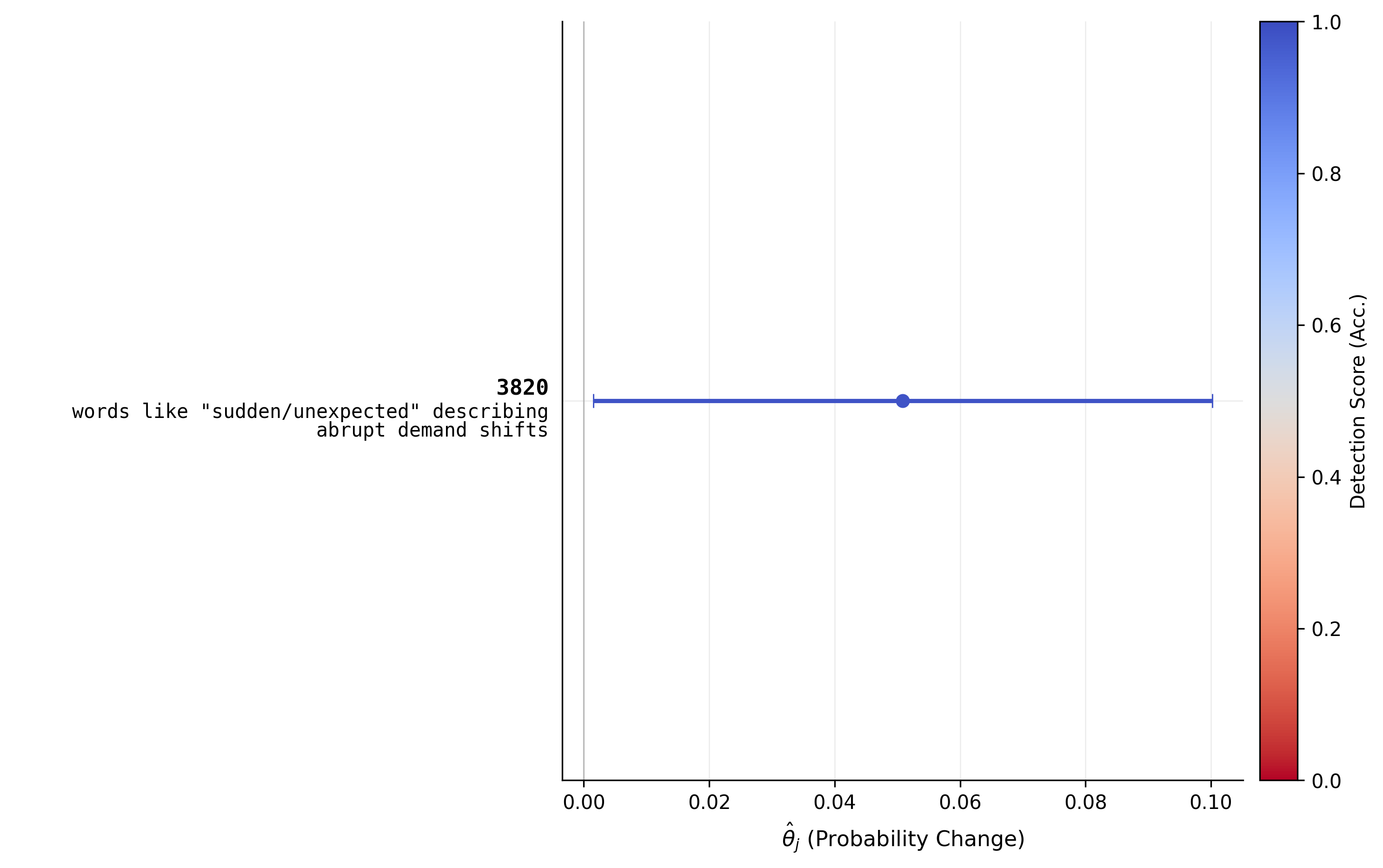}
\caption{Well-Interpreted Discoveries for First Narrative Provision Experiment of \cite{andre_narratives_2026}, Rank 71-.}
\label{fig:narratives8}
\end{figure}

\begin{figure}[H]
\centering
\includegraphics[width=0.9\textwidth]{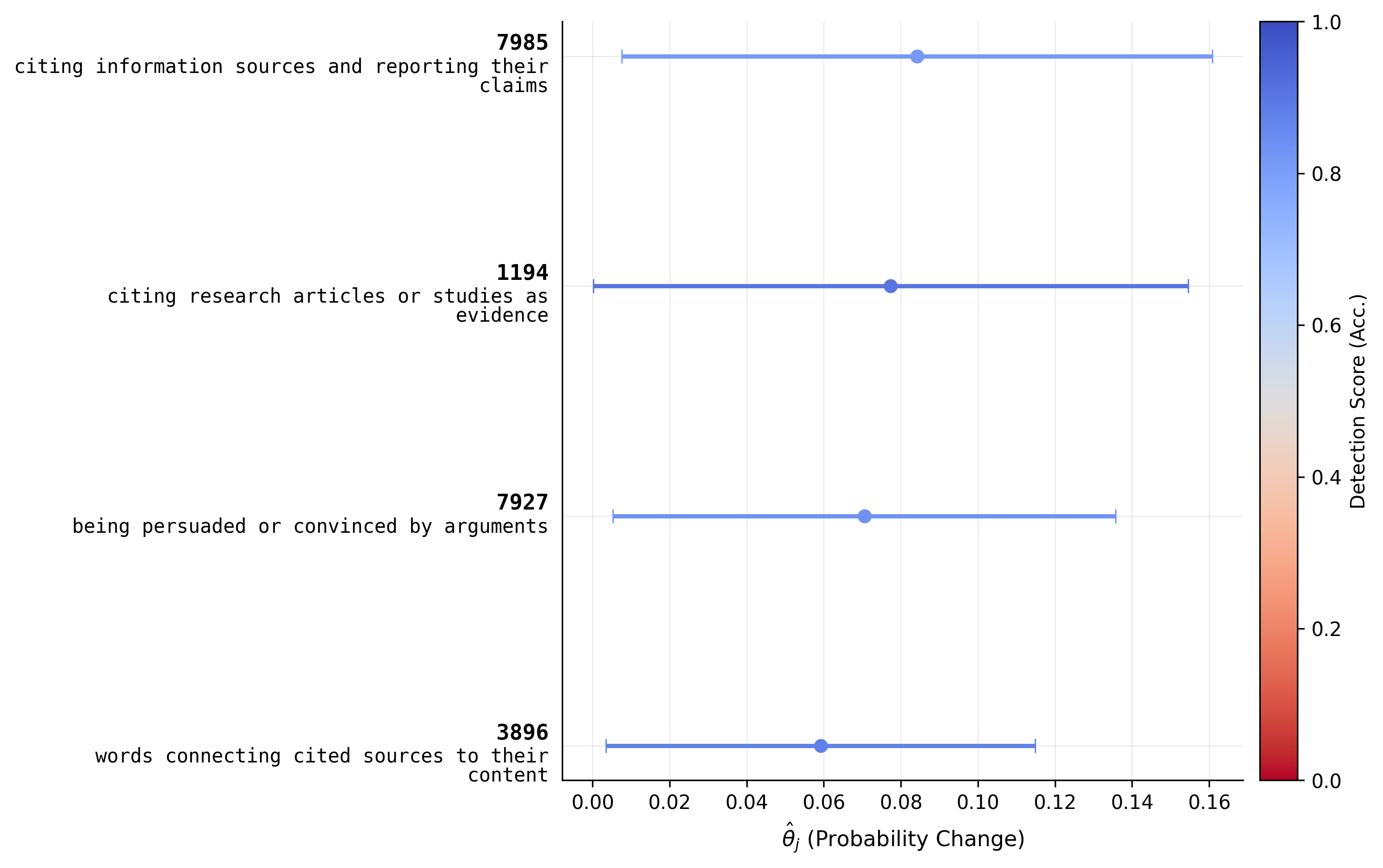}
\caption{Well-Interpreted Discoveries for Experiment 2 of \cite{bursztyn_justifying_2023}, Rank 11-.}
\label{fig:dissent2}
\end{figure}

\begin{figure}[H]
\centering
\includegraphics[width=0.9\textwidth]{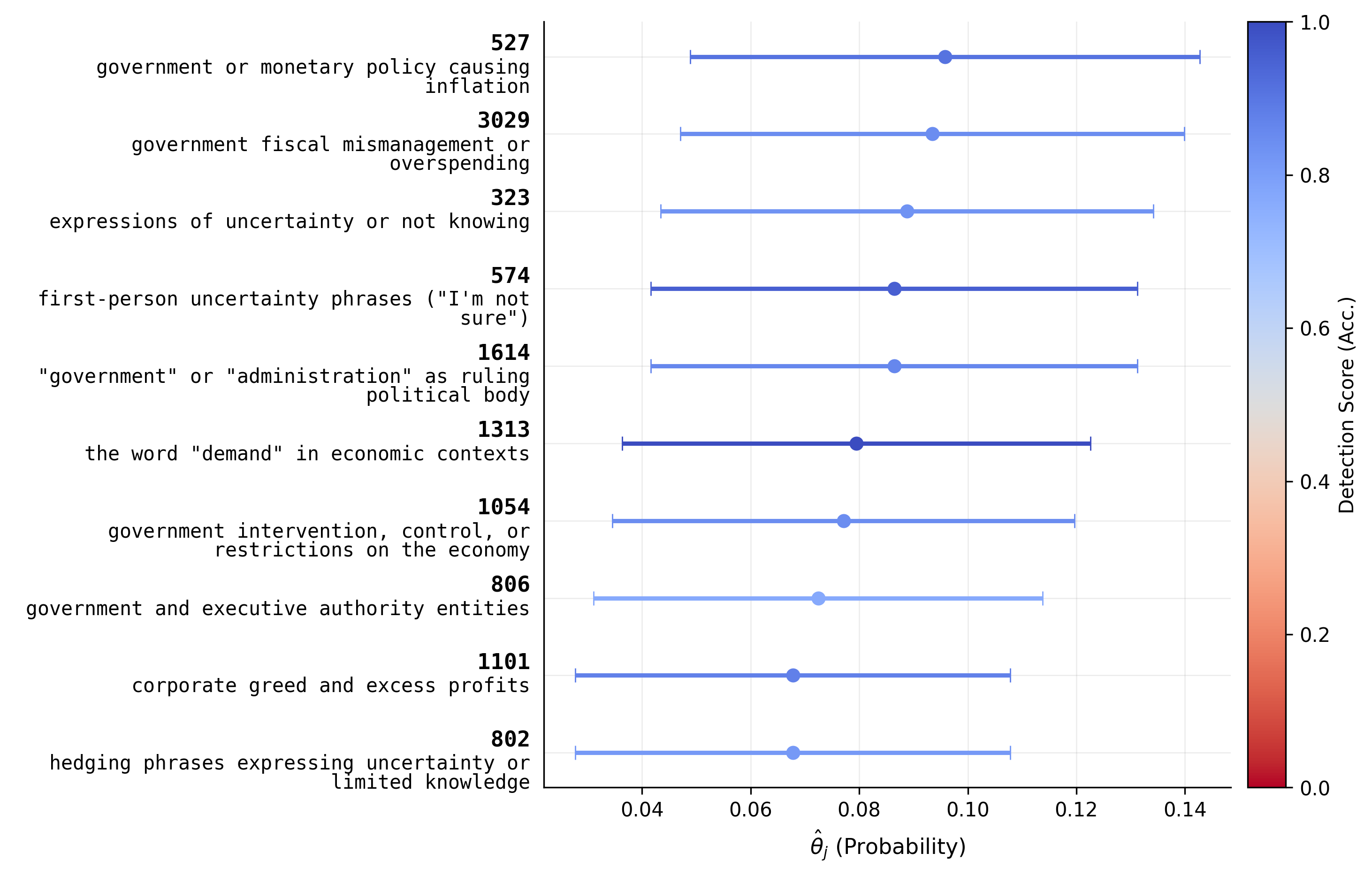}
\caption{Well-Interpreted Discoveries for Responses to Prompt ``High inflation is caused by...'' in \cite{stantcheva_why_2024}, Rank 21-30.}
\label{fig:infl3}
\end{figure}

\begin{figure}[H]
\centering
\includegraphics[width=0.9\textwidth]{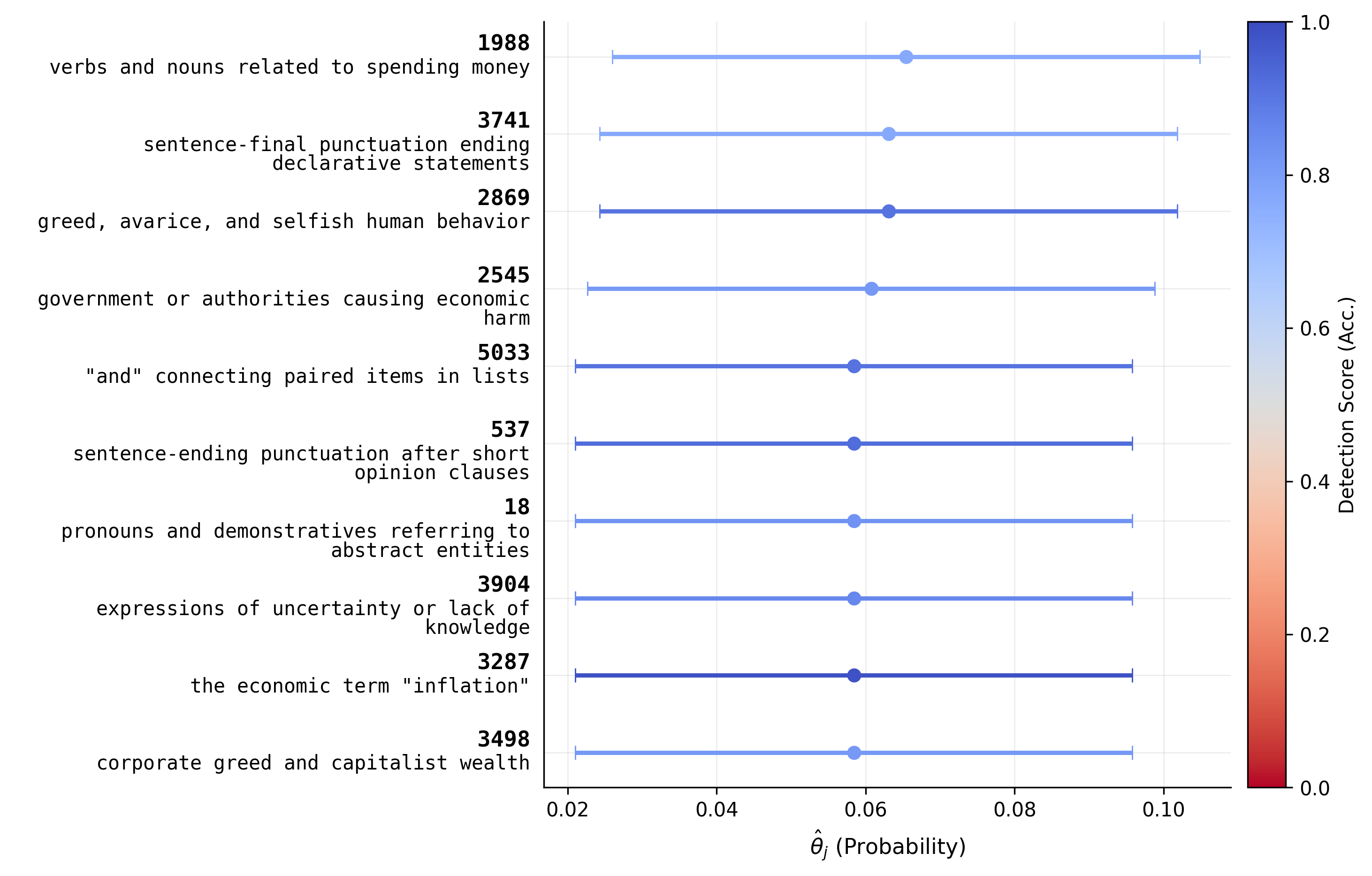}
\caption{Well-Interpreted Discoveries for Responses to Prompt ``High inflation is caused by...'' in \cite{stantcheva_why_2024}, Rank 31-40.}
\label{fig:infl4}
\end{figure}

\begin{figure}[H]
\centering
\includegraphics[width=0.9\textwidth]{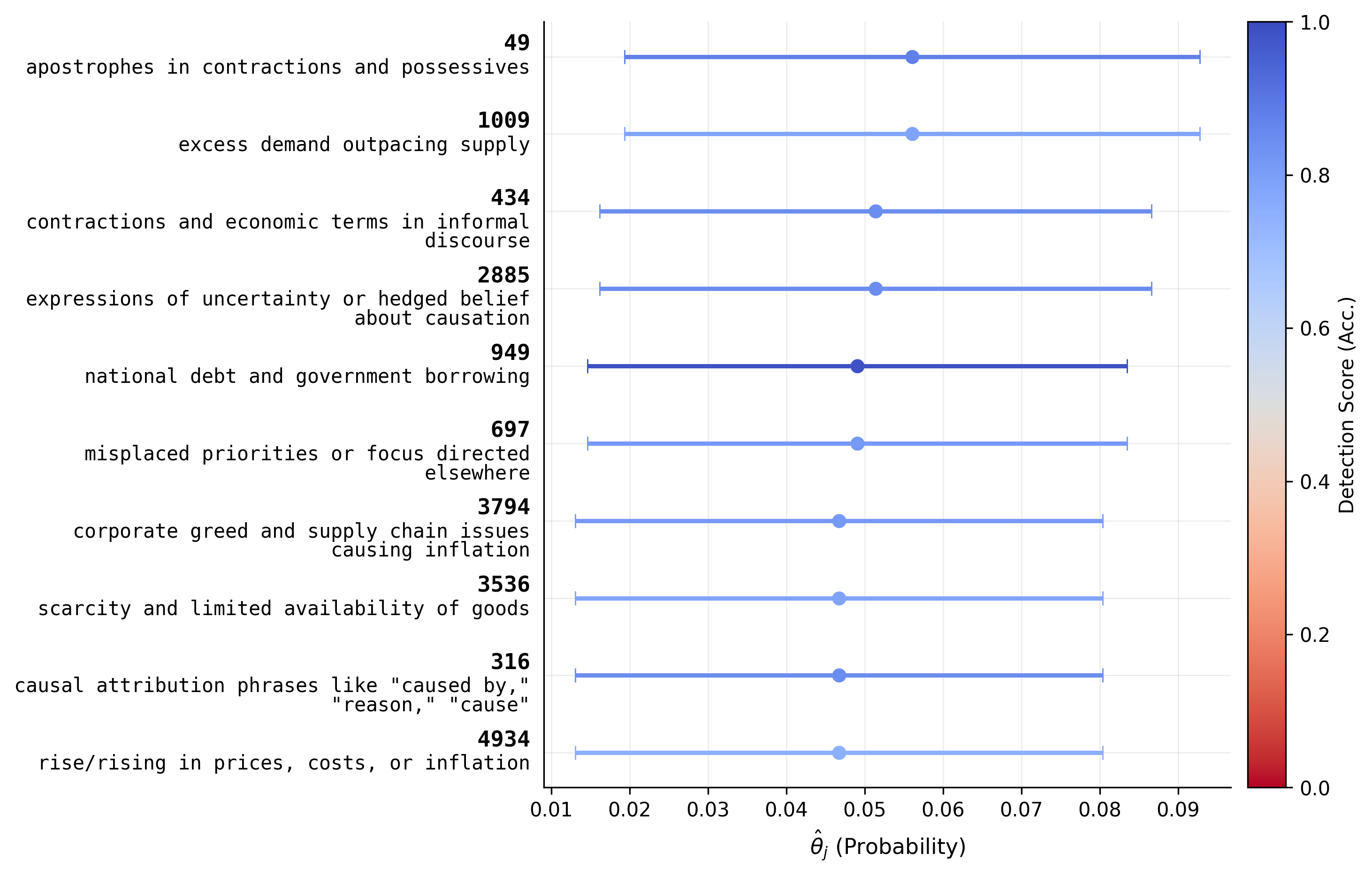}
\caption{Well-Interpreted Discoveries for Responses to Prompt ``High inflation is caused by...'' in \cite{stantcheva_why_2024}, Rank 41-50.}
\label{fig:infl5}
\end{figure}

\begin{figure}[H]
\centering
\includegraphics[width=0.9\textwidth]{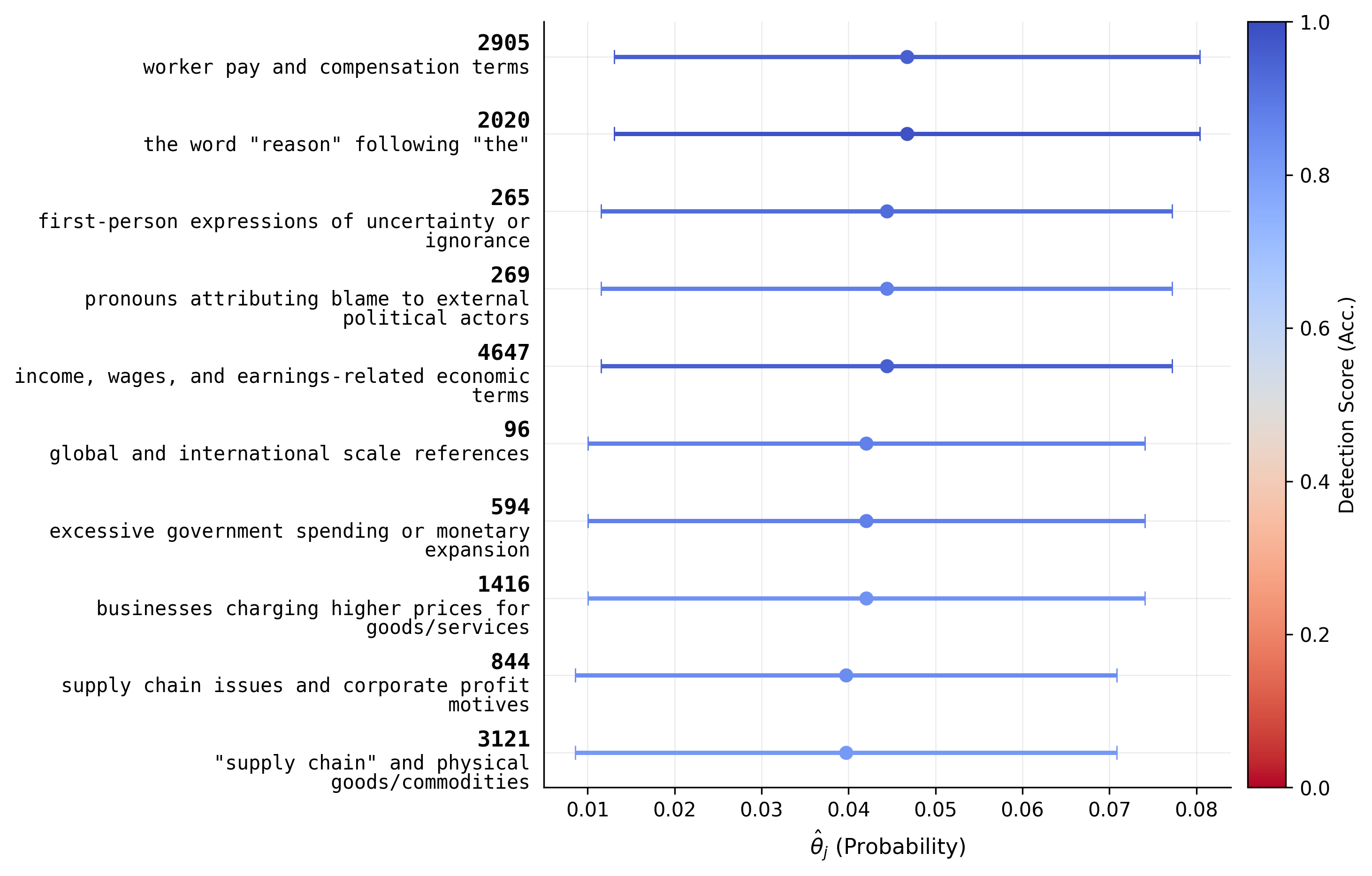}
\caption{Well-Interpreted Discoveries for Responses to Prompt ``High inflation is caused by...'' in \cite{stantcheva_why_2024}, Rank 51-60.}
\label{fig:infl6}
\end{figure}

\begin{figure}[H]
\centering
\includegraphics[width=0.9\textwidth]{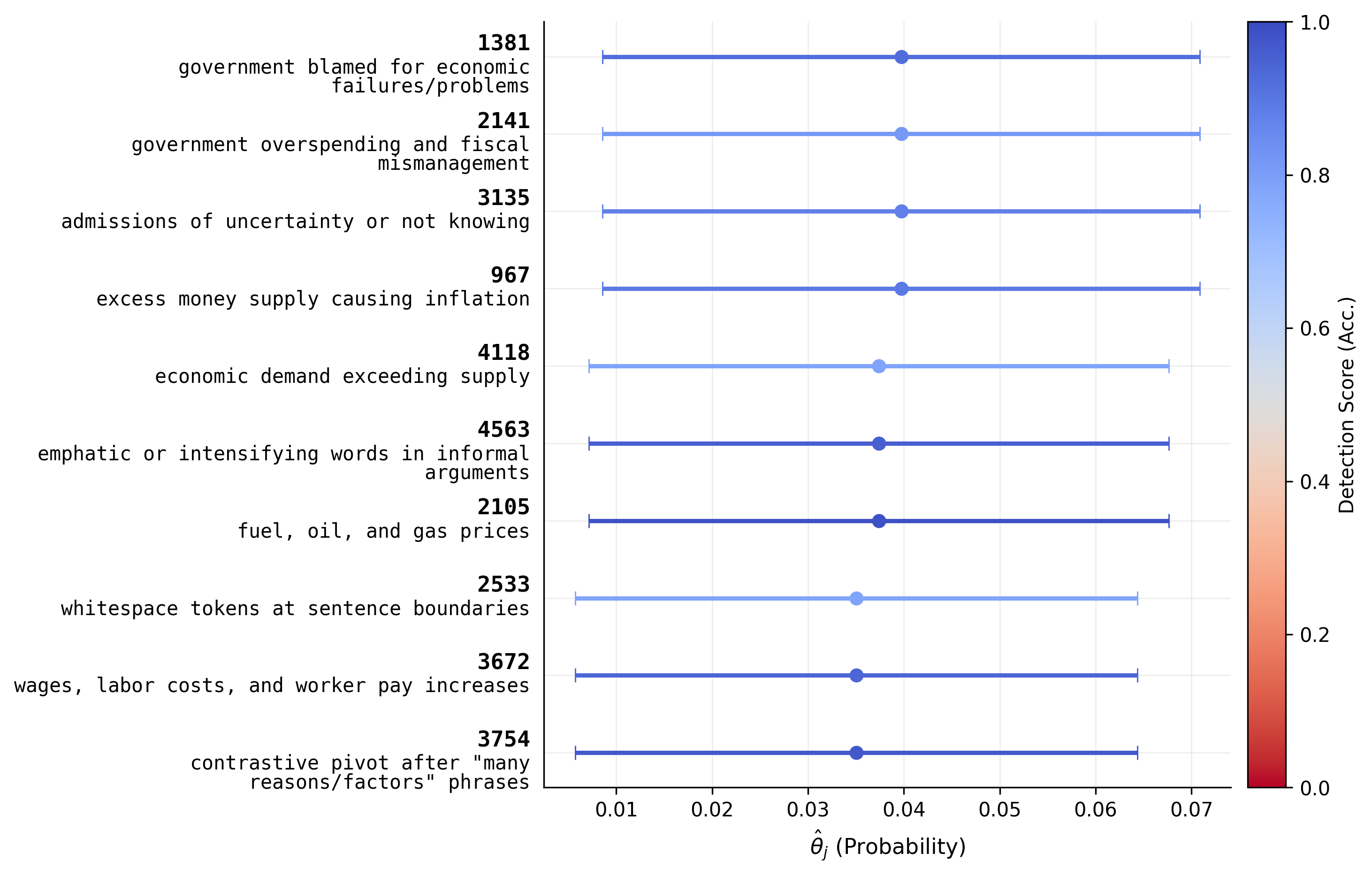}
\caption{Well-Interpreted Discoveries for Responses to Prompt ``High inflation is caused by...'' in \cite{stantcheva_why_2024}, Rank 61-70.}
\label{fig:infl7}
\end{figure}

\begin{figure}[H]
\centering
\includegraphics[width=0.9\textwidth]{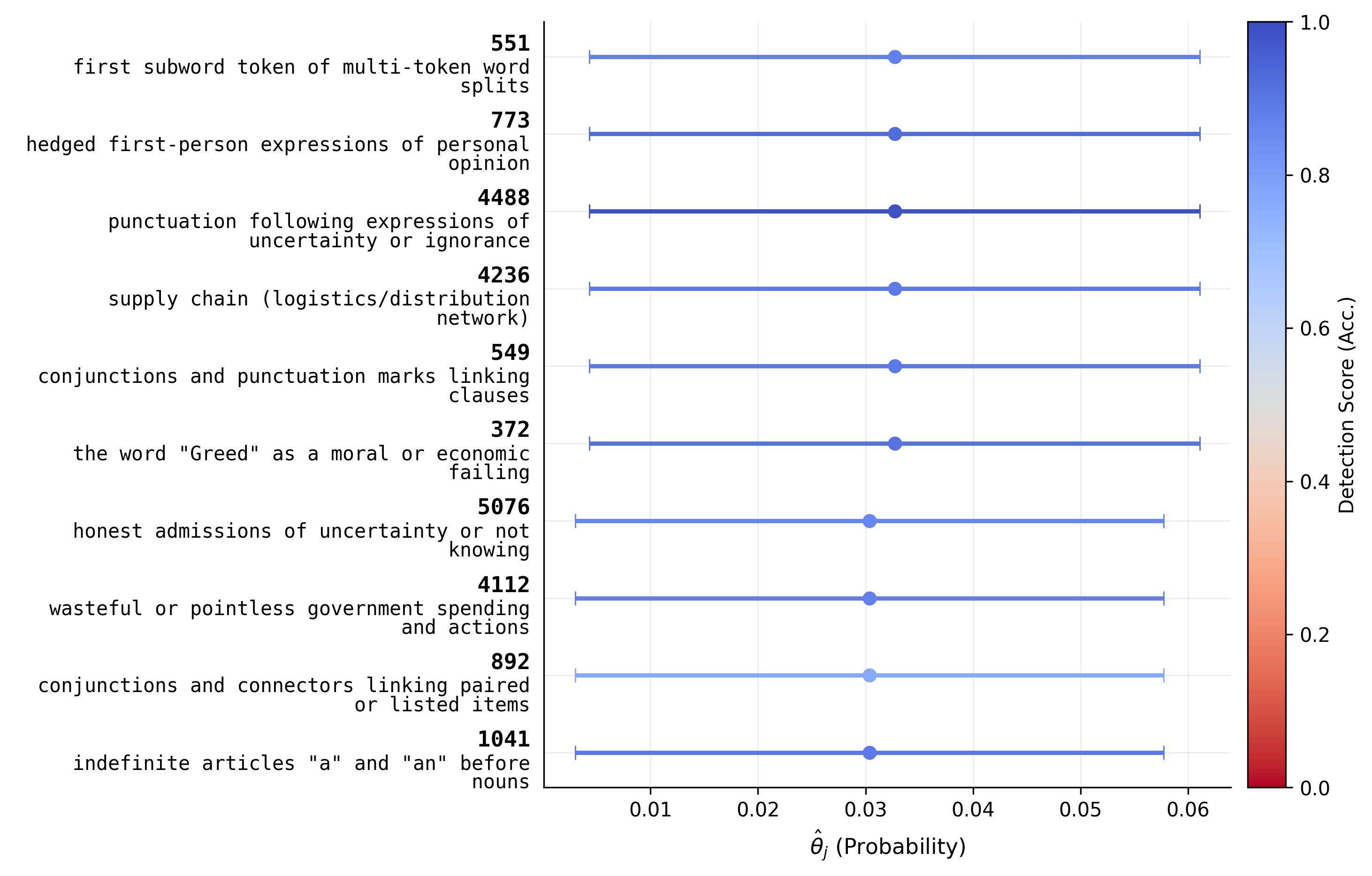}
\caption{Well-Interpreted Discoveries for Responses to Prompt ``High inflation is caused by...'' in \cite{stantcheva_why_2024}, Rank 71-80.}
\label{fig:infl8}
\end{figure}

\begin{figure}[H]
\centering
\includegraphics[width=0.9\textwidth]{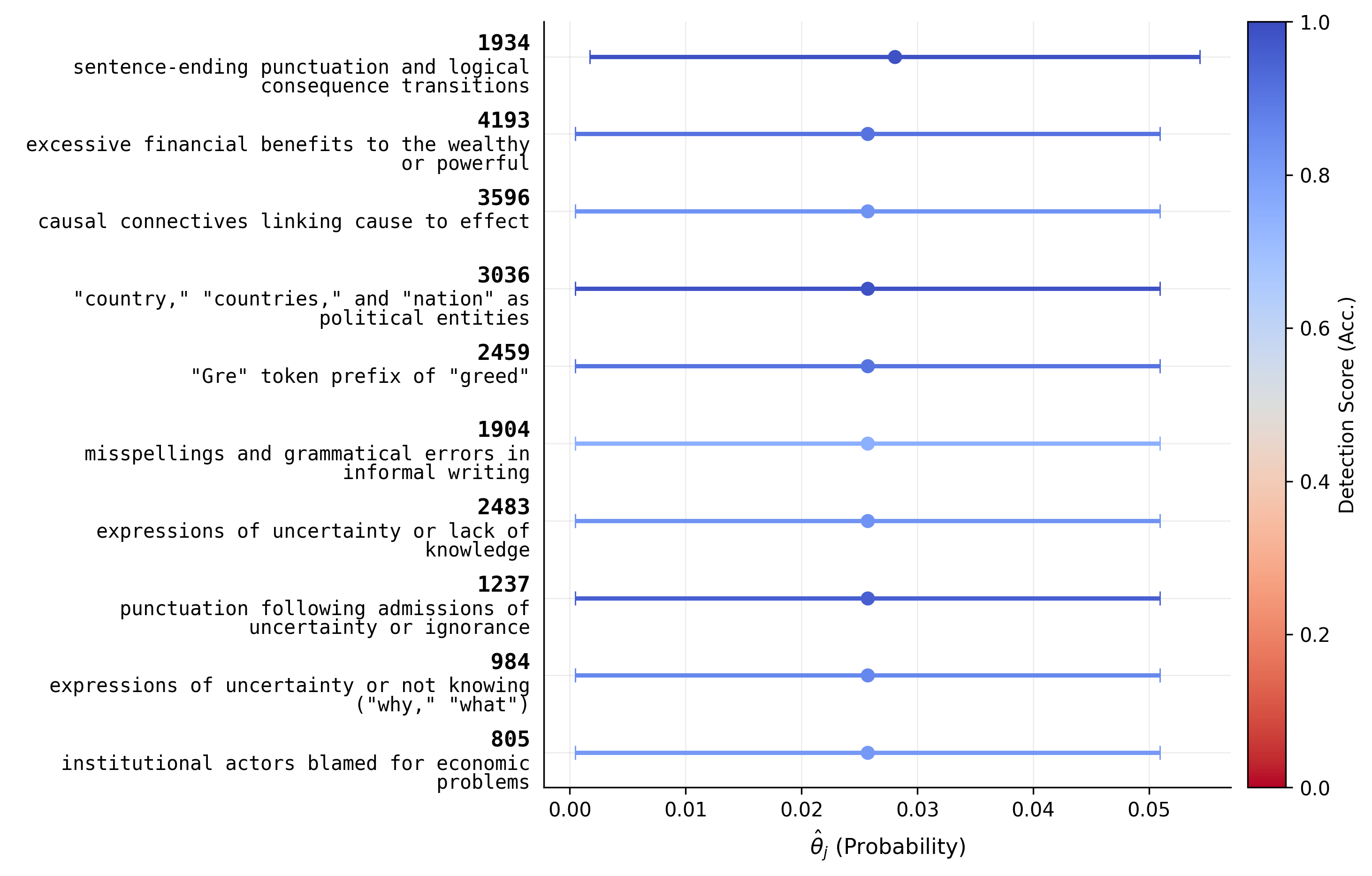}
\caption{Well-Interpreted Discoveries for Responses to Prompt ``High inflation is caused by...'' in \cite{stantcheva_why_2024}, Rank 81-90.}
\label{fig:infl9}
\end{figure}

\begin{figure}[H]
\centering
\includegraphics[width=0.9\textwidth]{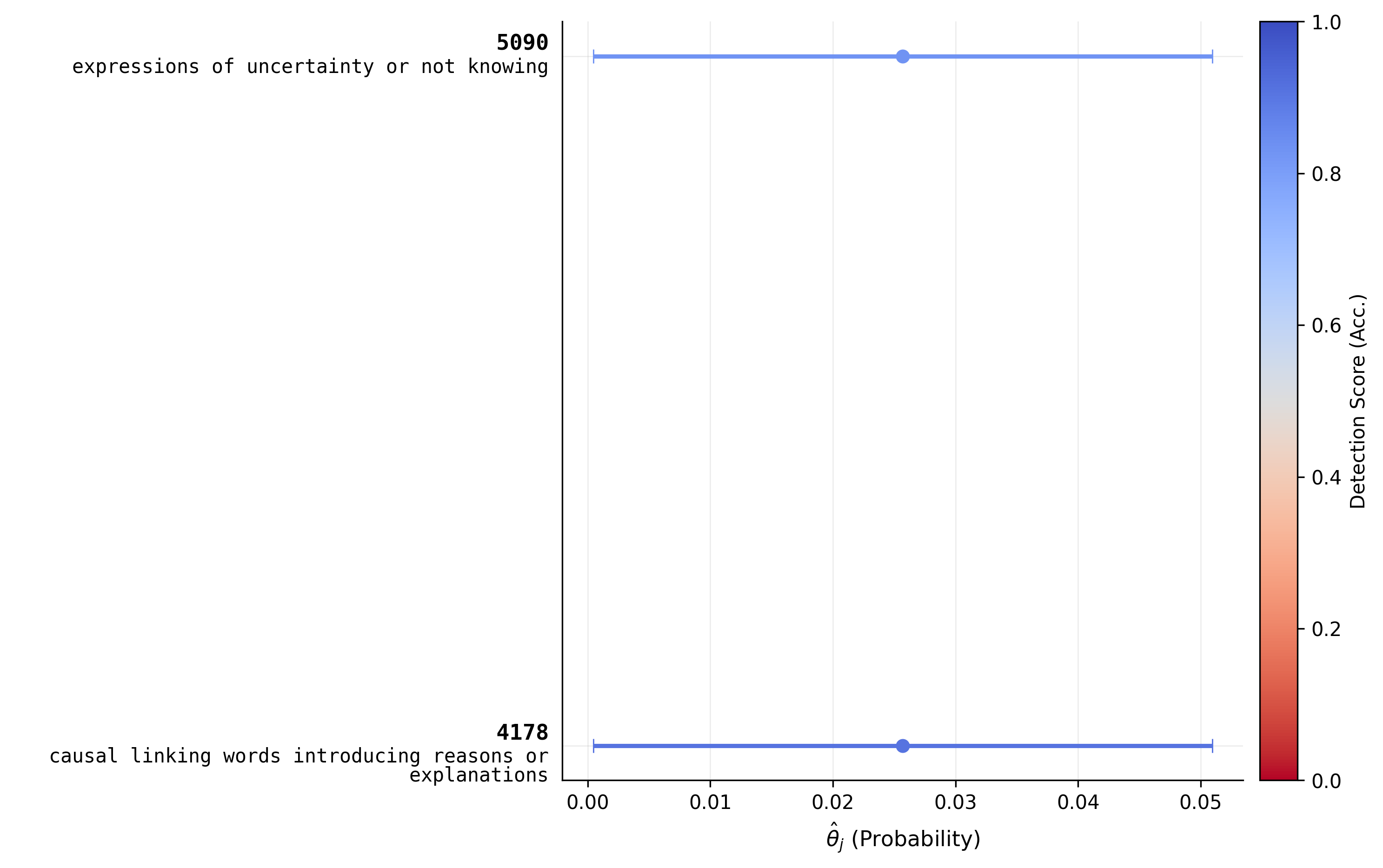}
\caption{Well-Interpreted Discoveries for Responses to Prompt ``High inflation is caused by...'' in \cite{stantcheva_why_2024}, Rank 91-.}
\label{fig:infl10}
\end{figure}

\begin{figure}[H]
\centering
\includegraphics[width=0.9\textwidth]{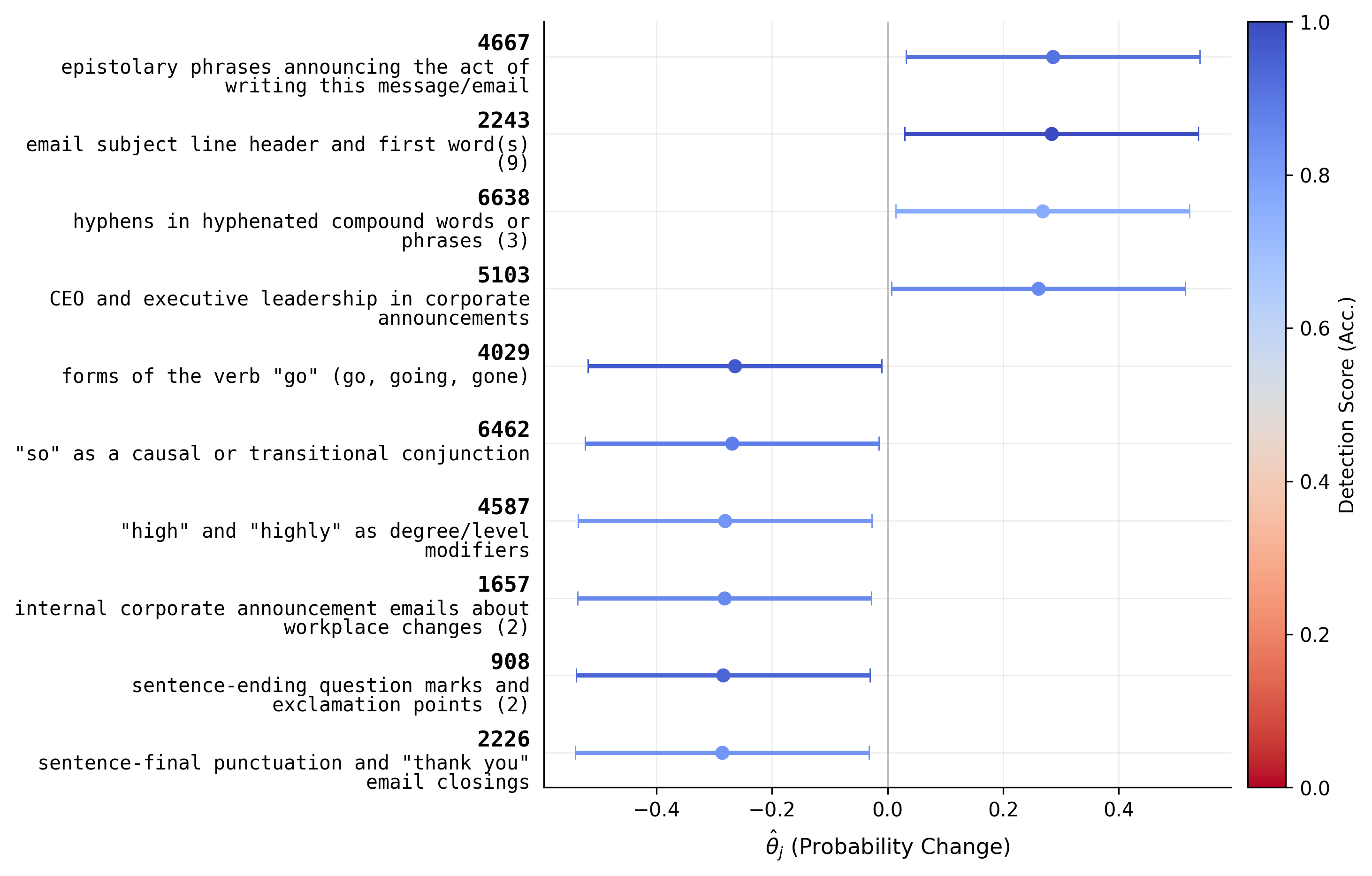}
\caption{Well-Interpreted Discoveries for \cite{noy_experimental_2023}, Rank 21-30.}
\label{fig:chatgpt3}
\end{figure}

\begin{figure}[H]
\centering
\includegraphics[width=0.9\textwidth]{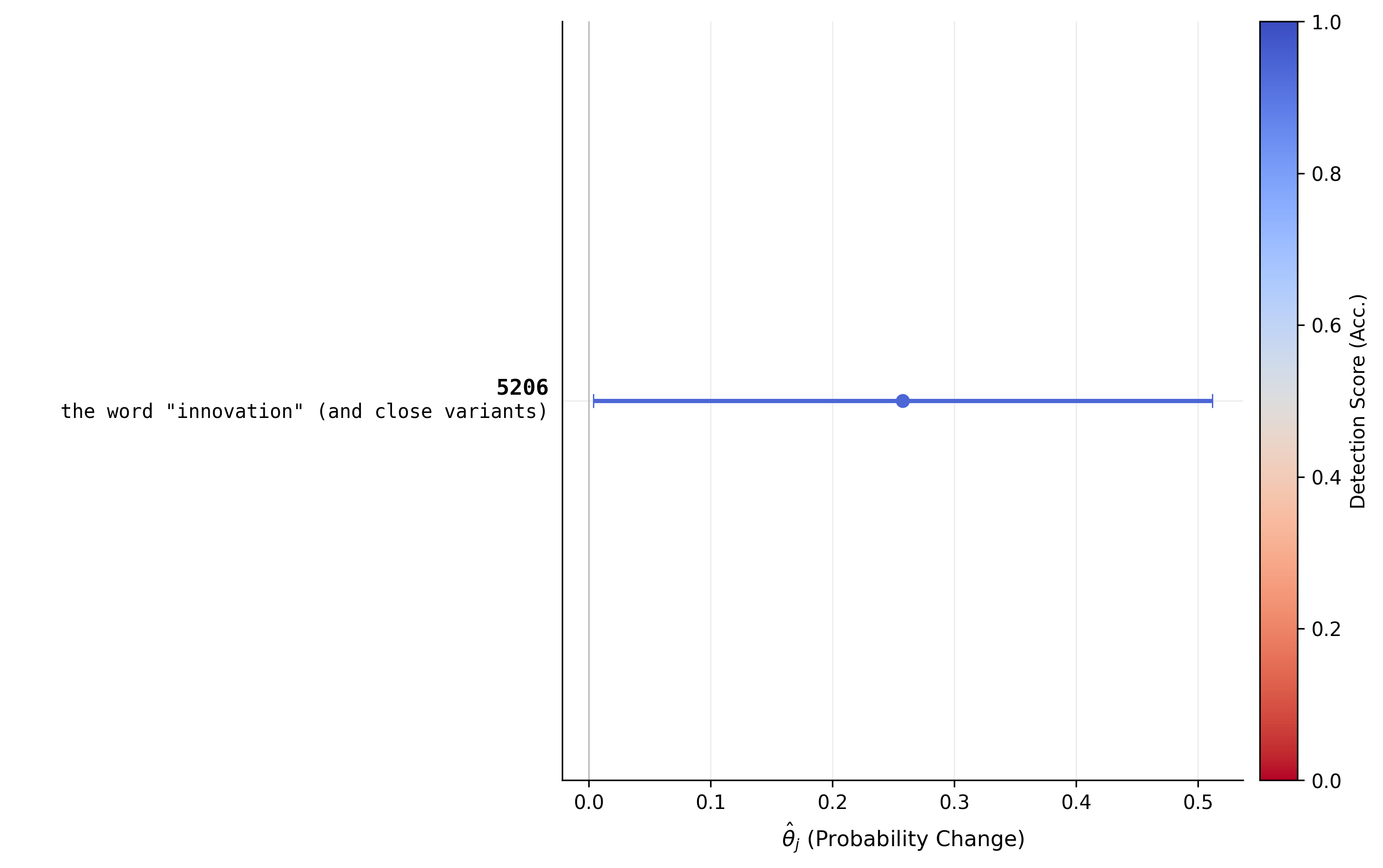}
\caption{Well-Interpreted Discoveries for \cite{noy_experimental_2023}, Rank 31-.}
\label{fig:chatgpt4}
\end{figure}


\subsection{Implementation Details }\label{app_autointerp}

\subsubsection{Concept Embedding}

I use the T-SAE \citep{bhalla2026temporal} as the basis for a concept embedding, preferring it to alternatives due to its ``temporal loss'' modification helping to separate out semantic from syntactic features in the learning process---leading to greater interpretability in theory and in practice. Because off-the-shelf feature descriptions are available for the T-SAE, I use them to inform a simple filtering procedure that generates the final concept embedding: if multiple features have the exact same off-the-shelf description up to whitespace and casing, they are max-pooled for the purposes of indicating the presence of the (putatively same) underlying concept. This filtering thereby reduces the dimensionality of the concept embedding.

\subsubsection{Local Autointerp}

Building autointerp pipelines for scalable description of features produced from sparse dictionary learning methods is an active area of AI/ML research. Most best practices rely on leveraging LLMs themselves to interpret the features and later score these interpretations, per the pioneering work of \cite{bills2023language}. 

Heuristically, most of these pipelines operate by collecting text samples on which a given feature activates the most in a particular corpus (as well as perhaps sampling additional text samples across the empirical distribution of activation strengths) and weaving information about these activations---alongside the activating text---into prompts that LLMs are asked to interpret. Though autointerp descriptions are available for the T-SAE and some other SAEs, these descriptions are learned based on the distribution of the SAE's pretraining data (the Pile \citep{gao_pile_2021} for the T-SAE), and as such benefit from being refined through a \textit{local} autointerpretation strategy that uses text samples from the unstructured dataset distribution $P$.

For the purposes of generation, I take cues from the methods of \cite{paulo_automatically_2024,bills2023language,karvonen_saebench_2025} to arrive at a prompt of the form:
\begin{quote}
\begin{verbatim}
{"role": "system", "content": """
You are a meticulous AI researcher studying the internal
representations of language models. Your task is to produce short,
precise descriptions of what a neuron responds to. Descriptions must
be specific enough that a reader could accurately predict which new
texts would activate the neuron.
"""},
{"role": "user", "content": f"""
A particular neuron most activated on the following examples when a
corpus of texts was passed through a language model. Tokens that
triggered the neuron are wrapped in << ... >> and annotated with
their activation strength as <<token(v)>>, where v is an integer
from 1 (weakly activating) to 10 (the neuron's strongest response).
Higher-numbered tokens are the most informative; weight them
accordingly, and note that some examples activate the neuron more
strongly overall than others. Unmarked tokens did not activate the
neuron but provide context.

Write a single short phrase (ideally 3-8 words) describing the
specific concept, topic, or linguistic pattern this neuron responds
to, e.g., `news articles and formal reports' or `verbs and nouns
related to actions'. This pattern/description should be abstracted
away from the broad topic of the corpus itself, which will be
consistent across all of the examples. Delimit your answer as
[[description here]].

The << >>( ) glyphs are annotations, not part of the text; base
your description on which tokens they mark and how strongly.

Examples:
{numbered_texts}
"""}
\end{verbatim}
\end{quote}
where $\texttt{numbered\_texts}$ is an enumerated concatenation of $L$ exemplar texts---modified as the prompt suggests---associated with the tokens that activated most highly on a given feature in the main sample of the researcher's unstructured dataset. I use $L=10$ in the empirical examples of this paper, and Anthropic's Claude Sonnet (high reasoning effort, June 2026), to perform interpretations. In practice, I employ a floor of 25\% of an exemplar text's maximum token activation strength for the feature for annotation. Activation strength annotations \texttt{v} are determined by  $$\frac{\text{activation strength} \times 10}{\max(\text{activation strength across $L$ exemplars})} $$ rounded to the nearest integer. For an exact implementation, see code at \url{https://github.com/jscarlson/hdmht-discovery}.

For the purposes of evaluation, I again adapt the detection scoring prompting strategy of \cite{paulo_automatically_2024}, using prompts of the form:
\begin{quote}
\begin{verbatim}
{"role": "system", "content": """
You are an intelligent and meticulous linguistics researcher.
You will be provided a certain latent attribute of text, such as
``male pronouns'' or ``text with negative sentiment''.
You will then be given a text example. Your task
is to determine if the example possesses the latent attribute.
Return 1 if the text possesses the latent attribute,
and return 0 otherwise. Return only this number.
"""
},
{"role": "user", "content": f"""
LATENT ATTRIBUTE: {description}
TEXT EXAMPLE: {text}
"""
}
\end{verbatim}
\end{quote}
where \texttt{description} is an autointerp description being evaluated and \texttt{text} is a text sample being analyzed for the presence of the relevant feature. I use the open-source model Gemma 4 31B to perform evaluations, for the sake of reproducibility and cost effectiveness.

\subsection{Pseudocode Algorithm}\label{sec_pseudocode}

\begin{algorithm}[H]
\caption{Interpretable Discovery from Text Data with $k$-FWER Control (Simplified)}\label{algo1}
\label{alg:framework}
{\footnotesize
\begin{algorithmic}[1]
\Require Data $\{(W_i,Z_i)\}_{i=1}^{n+m}$ drawn i.i.d. from distribution $P$; concept embedding function $h:\mathcal{Z}\to\{0,1\}^{\tilde p}$; map $g:\mathcal{W}\times\{0,1\}\to\mathbb{R}$; level $\alpha$; error control $k$; explainer LLM $\mathcal{E}$; eval dataset size $m$; bootstrap iterations $\mathfrak{B}$.
\Require $k \ll \tilde p$.

\Statex \textbf{Split sample}
\State Randomly split $[n+m]$ into $\mathcal{I}^{\mathrm{estim}}$ and $\mathcal{I}^{\mathrm{eval}}$ with $|\mathcal{I}^{\mathrm{estim}}|= n$ and $|\mathcal{I}^{\mathrm{eval}}|= m$.
\Statex \textbf{Create concept embeddings}
\For{$i=1,\dots,n+m$} \State $Y_i \gets h(Z_i)\in\{0,1\}^{\tilde p}$. \EndFor
\Statex \textbf{Drop degenerate features}
\State Initialize $\mathcal{J} \gets \emptyset$.
\For{$j=1,\dots,\tilde p$} \State If $\exists\, i \text{ such that } Y_{ij}\neq 0$, $\mathcal{J} \gets \mathcal{J} \cup \{j\}$. \EndFor
\State $ p \gets |\mathcal{J}|$.

\Statex \textbf{Compute concept-level estimates and t-stats}
\For{$j\in \mathcal{J}$}
  \State For $i\in\mathcal{I}^{\mathrm{estim}}:X_{ij} \gets g(W_i,Y_{ij})$.
  \State $\hat\theta_j \gets \frac{1}{n}\sum_{i\in\mathcal{I}^{\mathrm{estim}}} X_{ij}; \hat\Sigma_{jj} \gets \frac{1}{n}\sum_{i\in\mathcal{I}^{\mathrm{estim}}}(X_{ij}-\hat\theta_j)^2$.
  \State $T_{n,j} \gets \sqrt{n}\,\hat\theta_j/\sqrt{\hat\Sigma_{jj}}$.
\EndFor
\State $\hat\Lambda \gets \mathrm{diag}(\hat\Sigma_{11},\dots,\hat\Sigma_{pp})$.

\Statex \textbf{High-dimensional selective inference}
\State Critical value $\hat{c}_{n}(1-\alpha, k)\gets 1-\alpha \text{ quantile of }\left|\hat{\Lambda}^{-1 / 2} S_n^B\right|_{[k]}$ under bootstrap law, approx. using $\mathfrak{B}$ multiplier bootstrap iterations with $\{X_i \}_{i \in \mathcal{I}^\textrm{estim}}$.
\State Selected set $\hat J \gets \{j : |T_{n,j}|>\hat{c}_{n}(1-\alpha, k)\}\subseteq[p]$.
\State For $j\in[p]$: gen. sim. CI $\mathcal{C}_j  \gets \left[\hat\theta_j-\hat{c}_{n}(1-\alpha, k)\sqrt{\frac{\hat\Sigma_{jj}}{n}}, \hat\theta_j+\hat{c}_{n}(1-\alpha, k)\sqrt{\frac{\hat\Sigma_{jj}}{n}}\right] \in \mathbb{R}^2$.

\Statex \textbf{Autointerpretation}
\State Initialize $\texttt{CLS}$ using $\mathcal{E}$ and fixed classification prompt.
\For{$\hat \jmath\in\hat J$}
  \State Generate descriptions $\hat\eta_{\hat \jmath} \gets \mathcal{E}(\{Z_i \}_{i \in \mathcal{I}^\textrm{estim}})$ (e.g., implemented per App. Sec. \ref{app_autointerp}).
  \For{$i\in\mathcal{I}^{\mathrm{eval}}$}
     \State $\hat Y_{i \hat \jmath} \gets \mathtt{CLS}(Z_i,\hat\eta_{\hat \jmath})\in\{0,1\}; S_{i\hat \jmath} \gets \mathbf{1}\{Y_{i \hat \jmath}=\hat Y_{i \hat \jmath}\}$.
  \EndFor
  \State Detection score $\hat\theta_{\hat \jmath}^{\mathrm{acc}} \gets \frac{1}{m}\sum_{i\in\mathcal{I}^{\mathrm{eval}}} S_{i\hat \jmath}$.
  \State Optionally: form CI for $\theta_{\hat \jmath}^{\mathrm{acc}}(\hat\eta_{\hat \jmath})$ using plug-in estimators under $(\hat\eta_{\hat\jmath}, \hat \jmath)$-conditional law.
\EndFor

\Statex \textbf{Output:} $\{(j,\hat\theta_{j},\mathcal{C}_{j}):j\in [p]\}$ and $\{(\hat \jmath,\hat\eta_{\hat \jmath},\hat\theta_{\hat \jmath}^{\mathrm{acc}}):\hat \jmath\in\hat J\}$.
\end{algorithmic}}
\end{algorithm}

\subsection{\cite{romano_control_2007} Algorithms}\label{sec_rw21}

I reproduce Algorithms 2.1 and 2.2 in \cite{romano_control_2007} in full here for completeness, using the notation of this paper.

\begin{algorithm}[H]
\caption{\cite{romano_control_2007} Algorithm 2.1}\label{algo2.1}
\begin{algorithmic}[1]
    \Statex 1. Let $A_1=\{1, \ldots, p\}$. If $\max \left(T_{n, j}: j \in A_1\right) \leq \hat{c}_{n, A_1}(1-\alpha, k)$, then accept all hypotheses and stop; otherwise, reject any $H_{0,j}$ for which $T_{n, j}>\hat{c}_{n, A_1}(1- \alpha, k)$ and continue.
    
    \Statex 2. Let $R_2$ be the indices $j$ of hypotheses $H_{0,j}$ previously rejected, and let $A_2$ be the indices of the remaining hypotheses. If $\left|R_2\right|<k$, then stop. Otherwise, let
    $$
    \hat{d}_{n, A_2}(1-\alpha, k)=\max _{I \subset R_2,|I|=k-1}\left\{\hat{c}_{n, K}(1-\alpha, k): K=A_2 \cup I\right\} .
    $$
    Then, reject any $H_{0,j}$ with $j \in A_2$ satisfying $T_{n, j}>\hat{d}_{n, A_2}(1-\alpha, k)$. If there are no further rejections, stop.
    
    \Statex $\quad\vdots$

    \Statex $\ell$. Let $R_\ell$ be the indices $j$ of hypotheses $H_{0,j}$ previously rejected, and let $A_\ell$ be the indices of the remaining hypotheses. Let

    $$
    \hat{d}_{n, A_\ell}(1-\alpha, k)=\max _{I \subset R_\ell,|I|=k-1}\left\{\hat{c}_{n, K}(1-\alpha, k): K=A_\ell \cup I\right\} .
    $$

    \Statex Then, reject any $H_{0,j}$ with $j \in A_\ell$ satisfying $T_{n, j}>\hat{d}_{n, A_\ell}(1-\alpha, k)$. If there are no further rejections, stop.
    
    \Statex $\quad\vdots$
    
    \Statex And so on.
\end{algorithmic}   
\end{algorithm}

\begin{algorithm}[H]
\caption{\cite{romano_control_2007} Algorithm 2.2}\label{algo2.2}
\begin{algorithmic}[1]
    \Statex The algorithm is analogous to Algorithm 2.1. The only difference is that in any step $\ell>1$ the critical value
    $$
    \hat{d}_{n, A_\ell}(1-\alpha, k)=\max _{I \subset R_\ell,|I|=k-1}\left\{\hat{c}_{n, K}(1-\alpha, k): K=A_\ell \cup I\right\}
    $$
    is replaced by the critical value
    $$
    \begin{aligned}
    & \tilde{d}_{n, A_\ell}(1-\alpha, k)=\hat{c}_{n, K}(1-\alpha, k), \\
    & \quad \text { where } K=\left\{r_{\left(\left|R_\ell\right|-k+2\right)}, r_{\left(\left|R_\ell\right|-k+1\right)}, \ldots, r_{(p)}\right\} .
    \end{aligned}
    $$
\end{algorithmic}   
\end{algorithm}

\subsection{Additional Results}\label{sec_hdclts}

\begin{restatable}[High-dimensional bootstrap for the small $k$-max coordinate of approximate means]{theorem}{Hdapproxboot}\label{lem_hdkapproxboot}
Let $\tilde S_n^B := S_n^B + R_n$, and assume that $
\| R_n \|_\infty = o_P(1/\sqrt{\log (pn)})$. Further assume that $k$ is fixed (i.e., does not grow with $n,p$). If Assumptions \ref{ass_mom} and \ref{ass_rate} hold, then as $n,p\to\infty$
\[
\sup _{t \in \mathbb{R}}\left|P^B\left( \tilde S_{n,[k]}^B \leq t \right)-P\left(N(0,\Sigma)_{[k]} \leq t\right)\right| \xrightarrow[]{P} 0.
\]
\end{restatable}

Define the diagonal matrix of asymptotic variances $\Lambda := \text{diag}(\Sigma)$, as well as the corresponding correlation matrix $\Sigma_0 := \Lambda^{-1/2} \Sigma \Lambda^{-1/2}$. Let plug-in estimates be $\widehat{\Sigma}_{j j}:=n^{-1} \sum_{i=1}^n\left(X_{i j}-\bar{X}_{n, j}\right)^2$ and $\widehat{\Lambda}:=\operatorname{diag}\left\{\widehat{\Sigma}_{11}, \ldots, \widehat{\Sigma}_{p p}\right\}$.

\begin{restatable}[High-dimensional CLT for the  small $k$-max studentized coordinate]{theorem}{Corhd}\label{cor_hdkapprox}
Assume $X_i \overset{\text{iid}}{\sim} P$. If Assumptions \ref{ass_mom} and \ref{ass_rate} hold, $B_n = O(1)$, and $k$ is fixed (i.e., does not grow with $n,p$), then, by application of Theorem \ref{lem_hdkapprox}, as $n,p\to\infty$
\[
\sup _{t \in \mathbb{R}}\left|P\left( \left(\hat\Lambda^{-1/2}S_n\right)_{[k]} \leq t\right)-P\left(N(0,\Sigma_0)_{[k]} \leq t\right)\right| \to 0.
\]
\end{restatable} 

\begin{restatable}[High-dimensional bootstrap for the  small $k$-max studentized coordinate]{theorem}{Corhdboot}\label{cor_hdkapproxboot}
Assume $X_i \overset{\text{iid}}{\sim} P$. If Assumptions \ref{ass_mom} and \ref{ass_rate} hold, $B_n = O(1)$, and $k$ is fixed (i.e., does not grow with $n,p$), then, by application of Theorem \ref{lem_hdkapproxboot}, as $n,p\to\infty$
\[
\sup _{t \in \mathbb{R}}\left|P^B\left( \left(\hat\Lambda^{-1/2}S_n^B\right)_{[k]} \leq t \right)-P\left(N(0,\Sigma_0)_{[k]} \leq t\right)\right| \xrightarrow[]{P} 0.
\]
\end{restatable} 

\begin{corollary}[Generalized simultaneous confidence intervals]\label{cor_gsci}
Under the conditions of the two-sided analog of Theorem \ref{thm_stu_kfwer},
let $\hat c := \hat{c}_{n,[p]}(1-\alpha,k)$ be the single-step critical
value, i.e., the $1-\alpha$ quantile of
$|\hat\Lambda^{-1/2}S_n^B|_{[k]}$ under $P^B$, and define, for $j\in[p]$,
\[
\mathcal{C}_j := \left[\hat\theta_j - \hat c\,\sqrt{\hat\Sigma_{jj}/n},\;
\hat\theta_j + \hat c\,\sqrt{\hat\Sigma_{jj}/n}\right].
\]
Then
\[
\liminf_{n,p\to\infty}
P\Big(\big|\{j\in[p]:\theta_j(P)\notin \mathcal{C}_j\}\big| < k\Big)\ \ge\ 1-\alpha.
\]
\end{corollary}
\begin{proof}
The event that $k$ or more parameters fall outside their intervals is exactly
the event $\big(|\hat\Lambda^{-1/2}S_n|\big)_{[k]} > \hat c$, where
$S_{n,j}=\sqrt{n}(\hat\theta_j-\theta_j(P))$. The proof of (the two-sided
analog of) Theorem \ref{thm_stu_kfwer}, applied with $K=[p]$, establishes
$\limsup_{n,p\to\infty} P\big((|\hat\Lambda^{-1/2}S_n|)_{[k]} > \hat
c_{n,[p]}(1-\alpha,k)\big) \le \alpha$, which is the claim.
\end{proof}

\subsection{Many Linear Regressions}\label{sec_manylinprop}

I now state two propositions associated with Example \ref{rmk_manylinear}.

\begin{restatable}[Influence function approximation for many linear regressions]{proposition}{Hdreg}\label{thm_hdreg}

Let $\{(Y_i,T_i,D_i)\}_{i=1}^n$ be drawn i.i.d. from distribution $P$, where
$Y_i=(Y_{i1},\ldots,Y_{ip})^{\mathtt T}\in\mathbb R^p$, $T_i\in\mathbb R$, and $D_i\in\mathbb R^d$ with fixed $d$, and $D_i$ includes a constant.
Let $L_P[\cdot\mid\cdot]$ denote the linear projection operator.
Define population quantities
\[
\begin{gathered}
\widetilde T := T - L_P[T\mid D], \qquad
\widetilde Y_j := Y_j - L_P[Y_j\mid D], \qquad
\Omega := E_P[\widetilde T^2], \qquad
\theta_j := \Omega^{-1}E_P[\widetilde T\,\widetilde Y_j],\\
\widetilde U_j := \widetilde Y_j - \theta_j \widetilde T, \qquad
\psi_{ij} := \Omega^{-1}\widetilde T_i \widetilde U_{ij}.
\end{gathered}
\]
Define feasible sample analog quantities as
\[
\begin{gathered}
\widehat{\gamma}_T:=E_n[D_i D_i^{\mathtt{T}}]^{-1}E_n[ D_i T_i],\qquad \widehat{\gamma}_{Y_j}:=E_n[D_i D_i^{\mathtt{T}}]^{-1}E_n[D_i Y_{i j}],\\
\widehat{\widetilde{T}}_i:=T_i-D_i^{\mathtt{T}} \widehat{\gamma}_T, \qquad
\widehat{\widetilde{Y}}_{ij}:=Y_{ij}-D_i^{\mathtt{T}} \widehat{\gamma}_{Y_j},\\
\widehat\Omega:=E_n[\widehat{\widetilde T}_i^2], \qquad
\widehat\theta_j:=\widehat\Omega^{-1}E_n[\widehat{\widetilde T}_i\,\widehat{\widetilde Y}_{ij}],\qquad
\widehat{\widetilde U}_{ij}:=\widehat{\widetilde Y}_{ij}-\widehat\theta_j\widehat{\widetilde T}_i, \qquad
\widehat\psi_{ij}:=\widehat\Omega^{-1}\widehat{\widetilde T}_i \widehat{\widetilde U}_{ij}.
\end{gathered}
\]
Assume:
\begin{enumerate}
\item[(ND)] $\Sigma_D:=E_P[DD^{\mathtt T}]$ is nonsingular and $0<C_1\le \Omega\le C_2<\infty$.
\item[(SG)] There exists $C_3<\infty$ such that $\|T\|_{\psi_2}\le C_3$, $\max_{\ell\le d}\|D_\ell\|_{\psi_2}\le C_3$, $\max_{j\le p}\|Y_j\|_{\psi_2}\le C_3$.
\end{enumerate}
Then, if Assumption \ref{ass_rate} holds,
\[
\max_{j\in[p]}\sqrt{E_n\!\left[(\widehat\psi_{ij}-\psi_{ij})^2\right]}
= o_P\!\left(\frac1{\log(pn)}\right).
\]
\end{restatable}

\begin{restatable}[Linearization error for many linear regressions]{proposition}{Hdlin}\label{thm_hdlin}
Assume the same setup as Proposition \ref{thm_hdreg}. Let
\[
\sqrt{n}(\hat\theta_j - \theta_j) = \frac{1}{\sqrt{n}}\sum_{i=1}^n \psi_{ij} + R_{n,j}
\]
where
\[
R_{n,j}
=
(\widehat\Omega^{-1}-\Omega^{-1})\sqrt{n}E_n[\widetilde T\,\widetilde U_j]
+\widehat\Omega^{-1}\sqrt{n}E_n[(\widehat{\widetilde T}-\widetilde T)\widetilde U_j].
\]
Let \(R_n:=(R_{n,1},\ldots,R_{n,p})^{\mathtt T}\). Then, under Assumption \ref{ass_rate},
\[
\|R_n\|_\infty = o_P\!\left(1/\sqrt{\log (pn)}\right).
\]
\end{restatable}

\subsection{Proofs}\label{sec_app}

Before stating proofs of the results in the main text, I state several additional useful lemmas.

\begin{lemma}\label{lem_op}
For a sequence of random variables $U_n$ and for a deterministic sequence $r_n$, 
\[ U_n = o_P(r_n) \Longleftrightarrow P^B(|U_n/r_n| \geq \varepsilon) = o_P(1) \text{ for any } \varepsilon > 0.\]
\end{lemma}
\begin{proof}
Note that, by definition, we have that $U_n = o_P(r_n)$ if for any $\varepsilon > 0$ that
\[
\lim_{n\to\infty} P(|U_n/r_n| \geq \varepsilon) = \lim_{n\to\infty} E[Z_{n,\varepsilon}] = 0
\]
where $Z_{n,\varepsilon} := P(|U_n/r_n| \geq \varepsilon \mid X)$, as by the law of total expectation $ P(|U_n/r_n| \geq \varepsilon) = E[Z_{n,\varepsilon}]$. By Markov's inequality, because $Z_{n,\varepsilon}$ is always positive, for any $\delta > 0$,
\[
P(Z_{n,\varepsilon} \geq \delta) \leq \frac{E[Z_{n,\varepsilon}]}{\delta}.
\]
Thus we have that
\[
0 \leq \lim_{n\to\infty} P(Z_{n,\varepsilon} \geq \delta) \leq \frac{\lim_{n\to\infty} E[Z_{n,\varepsilon}]}{\delta} = 0
\]
and $\lim_{n\to\infty} P(|Z_{n,\varepsilon}| \geq \delta) = 0$ for any $\varepsilon, \delta > 0$. As such, $P(|U_n/r_n| \geq \varepsilon \mid X) = P^B(|U_n/r_n| \geq \varepsilon) = o_P(1)$ for any $\varepsilon > 0$ if $U_n = o_P(r_n)$.

For the other direction, note further that if $P(|U_n/r_n| \geq \varepsilon \mid X) = o_P(1)$ then the boundedness of $Z_{n,\varepsilon}$ permits using the bounded convergence theorem\footnote{See, e.g., Lemma 10.7 at this link: \url{https://people.math.wisc.edu/~roch/teaching_files/275b.1.12w/lect10-web.pdf}.} to show for any $\varepsilon > 0$
\[
P(|U_n/r_n| \geq \varepsilon) = o(1).
\]
\end{proof}

\begin{lemma}\label{lem_conp}
If for $\delta_n \searrow 0, \nu_n \searrow 0$, and sequence of functions $d_n : \mathcal{X}_n \to \mathbb{R}$, for $X_n \in \mathcal{X}_n$
\[
P(d_n(X_n) \geq \delta_n) \leq \nu_n
\]
then
\[
d_n(X_n) = o_P(1).
\]
\end{lemma}
\begin{proof}
We want to show that, for any fixed $\delta > 0$, 
\[
\lim_{n\to\infty} P( d_n(X_n)\geq \delta) = 0
\]
(the definition of convergence in probability). Note that, given the definition $\delta_n = o(1)$, for some $n^*$, for all $n \geq n^*$ we have that $\delta _n < \delta$. Thus, define the events:
\[
A_n := \{d_n(X_n) \geq \delta \}, \quad B_n:= \{d_n(X_n) \geq \delta_n \}
\]
For all $n \geq n^*$, $A_n \subseteq B_n$. Thus, for all $n \geq n^*$
\[
P\{d_n(X_n) \geq \delta \} \leq P\{d_n(X_n) \geq \delta_n \}.
\]

Notice that then, using that $\nu_n = o(1)$, 
\[
0\leq \lim_{n\to\infty} P\{d_n(X_n) \geq \delta \} \leq \lim_{n\to\infty}P\{d_n(X_n) \geq \delta_n \}=0.
\]
Because this holds for any choice of $\delta > 0$, we have proven the stated result.
\end{proof}

\begin{lemma}\label{lem_bernstein}
For $\{X_{i}\}_{i\in [n]}$ independent subexponential random vectors of dimension $p$, with each $\|X_{ij}-E[X_{ij}]\|_{\psi_1} \leq B_n$ for $i \in [n]$ and $j \in [p]$, then, under Assumption \ref{ass_rate},
\[
\left\| S_n \right\|_\infty = \left\| \frac{1}{\sqrt{n}}\sum_{i=1}^n (X_i - E[X_i]) \right\|_\infty = O_P\left(B_n \sqrt{\log p }\right).
\]
\end{lemma}
\begin{proof}
We start by noting Proposition 2.9.2 of \cite{vershynin_high-dimensional_2018}, which implies here that
\[
P\left\{\left|\frac{1}{n}\sum_{i=1}^n (X_{ij} - E[X_{ij}])\right| \geq t\right\} \leq 2 \exp \left[-c \min \left(\frac{nt^2}{B_n^2}, \frac{nt}{B_n}\right)\right].
\]
Thus by the union bound,
\[
P\left\{\max_{j\in [p]} \left|\frac{1}{n}\sum_{i=1}^n (X_{ij} - E[X_{ij}])\right| \geq t\right\} \leq 2p \exp \left[-c \min \left(\frac{nt^2}{B_n^2}, \frac{nt}{B_n}\right)\right].
\]
Defining that
\[
\varepsilon := 2p \exp \left[-c \min \left(\frac{nt^2}{B_n^2}, \frac{nt}{B_n}\right)\right]
\]
then we have that for $\tilde c = 1/c$
\[
\tilde c \frac{\log(2p/\varepsilon)}{n} = \min\left\{ t^2/B_n^2, t/B_n \right\}
\]
and therefore
\[
t =  \tilde c B_n\frac{\log(2p/\varepsilon)}{n} \vee \tilde c B_n \sqrt{\frac{\log(2p/\varepsilon)}{n}} ,
\]
meaning that 
\[
P\left\{\max_{j\in [p]} \left|\frac{1}{n}\sum_{i=1}^n (X_{ij} - E[X_{ij}])\right| \geq \tilde c B_n\left(\frac{\log(2p/\varepsilon)}{n} \vee  \sqrt{\frac{\log(2p/\varepsilon)}{n}}\right)\right\} \leq \varepsilon.
\]
As such, we have
\[
P\left\{\max_{j\in [p]} \left|\frac{1}{\sqrt{n}}\sum_{i=1}^n (X_{ij} - E[X_{ij}])\right| \geq \tilde c B_n\left(\frac{\log(2p/\varepsilon)}{\sqrt{n}} \vee  \sqrt{\log(2p/\varepsilon)}\right)\right\} \leq \varepsilon.
\]

However, under Assumption \ref{ass_rate}, $\log(p)/\sqrt{n} = o(1)$, so the Gaussian tail term dominates for large $n$, and thus for large $n$ we may set $M_\varepsilon := \tilde c   \sqrt{\frac{\log(2/\varepsilon)}{\log 3}+1}$ and observe that
\[
P\left\{\frac{\max_{j\in [p]} \left|\frac{1}{\sqrt{n}}\sum_{i=1}^n (X_{ij} - E[X_{ij}])\right|}{B_n \sqrt{\log p}} \geq M_\varepsilon\right\} \leq \varepsilon.
\]
We then conclude using the definition of stochastic boundedness that
\[
\left\| \frac{1}{\sqrt{n}}\sum_{i=1}^n (X_i - E[X_i]) \right\|_\infty = O_P\left(B_n  \sqrt{\log p}\right)
\]
as claimed.
\end{proof}

\begin{lemma}\label{lem_sbb}
Under Assumptions \ref{ass_mom} and \ref{ass_rate},
\[
\|S_n^B\|_\infty = O_P(B_n \sqrt{\log p}).
\]
\end{lemma}
\begin{proof}
Using Assumptions \ref{ass_mom} and \ref{ass_rate}, the proof of Lemma \ref{lem_bernstein} implies that
\[
P\left\{\| S_{n} \|_\infty \geq \tilde c B_n\left(\frac{\log(2p/\varepsilon)}{\sqrt{n}} \vee \sqrt{\log(2p/\varepsilon)}\right)\right\} \leq \varepsilon.
\]
Using Lemma 4.6 of \cite{chernozhuokov_improved_2022} and Lemma \ref{lem_conp}, notice that 
\begin{equation}\label{eqn_tail}
b_n(X):=P\left(\|S_n^B\|_\infty > \tilde c B_n\left(\frac{\log(2p/\varepsilon)}{\sqrt{n}} \vee \sqrt{\log(2p/\varepsilon)}\right) \mid X\right) \leq \varepsilon + o_P(1).
\end{equation}

Define the event $A_n:=\left\{b_n(X) > \varepsilon + \eta \right\}$ for any $\eta > 0$.
Then the earlier statement implies that 
$P(A_n) = o(1)$.
Now, by the law of total expectation and the law of total probability,
\begin{align*}
    &P\left(\|S_n^B\|_\infty > \tilde c B_n\left(\frac{\log(2p/\varepsilon)}{\sqrt{n}} \vee \sqrt{\log(2p/\varepsilon)}\right) \right) \\
    &= E[b_n(X)] \\
    &= E[b_n(X) \mid A_n] P(A_n) + E[b_n(X) \mid A_n^c] P(A_n^c) \\
    &= E[b_n(X) \mid A_n^c] + o(1) \tag{Set $\eta = o(1)$} \\
    &\leq \varepsilon + o(1).
\end{align*}
It then follows, using the same logic as in Lemma \ref{lem_bernstein}, along with Assumption \ref{ass_rate}, that
\[
\|S_n^B\|_\infty = O_P(B_n \sqrt{\log p}).
\]
\end{proof}

\begin{lemma}\label{lem_sbb_if}
Suppose Assumptions \ref{ass_mom} (treating the influence functions
$\psi_{ij}$ as data), \ref{ass_rate}, and \ref{ass_approx} hold. Then $$\|\tilde
S_n^B\|_\infty = O_P(B_n\sqrt{\log p}).$$ 
\end{lemma}
\begin{proof}
First, we will decompose $\tilde S_n^B = G_n + \Delta_n$, where $G_n := n^{-1/2}\sum_{i=1}^n
\xi_i\psi_i$ and $\Delta_n := n^{-1/2}\sum_{i=1}^n \xi_i(\hat\psi_i-\psi_i)$.
 
Conditional on $X$, $G_{n,j}$ is centered Gaussian with
variance $\hat\sigma_j^2 := n^{-1}\sum_{i=1}^n \psi_{ij}^2$. By Assumption
\ref{ass_mom}(iii) and Jensen's inequality, $n^{-1}\sum_{i=1}^n
E[\psi_{ij}^2] \le (n^{-1}\sum_{i=1}^n E[\psi_{ij}^4])^{1/2}\le B_n b_2$.
Moreover, $\psi_{ij}^2$ is sub-Weibull of order $1/2$ with
$\|\psi_{ij}^2\|_{\psi_{1/2}}\lesssim B_n^2$, so Theorem 3.4 of
\cite{kuchibhotla_moving_2022} together with a union bound over $j\in[p]$
gives
\[
\max_{j\in[p]} \Big|\hat\sigma_j^2 - n^{-1}\sum_{i=1}^n E[\psi_{ij}^2]\Big| =
O_P\!\left(B_n\sqrt{\frac{\log p}{n}} +
\frac{B_n^2\log^2(2n)\log^2(p)}{n}\right) = o_P(B_n^2)
\]
where the last equality follows from Assumption \ref{ass_rate}, and so $\max_{j\in[p]} \hat\sigma_j =
O_P(B_n)$ (recalling $B_n>1$). The Borell-TIS inequality (as per \citep{belloni_high-dimensional_2018})---applied as in the proof of Theorem
\ref{thm_hdbootinfluence}---yields,
for any $\varepsilon\in(0,1)$, almost surely,
\[
P^B\Big(\|G_n\|_\infty > \max_{j\in[p]}\hat\sigma_j\,
\big(\sqrt{2\log(2p)}+\sqrt{2\log(1/\varepsilon)}\big)\Big)\le 2\varepsilon,
\]
and applying the un-conditioning argument exactly as in the proof of Lemma \ref{lem_sbb} (via the
law of total probability on the event $\{\max_j\hat\sigma_j\le M B_n\}$ for some $M < \infty$,
whose probability can be made arbitrarily close to one) delivers
$\|G_n\|_\infty = O_P(B_n\sqrt{\log p})$.
 
Then, conditional on $X$, note that $\Delta_{n,j}$ is a centered Gaussian
with variance $n^{-1}\sum_{i=1}^n(\hat\psi_{ij}-\psi_{ij})^2$, so by the same inequality as above, with $P^B$-probability at least $1-2\varepsilon$,
\[
\|\Delta_n\|_\infty \le
\max_{j\in[p]}\Big(n^{-1}\sum_{i=1}^n(\hat\psi_{ij}-\psi_{ij})^2\Big)^{1/2}
\big(\sqrt{2\log(2p)}+\sqrt{2\log(1/\varepsilon)}\big).
\]
Choosing $\varepsilon=1/(2n)$ and applying Assumption \ref{ass_approx}, the
right-hand side is $o_P(1/\log(pn)) O(\sqrt{\log(pn)}) = o_P(1)$;
un-conditioning again (as above) gives $\|\Delta_n\|_\infty = o_P(1)$.
 
Combining these results, it is immediate that $\|\tilde S_n^B\|_\infty = O_P(B_n\sqrt{\log p})$. 

\end{proof}

We now prove the results from the main text. For any set of real numbers $A$ with $|A|>k$, we use the notation $\kmax (A)$ to denote the $k$-th largest element of $A$, following \cite{romano_control_2007}.

\Smallk*
\begin{proof}
To start, we may assume $|I(P)|\ge k$: otherwise at most $|I(P)|<k$ true
null hypotheses can be falsely rejected, so that $k\text{-FWER}_P = 0 \le
\alpha$ holds trivially, and parts (ii) and (iii) are unaffected. (This will also be implicitly assumed in subsequent proofs of the \cite{romano_control_2007} algorithms' validity in high-dimensions.)

Now consider Algorithm 2.1 from \cite{romano_control_2007}, which at each stage takes as inputs critical values $\hat{c}_{n, K}(1-\alpha,k)$ and a vector of test statistics $T_n$; we denote this as $\texttt{RW-2.1}(T_n, \hat{c}_{n, K}(1-\alpha,k))$.

The choice of critical values we will use is based on a bootstrap of the $1-\alpha$ quantile of the $\kmax$:
\[
\hat{c}_{n, K}(1-\alpha, k) := 1-\alpha \text{ quantile of } S_{n,K,[k]}^B \mid X.
\]
As in \cite{romano_control_2007}, define $I(P)$ to be the set of indices of true null hypotheses under $P$. Note that, for any $K \supset I(P)$, 
$$
\hat{c}_{n, K}(1-\alpha, k) \geq \hat{c}_{n, I(P)}(1-\alpha, k)
$$
because for any $K \supset I(P)$, and under any distribution, $ S_{n,K,[k]}^B \geq  S_{n,I(P),[k]}^B$ almost surely (i.e., the $\kmax$ statistic can only get larger if we include more test statistics without dropping the others).  
As such, Theorem 2.1 (i) holds, and we conclude that $\texttt{RW-2.1}(T_n, \hat{c}_{n, K}(1-\alpha,k))$ delivers:
\[
k\text{-FWER}_P \leq P\left\{\kmax \left(T_{n,j}: j \in I(P)\right)>\hat{c}_{n, I(P)}(1-\alpha, k)\right\}.
\]
Specifically, then, it is sufficient to show that
\[
\limsup_{n,p} P\left\{\kmax \left(T_{n,j}: j \in I(P)\right)> \hat{c}_{n, I(P)}(1-\alpha, k)\right\} \leq \alpha.
\]
To continue to use the notation of \cite{romano_control_2007}, let $\hat\theta_{n,j}:=\bar X_{n,j}$. Since $\theta_j(P) \leq 0$ for $j \in I(P)$, it follows that, almost surely,

$$
\begin{aligned}
\kmax \left(T_{n,j}: j \in I(P)\right) & =\kmax \left(\sqrt{n} \hat{\theta}_{n, j}: j \in I(P)\right) \\
& \leq \kmax \left(\sqrt{n}\left[\hat{\theta}_{n, j}-\theta_j(P)\right]: j \in I(P)\right)\\
& = \kmax \left(S_{n,j}: j \in I(P)\right)
\end{aligned}
$$
and therefore
\begin{align*}
    &P\left\{\kmax \left(T_{n,j}: j \in I(P)\right)> \hat{c}_{n, I(P)}(1-\alpha, k)\right\} \\
    \leq &P\left\{\kmax \left(S_{n,j}: j \in I(P)\right)>\hat{c}_{n, I(P)}(1-\alpha, k)\right\} .
\end{align*}

If we can show that the limit of the quantity on the right-hand side is no greater than $\alpha$, the proof is complete. However, Theorem 2.2 of \cite{ding_gaussian_2025} delivers exactly that, for any $K$,
\[
P\left\{S_{n,K,[k]}> \hat{c}_{n,K}\left(1-\alpha, k\right)\right\} \leq \alpha + o(1)
\]
under Assumptions \ref{ass_mom} and \ref{ass_rate}. So, chaining inequalities, we note that
\begin{align*}
     k\text{-FWER}_P \leq P\left\{S_{n,I(P),[k]}> \hat{c}_{n,I(P)}\left(1-\alpha, k\right)\right\}
    \leq \alpha +o(1)
\end{align*}
and so
\begin{align*}
    \limsup_{n,p\to\infty}
     k\text{-FWER}_P 
    \leq \alpha 
\end{align*}
which is precisely what we wanted to show.

To prove the second statement, consider the $\tilde H_{0,j}$ corresponding to all $\theta_j(P)>0$. Note that, for the Gaussian multiplier bootstrap, using Assumptions \ref{ass_mom} and \ref{ass_rate}, by the tail inequality Equation \ref{eqn_tail} in the proof of Lemma \ref{lem_sbb},
\[
\hat c_{n, [p]}\left(1-\alpha, k\right) \leq \hat c_{n, [p]}\left(1-\alpha, 1\right)  = O_P\left(B_n \sqrt{\log p} \right).
\]
Furthermore, each $S_{n,j} = \sqrt{n} \left[\hat{\theta}_{n,j}-\theta_j(P)\right]$ is tight, so 
$$
T_{n, j}=\sqrt{n} \hat{\theta}_{n,j} =\sqrt{n} (\hat{\theta}_{n,j}-\theta_j(P)) + \sqrt{n} \theta_j(P)\xrightarrow{P} \infty.
$$
However, because $\sqrt{n}$ grows faster than $B_n \sqrt{\log p}$ by Assumption \ref{ass_rate}, we have that also 
$$
\frac{T_{n, j}}{B_n \sqrt{\log p} \,\,} \xrightarrow{P} \infty.
$$
Therefore, with probability tending to one, $T_{n, j}>\hat{c}_{n,[p]}\left(1-\alpha, k\right)$,  resulting in the rejection of $\tilde H_{0,j}$ in the first step of Algorithm 2.1, so long as $\theta_j(P)$ is fixed or approaches zero slower than a rate of $B_n \sqrt{\log p/n}$.

To prove the third statement, the asymptotic validity of the streamlined Algorithm 2.2 follows from the separation assumption 
\[
\min_{j \notin I(P)}\theta_j(P) \gg B_n \sqrt{\log p /n}
\]
using the same logic as for the proof of the second statement and the argument of proof of Theorem 3.3 of \cite{romano_control_2007}, i.e., $\min \left(T_{n, j}: j \notin I(P)\right)$ is diverging at rate $\sqrt{n}$, and if any $\theta_j(P)=0$ then $
\max \left(T_{n, j}: j \in I(P)\right) = O_P(B_n \sqrt{\log p})
$.
\end{proof}

\Smallktwo*
\begin{proof}
We again start by considering Algorithm 2.1 from \cite{romano_control_2007}, which at each stage takes as inputs critical values $\hat{c}_{n, K}(1-\alpha,k)$ and a vector of test statistics $|T_n|$; we denote this as $\texttt{RW-2.1}(|T_n|, \hat{c}_{n, K}(1-\alpha,k))$.

The choice of critical values we will use is based on a bootstrap of the $1-\alpha$ quantile of the $\kmax$ of absolute values:
\[
\hat{c}_{n, K}(1-\alpha, k) := 1-\alpha \text{ quantile of } |S_{n,K}^B|_{[k]} \mid X.
\]
As in \cite{romano_control_2007}, define $I(P)$ to be the set of indices of true null hypotheses under $P$. Note that, for any $K \supset I(P)$, 
$$
\hat{c}_{n, K}(1-\alpha, k) \geq \hat{c}_{n, I(P)}(1-\alpha, k)
$$
because for any $K \supset I(P)$, and under any distribution, $ |S_{n,K}^B|_{[k]} \geq  |S_{n,I(P)}^B|_{[k]}$ almost surely (i.e., the $\kmax$ statistic can only get larger if we include more test statistics without dropping the others).  
As such, Theorem 2.1 (i) holds, and we conclude that $\texttt{RW-2.1}(|T_n|, \hat{c}_{n, K}(1-\alpha,k))$ delivers:
\[
k\text{-FWER}_P \leq P\left\{\kmax \left(|T_{n,j}|: j \in I(P)\right)>\hat{c}_{n, I(P)}(1-\alpha, k)\right\}.
\]
Specifically, then, it is sufficient to show that
\[
\limsup_{n,p} P\left\{\kmax \left(|T_{n,j}|: j \in I(P)\right)> \hat{c}_{n, I(P)}(1-\alpha, k)\right\} \leq \alpha.
\]
To continue to use the notation of \cite{romano_control_2007}, let $\hat\theta_{n,j}:=\bar X_{n,j}$. Since $\theta_j(P) = 0$ for $j \in I(P)$, it follows that, almost surely,

$$
\begin{aligned}
\kmax \left(|T_{n,j}|: j \in I(P)\right) & =\kmax \left(\sqrt{n} |\hat{\theta}_{n, j}|: j \in I(P)\right) \\
& = \kmax \left(\sqrt{n}\left|\hat{\theta}_{n, j}-\theta_j(P)\right|: j \in I(P)\right)\\
& = \kmax \left(|S_{n,j}|: j \in I(P)\right)
\end{aligned}
$$
and therefore
\begin{align*}
    &P\left\{\kmax \left(|T_{n,j}|: j \in I(P)\right)> \hat{c}_{n, I(P)}(1-\alpha, k)\right\} \\
    = &P\left\{\kmax \left(|S_{n,j}|: j \in I(P)\right)>\hat{c}_{n, I(P)}(1-\alpha, k)\right\} .
\end{align*}

If we can show that the limit of the quantity on the right-hand side is no greater than $\alpha$, the proof is complete. However, Theorem 2.2 and Remark 2 of \cite{ding_gaussian_2025} delivers exactly that, for any $K$,
\[
P\left\{|S_{n,K}|_{[k]}> \hat{c}_{n,K}\left(1-\alpha, k\right)\right\} \leq \alpha + o(1)
\]
under Assumptions \ref{ass_mom} and \ref{ass_rate}. So, chaining inequalities, we note that
\begin{align*}
     k\text{-FWER}_P \leq P\left\{|S_{n,I(P)}|_{[k]}> \hat{c}_{n,I(P)}\left(1-\alpha, k\right)\right\}
    \leq \alpha +o(1)
\end{align*}
and so
\begin{align*}
    \limsup_{n,p\to\infty}
     k\text{-FWER}_P 
    \leq \alpha 
\end{align*}
which is precisely what we wanted to show.

To prove the second statement, note that, for the Gaussian multiplier bootstrap, using Assumptions \ref{ass_mom} and \ref{ass_rate}, by the tail inequality Equation \ref{eqn_tail} in the proof of Lemma \ref{lem_sbb}, and using the insight from Remark 2 of \cite{ding_gaussian_2025} that quantiles for absolute values are identical to quantiles for a related $2p$ dimensional problem without absolute values,
\[
\hat c_{n, [p]}\left(1-\alpha, k\right) \leq \hat c_{n, [p]}\left(1-\alpha, 1\right)  = O_P\left(B_n \sqrt{\log p} \right).
\]
Now notice that each $S_{n,j} = \sqrt{n} \left[\hat{\theta}_{n,j}-\theta_j(P)\right]$ is tight, so consider that for the $ H_{0,j}$ corresponding to all $\theta_j(P) > 0$ that
$$
T_{n, j}=\sqrt{n} \hat{\theta}_{n,j} = \sqrt{n} (\hat{\theta}_{n,j}-\theta_j(P)) + \sqrt{n} \theta_j(P)\xrightarrow{P} \infty
$$
and that for the $H_{0,j}$ corresponding to all $\theta_j(P) < 0$ that
\[
T_{n, j}=\sqrt{n} \hat{\theta}_{n,j} = \sqrt{n} (\hat{\theta}_{n,j}-\theta_j(P)) + \sqrt{n} \theta_j(P)\xrightarrow{P} -\infty.
\]
However, because $\sqrt{n}$ grows faster than $B_n \sqrt{\log p}$ by Assumption \ref{ass_rate}, we have that also for all $ H_{0,j}$ for which $\theta_j(P)\neq 0$
$$
\frac{|T_{n, j}|}{B_n \sqrt{\log p} \,\,} \xrightarrow{P} \infty.
$$
Therefore, with probability tending to one, $|T_{n, j}|>\hat{c}_{n,[p]}\left(1-\alpha, k\right)$,  resulting in the rejection of $H_{0,j}$ in the first step of Algorithm 2.1, so long as $|\theta_j(P)|$ is fixed or approaches zero slower than a rate of $B_n \sqrt{\log p/n}$.

To prove the third statement, the asymptotic validity of the streamlined Algorithm 2.2 follows from the separation assumption 
\[
\min_{j \notin I(P)}|\theta_j(P)| \gg B_n \sqrt{\log p /n}
\]
using the same logic as for the proof of the second statement and the argument of proof of Theorem 3.3 of \cite{romano_control_2007}, i.e., $\min \left(|T_{n, j}|: j \notin I(P)\right)$ is diverging at rate $\sqrt{n}$, and if any $\theta_j(P)=0$ then $
\max \left(|T_{n, j}|: j \in I(P)\right) = O_P(B_n \sqrt{\log p})
$.
\end{proof}

Before proceeding to proofs of the high-dimensional CLTs, define the following deterministic sequences:
\[
\nu_n := 2n^{-1} + 3\sqrt{\frac{B_n^2 \log ^3(p n)}{n}}, \qquad \beta_n :=  k^2\left(\frac{B_n^2 \log ^5(p n)}{n}\right)^{1 / 4}
\]
where both $\nu_n = o(1)$ and $\beta_n = o(1)$ under Assumption \ref{ass_rate}.

\Hdapprox*
\begin{proof}
The proof of this lemma proceeds following the strategy of \cite{chernozhukov_high-dimensional_2023}.

First, recall that the function $t \mapsto t_{[k]}$ is 1-Lipschitz wrt to the sup-norm, meaning almost surely
\[
\left|\tilde S_{n,[k]} - S_{n,[k]} \right| = \left|(S_{n}+R_n)_{[k]} - S_{n,[k]} \right| \leq \| R_n \|_\infty.
\]
Consider the event $\{ \| R_n \|_\infty \leq \epsilon \}$, as well as the event $\{(S_{n}+R_n)_{[k]} \leq t \}$. Then observe that
\[
\{ \| R_n \|_\infty \leq \epsilon \} \cap \{(S_{n}+R_n)_{[k]} \leq t \} \subseteq \{ S_{n,[k]} \leq t + \epsilon \}
\]
As such notice that
\[
\{(S_{n}+R_n)_{[k]} \leq t \} = (\{(S_{n}+R_n)_{[k]} \leq t \} \cap \{ \| R_n \|_\infty \leq \epsilon \}) \cup (\{(S_{n}+R_n)_{[k]} \leq t \} \cap \{ \| R_n \|_\infty > \epsilon \} )
\]
and so
\[
P(\tilde S_{n,[k]} \leq t ) = P((S_{n}+R_n)_{[k]} \leq t ) \leq P( S_{n,[k]} \leq t + \epsilon ) + P( \| R_n \|_\infty > \epsilon).
\]
We now need to apply an anti-concentration result, which can be found as embedded in Lemma A.7 of \cite{ding_gaussian_2025}. It implies, for any $t$, for a $C_1$ depending on only $b_1, b_2$,
\[
P\left(\tilde S_{n,[k]}\leq t\right)  \leq P\left(S_{n,[k]} \leq t\right)+ C_1 k \epsilon \sqrt{1 \vee \log (p / \epsilon)}+ 2C_1\beta_n+P\left(\left\|R_n \right\|_{\infty}>\epsilon\right).
\]

We then have, using Assumptions \ref{ass_mom} and \ref{ass_rate} in addition to the high-dimensional CLT for the $k$ largest coordinate of Lemma A.6 in \cite{ding_gaussian_2025} (which requires Assumption \ref{ass_mom}):
\begin{align*}
    P\left(\tilde S_{n,[k]}\leq t\right) 
    &\leq P\left(S_{n,[k]} \leq t\right)+ C_1 k \epsilon \sqrt{1 \vee \log (p / \epsilon)}+ 2C_1\beta_n+P\left(\left\|R_n \right\|_{\infty}>\epsilon\right) \\
    &\leq P\left(N(0,\Sigma)_{[k]} \leq t\right)+ C_1 k \epsilon \sqrt{1 \vee \log (p / \epsilon)}+C_2\beta_n+P\left(\left\|R_n\right\|_{\infty}>\epsilon\right).
\end{align*}


Since $\sqrt{\log(pn)}\,\|R_n\|_\infty = o_P(1)$, there exists
a deterministic sequence $\delta_n \searrow 0$ such that
$P\big(\sqrt{\log(pn)}\,\|R_n\|_\infty > \delta_n\big)\to 0$. If we choose
\[
\epsilon=\epsilon_n := \max\Big\{\delta_n/\sqrt{\log(pn)},\ (pn)^{-1}\Big\},
\]
then $P(\|R_n\|_\infty > \epsilon_n)\to 0$, and since $\epsilon_n \ge
(pn)^{-1}$ implies $\log(p/\epsilon_n)\le \log(p^2 n)\le 2\log(pn)$, we have,
for fixed $k$,
\[
k\,\epsilon_n\sqrt{1\vee\log(p/\epsilon_n)}
\;\le\; \sqrt{2}\,k\,\max\Big\{\delta_n,\ \tfrac{\sqrt{\log(pn)}}{pn}\Big\}
\;\to\; 0 .
\]
It follows that
\[
P\left(\tilde S_{n,[k]}\leq t\right)  \leq P\left(N(0,\Sigma)_{[k]} \leq t\right) + \rho_n + C_2\beta_n
\]
where $\rho_n := C_1 k \epsilon_n \sqrt{1 \vee \log (p / \epsilon_n)}+
P\left(\left\|R_n \right\|_{\infty}>\epsilon_n\right) = o(1)$ is
deterministic.

From the reverse direction, note that
\[
\left\{\left\|R_n\right\|_{\infty} \leq \epsilon\right\} \cap\left\{S_{n,[k]} \leq t-\epsilon\right\} \subseteq\left\{\left(S_n+R_n\right)_{[k]} \leq t\right\} 
\]
so it also holds that similarly, partitioning $\{S_{n,[k]} \leq t - \epsilon \}$ using $\{ \| R_n \|_\infty \leq \epsilon \}$ and its complement, that
\begin{align*}
    &P\left(\tilde S_{n,[k]}\leq t\right)
    \geq P\left(S_{n,[k]} \leq t-\epsilon\right)-P\left(\left\|R_n\right\|_{\infty}>\epsilon\right),
\end{align*}
and so using identical arguments as above we conclude that, uniformly in $t$,
\[
\left|P\left( \tilde S_{n,[k]} \leq t\right)-P\left(N(0,\Sigma)_{[k]} \leq t\right)\right| \leq C_2\beta_n + \rho_n= o(1),
\]
proving the stated result.
\end{proof}

\Hdapproxboot*
\begin{proof}
The proof of this lemma proceeds following the strategy of Theorem \ref{lem_hdkapprox}. Using the same logic as the proof of Theorem \ref{lem_hdkapprox}, note that almost surely
\[
P^B(\tilde S_{n,[k]}^B \leq t) = P^B((S_{n}^B+R_n)_{[k]} \leq t ) \leq P^B( S_{n,[k]}^B \leq t + \epsilon) + P^B( \| R_n \|_\infty > \epsilon).
\]
Observe that, from \cite{ding_gaussian_2025} Lemma A.6 and Lemma A.8, using Assumptions \ref{ass_mom} and \ref{ass_rate} and the triangle inequality, with probability at least $1-\nu_n$
\[
\sup_{t \in \mathbb{R}} \left|P^B\left(S_{n,[k]}^B \leq t \right) - P\left(S_{n,[k]} \leq t\right)  \right| \leq C_3\beta_n,
\]
and thus with probability at least $1-\nu_n$
\[
P^B(\tilde S_{n,[k]}^B \leq t) = P^B((S_{n}^B+R_n)_{[k]} \leq t ) \leq P\left(S_{n,[k]} \leq t + \epsilon\right) + C_3\beta_n + P^B( \| R_n \|_\infty > \epsilon).
\]
Then again using Lemma A.7 and Lemma A.6 of \cite{ding_gaussian_2025} we get that with probability at least $1-\nu_n$
\[
P^B(\tilde S_{n,[k]}^B \leq t) \leq P\left(N(0,\Sigma)_{[k]} \leq t\right) + C_4\beta_n + \rho_n^B
\]
where $\rho_n^B :=  C_1 k \epsilon \sqrt{1 \vee \log (p / \epsilon)}+ P^B\left(\left\|R_n \right\|_{\infty}>\epsilon\right)$. The reverse direction proceeds similarly just as in Theorem \ref{lem_hdkapprox}, yielding still with probability at least $1-\nu_n$ (as in the reverse direction the ``bad'' event is the same)
\[
|P^B(\tilde S_{n,[k]}^B \leq t) - P\left(N(0,\Sigma)_{[k]} \leq t\right)| \leq C_4\beta_n + \rho_n^B 
\]
uniformly in $t$. We may re-write this as, with probability at least $1-\nu_n$,
\[
|P^B(\tilde S_{n,[k]}^B \leq t) - P\left(N(0,\Sigma)_{[k]} \leq t\right)| - \rho_n^B \leq C_4\beta_n 
\]
and so using Lemma \ref{lem_conp}, uniformly in $t$,
\[
|P^B(\tilde S_{n,[k]}^B \leq t) - P\left(N(0,\Sigma)_{[k]} \leq t\right)| - \rho_n^B = o_P(1).
\]

By Lemma \ref{lem_op}, $\sqrt{\log(pn)}\,\|R_n\|_\infty = o_P(1)$ implies
$P^B\big(\sqrt{\log(pn)}\,\|R_n\|_\infty > u\big) = o_P(1)$ for every fixed
$u>0$; a standard diagonalization argument then yields a deterministic
sequence $\delta_n\searrow 0$ with
$P^B\big(\sqrt{\log(pn)}\,\|R_n\|_\infty > \delta_n\big) = o_P(1)$. Setting
$\epsilon=\epsilon_n := \max\{\delta_n/\sqrt{\log(pn)},\,(pn)^{-1}\}$ as in the proof
of Theorem \ref{lem_hdkapprox}, we obtain $\log(p/\epsilon_n)\le 2\log(pn)$
and hence $C_1 k \epsilon_n\sqrt{1\vee\log(p/\epsilon_n)} = o(1)$ for fixed
$k$, so that
\[
\rho_n^B = C_1 k \epsilon_n \sqrt{1 \vee \log (p / \epsilon_n)}
+ P^B\left(\left\|R_n \right\|_{\infty}>\epsilon_n\right) = o(1) + o_P(1) = o_P(1).
\]
This means that
\[|P^B(\tilde S_{n,[k]}^B \leq t) - P\left(N(0,\Sigma)_{[k]} \leq t\right)| = o_P(1)\]
uniformly in $t$, completing the proof.
\end{proof}

\Corhd*
\begin{proof}
Let $\tilde S_n := \hat\Lambda^{-1/2}S_n$. Then we may write $\tilde S_n = \Lambda^{-1/2} S_n + R_n$, and note that
\[
R_n = \hat \Lambda^{-1/2} S_n -  \Lambda^{-1/2} S_n = (\hat \Lambda^{-1/2} -  \Lambda^{-1/2}) S_n,
\]
and thus we can use the machinery of Theorem \ref{lem_hdkapprox} to prove the desired result so long as we can show that $\| R_n \|_\infty = o_P(1/\sqrt{\log (pn)})$.

Note that, by the sub-multiplicative induced matrix norm inequality,
\begin{align*}
    \left\|\left(\widehat{\Lambda}^{-1 / 2}-\Lambda^{-1 / 2}\right) S_n\right\|_{\infty} &\leq \|\widehat{\Lambda}^{-1 / 2}-\Lambda^{-1 / 2} \|_{\infty} \left\|S_n\right\|_{\infty} = \|\widehat{\Lambda}^{-1 / 2}-\Lambda^{-1 / 2} \|_{\max} \left\|S_n\right\|_{\infty}
\end{align*}
where the last equality follows from the fact that the max elementwise norm is equal to the induced operator $\infty$-norm for diagonal matrices.
  
To control the first term on the right-hand side, we may turn to  \cite{kuchibhotla_moving_2022}, Theorem 4.2, which shows that if $X_{ij}$ are sub-Weibull with parameter $\alpha = 1$ (i.e., subexponential, granted by Assumption \ref{ass_mom}), then Assumption \ref{ass_rate} ensures that
\[
\|\widehat{\Lambda}-\Lambda \|_{\max} = o_P(1/\log (pn)).
\]
To see this, note that Theorem 4.2 of \cite{kuchibhotla_moving_2022} (spelled out in Remark 4.1) implies that
$$\|\hat\Sigma - \Sigma\|_{\max} = O_P\left(\sqrt{\frac{\log p}{n}}\vee\frac{(\log p)^2(\log n)^2}{n}\right)$$ if $B_n = O(1)$ (where $\|\cdot \|_{\max}$ is the maximum elementwise norm), which we assume. As a consequence, Assumption \ref{ass_rate} also delivers that $\|\widehat{\Lambda}-\Lambda \|_{\max} \leq \|\hat\Sigma - \Sigma\|_{\max} = o_P\left(1/\log (pn)\right)$, and thus $\|\widehat{\Lambda}^{-1 / 2}-\Lambda^{-1 / 2} \|_{\max} = o_P\left(1/\log (pn)\right)$ as well using a Taylor expansion argument and part (ii) of Assumption \ref{ass_mom}.\footnote{Since $\widehat\Lambda$ and $\Lambda$ are invariant under a common location
shift of the data, we may apply the cited result to the centered array
$X_i - E[X_i]$, whose marginal $\psi_1$-norms are bounded by $B_n$ under
Assumption \ref{ass_mom}(i).}

For the second term, note that, using Lemma \ref{lem_bernstein} via Assumptions \ref{ass_mom} and \ref{ass_rate}, that $\|S_n \|_\infty = O_P(B_n \sqrt{\log p }) = O_P(\sqrt{\log p })$, where the last equality follows because $B_n = O(1)$ by assumption.
Putting everything together then, we conclude
\[
\left\|\left(\widehat{\Lambda}^{-1 / 2}-\Lambda^{-1 / 2}\right)
S_n\right\|_{\infty} \leq o_P\left(1 / \log (pn)\right) O_P(\sqrt{\log p}) =
o_P(1/\sqrt{\log (pn)}).
\]
\end{proof}

\Corhdboot*
\begin{proof}
The proof of this theorem proceeds just as the proof of Theorem \ref{cor_hdkapprox}. Let $\tilde S_n^B := \hat\Lambda^{-1/2}S_n^B$ and
\[
R_n = \hat \Lambda^{-1/2} S_n^B -  \Lambda^{-1/2} S_n^B = (\hat \Lambda^{-1/2} -  \Lambda^{-1/2}) S_n^B,
\]
and thus we can use Theorem \ref{lem_hdkapproxboot} to prove the desired result so long as we can show that, sufficiently, $\|R_n \|_\infty  = o_P(1/\sqrt{\log (pn)})$. 

Note then that almost surely, as in Theorem \ref{cor_hdkapprox},
\begin{align*}
    \left\|\left(\widehat{\Lambda}^{-1 / 2}-\Lambda^{-1 / 2}\right) S_n^B\right\|_{\infty} &\leq \|\widehat{\Lambda}^{-1 / 2}-\Lambda^{-1 / 2} \|_{\max} \left\|S_n^B\right\|_{\infty}.
\end{align*}
Using the same arguments as in the proof of Theorem \ref{cor_hdkapprox}, we conclude that $
\|\widehat{\Lambda}^{-1/2}-\Lambda^{-1/2} \|_{\infty} = o_P(1/\log (pn))$. By Lemma \ref{lem_sbb} we have that $\|S_n^B \|_\infty = O_P( B_n\sqrt{\log p})$, and under $B_n = O(1)$ then $\|S_n^B \|_\infty = O_P( \sqrt{\log p})$.
Thus, we complete the proof as in Theorem \ref{cor_hdkapprox}.
\end{proof}

\Smallkstu*
\begin{proof}
We consider Algorithm 2.1 from \cite{romano_control_2007}, which at each stage takes as inputs critical values $\hat{c}_{n, K}(1-\alpha,k)$ and a vector of test statistics $T_n$; we denote this as $\texttt{RW-2.1}(T_n, \hat{c}_{n, K}(1-\alpha,k))$.

The choice of critical values we will use is based on a bootstrap of the $1-\alpha$ quantile of the $\kmax$:
\[
\hat{c}_{n, K}(1-\alpha, k) := 1-\alpha \text{ quantile of } (\hat\Lambda^{-1/2} S_{n}^B)_{K,[k]} \mid X.
\]
As in \cite{romano_control_2007}, define $I(P)$ to be the set of indices of true null hypotheses under $P$. Note that, for any $K \supset I(P)$, 
$$
\hat{c}_{n, K}(1-\alpha, k) \geq \hat{c}_{n, I(P)}(1-\alpha, k)
$$
because for any $K \supset I(P)$, and under any distribution, $ (\hat\Lambda^{-1/2} S_{n}^B)_{K,[k]} \geq  (\hat\Lambda^{-1/2} S_{n}^B)_{I(P),[k]}$ almost surely (i.e., the $\kmax$ statistic can only get larger if we include more test statistics without dropping the others).  

As such, Theorem 2.1 (i) holds, and we conclude that $\texttt{RW-2.1}(\hat\Lambda^{-1/2}T_n, \hat{c}_{n, K}(1-\alpha,k))$ delivers:
\[
k\text{-FWER}_P \leq P\left\{\kmax \left(\hat\Lambda_{jj}^{-1/2}T_{n,j}: j \in I(P)\right)>\hat{c}_{n, I(P)}(1-\alpha, k)\right\}.
\]
Specifically, then, it is sufficient to show that
\[
\limsup_{n,p} P\left\{\kmax \left(\hat\Lambda_{jj}^{-1/2}T_{n,j}: j \in I(P)\right)> \hat{c}_{n, I(P)}(1-\alpha, k)\right\} \leq \alpha.
\]
To continue to use the notation of \cite{romano_control_2007}, let $\hat\theta_{n,j}:=\bar X_{n,j}$. Since $\theta_j(P) \leq 0$ for $j \in I(P)$, it follows that, almost surely,

$$
\begin{aligned}
\kmax \left(\hat\Lambda_{jj}^{-1/2}T_{n,j}: j \in I(P)\right) & =\kmax \left(\hat\Lambda_{jj}^{-1/2}\sqrt{n} \hat{\theta}_{n, j}: j \in I(P)\right) \\
& \leq \kmax \left(\hat\Lambda_{jj}^{-1/2}\sqrt{n}\left[\hat{\theta}_{n, j}-\theta_j(P)\right]: j \in I(P)\right)\\
& = \kmax \left(\hat\Lambda_{jj}^{-1/2}S_{n,j}: j \in I(P)\right)
\end{aligned}
$$
and therefore
\begin{align*}
    &P\left\{\kmax \left(\hat\Lambda_{jj}^{-1/2}T_{n,j}: j \in I(P)\right)> \hat{c}_{n, I(P)}(1-\alpha, k)\right\} \\
    \leq &P\left\{\kmax \left(\hat\Lambda_{jj}^{-1/2}S_{n,j}: j \in I(P)\right)>\hat{c}_{n, I(P)}(1-\alpha, k)\right\} .
\end{align*}
If we can show that the limit of the quantity on the right-hand side is no greater than $\alpha$, the proof is complete. 

Recall, however, that Theorem \ref{cor_hdkapprox} (via Theorem \ref{lem_hdkapprox}) gives us that, for any viable $K$,
\[
\sup _{t \in \mathbb{R}}\left|P\left( \left(\hat\Lambda^{-1/2}S_n\right)_{K,[k]} \leq t\right)-P\left(N(0,\Sigma_0)_{K,[k]} \leq t\right)\right| \leq C_2\beta_n + \rho_n
\]
and Theorem \ref{cor_hdkapproxboot} (via Theorem \ref{lem_hdkapproxboot}) that, with probability at least $1-\nu_n$,  
\[
\sup _{t \in \mathbb{R}}\left|P^B\left( \left(\hat\Lambda^{-1/2}S_n^B\right)_{K,[k]} \leq t \right)-P\left(N(0,\Sigma_0)_{K,[k]} \leq t\right)\right| \leq C_4\beta_n + \rho_n^B
\]
and so by the triangle inequality, with probability at least $1-\nu_n$,
\[
\sup _{t \in \mathbb{R}}\left|P\left( \left(\hat\Lambda^{-1/2}S_n\right)_{K,[k]} \leq t\right)-P^B\left( \left(\hat\Lambda^{-1/2}S_n^B\right)_{K,[k]} \leq t \right)\right| \leq C_5\beta_n + \rho_n + \rho_n^B.
\]


Now let $c_{1-\gamma}$ be the $1-\gamma$ quantile of
$\left(\hat\Lambda^{-1/2}S_n\right)_{K,[k]}$, and use $c_{1-\gamma}^B =
\hat{c}_{n, K}(1-\gamma, k)$ for shorthand. Since $\rho_n^B = o_P(1)$, fix a
deterministic sequence $w_n \searrow 0$ such that $P(\rho_n^B > w_n) = o(1)$,
and define the deterministic sequences
\[
\check\beta_n := C_5\beta_n + \rho_n + w_n, \qquad
\tilde\beta_n := 2(C_2+C_5)\beta_n + 4\rho_n + 2w_n,
\]
so that $c_{1-\alpha+\check\beta_n}$ and $c_{1-\alpha-\tilde\beta_n}$ are
well-defined population quantiles (for all $n$ large enough that
$\check\beta_n < 1-\alpha$ and $\tilde\beta_n < \alpha$). Work on the
intersection of the event on which the uniform comparison of the previous
display holds and the event $\{\rho_n^B \le w_n\}$; this intersection has
probability at least $1-\nu_n - P(\rho_n^B > w_n)$, and on it the uniform
comparison holds with the deterministic bound $C_5\beta_n + \rho_n + w_n$.
Consequently, on this event, both
\begin{align*}
    P^B \left(\left(\hat\Lambda^{-1/2}S_n^B\right)_{K,[k]} \leq c_{1-\alpha+\check\beta_n}\right) \geq P\left(\left(\hat\Lambda^{-1/2}S_n\right)_{K,[k]} \leq c_{1-\alpha+\check\beta_n}\right) - (C_5\beta_n + \rho_n + w_n) \geq 1-\alpha
\end{align*}
and
\begin{align*}
    P^B \left(\left(\hat\Lambda^{-1/2}S_n^B\right)_{K,[k]} \leq c_{1-\alpha-\tilde\beta_n}\right) &\leq P \left(\left(\hat\Lambda^{-1/2}S_n\right)_{K,[k]} \leq c_{1-\alpha-\tilde\beta_n }\right) + C_5\beta_n + \rho_n + w_n \\
    &\leq (1-\alpha -\tilde\beta_n) + 2C_2\beta_n + 2\rho_n + C_5\beta_n + \rho_n + w_n \\
    &< 1-\alpha,
\end{align*}
where the second inequality comes from the fact that
$$P \left(\left(\hat\Lambda^{-1/2}S_n\right)_{K,[k]} \leq t + \epsilon\right) \leq P \left(\left(\hat\Lambda^{-1/2}S_n\right)_{K,[k]} \leq t\right) + C_1 k \epsilon \sqrt{1 \vee \log (p / \epsilon)} + 2(C_2\beta_n + \rho_n)$$
using the argument of Lemma A.7 of \cite{ding_gaussian_2025} via Theorem
\ref{lem_hdkapprox} (letting $t=c_{1-\alpha-\tilde\beta_n}-\epsilon$ and
taking $\epsilon \searrow 0$), and the final inequality holds because
$\tilde\beta_n > (2C_2+C_5)\beta_n + 3\rho_n + w_n$.

This in turn implies that
\[
P(c_{1-\alpha-\tilde\beta_n} < c^B_{1-\alpha} \leq c_{1-\alpha+\check\beta_n}) \geq 1- \nu_n - P(\rho_n^B > w_n).
\]
This in turn implies that both
\begin{align*}
    P((\hat\Lambda^{-1/2}S_n)_{K,[k]} > c^B_{1-\alpha}) &\leq P((\hat\Lambda^{-1/2}S_n)_{K,[k]} > c_{1-\alpha-\tilde\beta_n}) + (\nu_n + P(\rho_n^B > w_n)) \\
    &\leq \alpha + \tilde\beta_n + (\nu_n + P(\rho_n^B > w_n))
\end{align*}
and, using Lemma A.7 of \cite{ding_gaussian_2025} via Theorem \ref{lem_hdkapprox} again,
\begin{align*}
    P((\hat\Lambda^{-1/2}S_n)_{K,[k]} > c^B_{1-\alpha}) &\geq P((\hat\Lambda^{-1/2}S_n)_{K,[k]} > c_{1-(\alpha-\check\beta_n)}) - (\nu_n + P(\rho_n^B > w_n)) \\
    &\geq \alpha - \check\beta_n - 2(C_2\beta_n + \rho_n) - (\nu_n + P(\rho_n^B > w_n)).
\end{align*}
These two statements then jointly imply that 
\[
|P((\hat\Lambda^{-1/2}S_n)_{K,[k]}> \hat c_{n,K}(1 - \alpha,k)) - \alpha| \leq o(1).
\]
So, chaining inequalities, we arrive at
\begin{align*}
     k\text{-FWER}_P \leq P\left\{(\hat\Lambda^{-1/2}S_n)_{I(P),[k]}> \hat{c}_{n,I(P)}\left(1-\alpha, k\right)\right\}
    \leq \alpha +o(1)
\end{align*}
and so
\begin{align*}
    \limsup_{n,p\to\infty}
     k\text{-FWER}_P 
    \leq \alpha 
\end{align*}
which is precisely what we wanted to show.

To prove the second statement, consider the $\tilde H_{0,j}$ corresponding to all $\theta_j(P)>0$. Note that, for the Gaussian multiplier bootstrap, using Assumptions \ref{ass_mom} and \ref{ass_rate}, by the tail inequality Equation \ref{eqn_tail} in the proof of Lemma \ref{lem_sbb} we can conclude that (because $B_n = O(1)$)
\[
\hat c_{n, [p]}\left(1-\alpha, k\right) \leq \hat c_{n, [p]}\left(1-\alpha, 1\right)  = O_P\left(\sqrt{\log p} \right)
\]
(noting that premultiplication by $\hat\Lambda^{-1/2}$ under the data conditional law is just a particular choice of normalization). Furthermore, each $S_{n,j} = \sqrt{n} \left[\hat{\theta}_{n,j}-\theta_j(P)\right]$ is tight, and the conventional Lindeberg-Levy CLT and non-degeneracy and bounded fourth moments from Assumption \ref{ass_mom} yields $|\hat\Lambda_{jj}^{-1/2}-\Lambda_{jj}^{-1/2}|  =  O_P(n^{-1/2})$, so
\begin{align*}
\hat\Lambda_{jj}^{-1/2}T_{n, j}&=\hat\Lambda_{jj}^{-1/2}\sqrt{n} \hat{\theta}_{n,j} \\
&=(\Lambda_{jj}^{-1/2}+O_P(n^{-1/2}))\sqrt{n} (\hat{\theta}_{n,j}-\theta_j(P)) +O_P(n^{-1/2})\sqrt{n} \theta_j(P) + \Lambda_{jj}^{-1/2}\sqrt{n} \theta_j(P)\\
&=\Lambda_{jj}^{-1/2}O_P(1)+O_P(n^{-1/2}) +O_P(1)\theta_j(P) + \Lambda_{jj}^{-1/2}\sqrt{n} \theta_j(P)
\xrightarrow{P}\infty
\end{align*}
because the last term is diverging. However, because $\sqrt{n}$ grows faster than $ \sqrt{\log p}$ when $B_n = O(1)$ by Assumption \ref{ass_rate}, we have that also 
$$
\frac{\hat\Lambda_{jj}^{-1/2}T_{n, j}}{ \sqrt{\log p} \,\,} \xrightarrow{P} \infty.
$$
Therefore, with probability tending to one, $T_{n, j}>\hat{c}_{n,[p]}\left(1-\alpha, k\right)$,  resulting in the rejection of $\tilde H_{0,j}$ in the first step of Algorithm 2.1, so long as $\theta_j(P)$ is fixed or approaches zero slower than a rate of $ \sqrt{\log p/n}$.

To prove the third statement, the asymptotic validity of the streamlined Algorithm 2.2 follows from the separation assumption
\[
\min_{j \notin I(P)}\theta_j(P) \gg B_n \sqrt{\log p /n}
\]
using the same logic as for the proof of the second statement and the argument of proof of Theorem 3.3 of \cite{romano_control_2007}, i.e., $\min \left(\hat\Lambda_{jj}^{-1/2}T_{n, j}: j \notin I(P)\right)$ is diverging at rate $\sqrt{n}$, and if any $\theta_j(P)=0$ then $
\max \left(\hat\Lambda_{jj}^{-1/2}T_{n, j}: j \in I(P)\right) = O_P(B_n \sqrt{\log p})=O_P(\sqrt{\log p})
$.

\end{proof}

\Hdbootinfluence*
\begin{proof} We follow the proof strategy of \cite{belloni_high-dimensional_2018}, Theorem 2.2.

Because of Assumption \ref{ass_approx}, as per \cite{belloni_high-dimensional_2018}, there exist two deterministic sequences $\delta_n \searrow 0, v_n \searrow 0$ that satisfy 
\[
P\left(\max _{j \in[p]} \frac{1}{n} \sum_{i=1}^n\left[\left(\hat{\psi}_{i j}-\psi_{i j}\right)^2\right]>\delta_n^2 / \log ^2(p n)\right) \leq v_n,
\]
Set $\Delta_n:=\delta_n / \sqrt{\log (p n)}$, and observe that, under the bootstrap law $P^B$, that $\frac{1}{\sqrt{n}} \sum_{i=1}^n \xi_i\left(\hat{\psi}_{i j}-\psi_{i j}\right)$ is a zero-mean Gaussian with variance $n^{-1} \sum_{i=1}^n\left(\hat\psi_{i j}-\psi_{i j}\right)^2$. Thus by the Borell-TIS inequality, for any $\alpha \in (0,1)$,
\[
P^B\left(\max _{j \in[p]}\left|\frac{1}{\sqrt{n}} \sum_{i=1}^n \xi_i\left(\hat{\psi}_{i j}-\psi_{i j}\right)\right|>\max _{j \in[p]} \sqrt{n^{-1}\sum_{i=1}^n\left[\left(\hat{\psi}_{i j}-\psi_{i j}\right)^2\right]}(\sqrt{2 \log (2p)}+\sqrt{2 \log (1 / \alpha)})\right) \leq 2\alpha.
\]
Using Assumption \ref{ass_approx} and setting $\alpha=\delta_n / 2 \geq 1 /(2 n)$, we then have that
\[
P^B\left(\max _{j \in[p]}\left|\frac{1}{\sqrt{n}} \sum_{i=1}^n \xi_i\left(\hat{\psi}_{i j}-\psi_{i j}\right)\right|>C_7 \Delta_n\right) \leq \delta_n
\]
with probability $1 - v_n$ for some universal constant $C_7$. So far, these arguments are identical to \cite{belloni_high-dimensional_2018}.

Now consider the event $\left\{\max _{j \in[p]}\left|\frac{1}{\sqrt{n}} \sum_{i=1}^n \xi_i\left(\hat{\psi}_{i j}-\psi_{i j}\right)\right|>C_7 \Delta_n\right\}$ for any $p,n$ and $t>0$ and $k \geq 1 $ and notice that
\begin{align*}
    &\left\{\max _{j \in[p]}\left|\frac{1}{\sqrt{n}} \sum_{i=1}^n \xi_i\left(\hat{\psi}_{i j}-\psi_{i j}\right)\right|>C_7 \Delta_n\right\} \\&\supseteq \left\{\left|\kmax_{j \in[p]}\frac{1}{\sqrt{n}} \sum_{i=1}^n \xi_i\hat{\psi}_{i j}-\kmax_{j \in[p]}\frac{1}{\sqrt{n}}\sum_{i=1}^n\xi_i\psi_{i j}\right|>C_7 \Delta_n\right\} \\
    &\supseteq \left\{\kmax_{j \in[p]}\frac{1}{\sqrt{n}}\sum_{i=1}^n\xi_i\psi_{i j} - \kmax_{j \in[p]}\frac{1}{\sqrt{n}} \sum_{i=1}^n \xi_i\hat{\psi}_{i j}>C_7 \Delta_n\right\} \\
    &\supseteq \left\{\kmax_{j \in[p]}\frac{1}{\sqrt{n}}\sum_{i=1}^n\xi_i \psi_{i j} > t + C_7 \Delta_n \right\} \cap\left\{ \kmax_{j \in[p]}\frac{1}{\sqrt{n}} \sum_{i=1}^n \xi_i\hat{\psi}_{i j} \leq t \right\}.
\end{align*}
where the second line follows by noting that the $\kmax$ function is $L_\infty$ Lipschitz. Then notice that, almost surely, by a Boole-Fréchet inequality,
\begin{align*}
    &P^B\left(\left\{\kmax_{j \in[p]}\frac{1}{\sqrt{n}}\sum_{i=1}^n\xi_i \psi_{i j} > t + C_7 \Delta_n \right\} \cap\left\{ \kmax_{j \in[p]}\frac{1}{\sqrt{n}} \sum_{i=1}^n \xi_i\hat{\psi}_{i j} \leq t \right\}\right)\\
    \geq &P^B\left(\left\{ \kmax_{j \in[p]}\frac{1}{\sqrt{n}} \sum_{i=1}^n \xi_i\hat{\psi}_{i j} \leq t \right\}\right) - P^B\left(\left\{\kmax_{j \in[p]}\frac{1}{\sqrt{n}}\sum_{i=1}^n\xi_i \psi_{i j} \leq  t + C_7 \Delta_n \right\} \right),
\end{align*}
meaning that still with probability $1-v_n$
\[
P^B\left( \kmax_{j \in[p]}\frac{1}{\sqrt{n}} \sum_{i=1}^n \xi_i\hat{\psi}_{i j} \leq t \right)  \leq P^B\left(\kmax_{j \in[p]}\frac{1}{\sqrt{n}}\sum_{i=1}^n\xi_i \psi_{i j} \leq  t + C_7 \Delta_n  \right)  + \delta_n.
\]
However, we also have, using equation (34) (which invokes anti-concentration result Corollary A.1) and Proposition 2.4 from \cite{ding_gaussian_2025} and the triangle inequality,
\begin{align*}
   P^B\left( \kmax_{j \in[p]}\frac{1}{\sqrt{n}} \sum_{i=1}^n \xi_i\hat{\psi}_{i j} \leq t \right)  \leq P\left(N(0,\Sigma)_{[k]} \leq  t  \right)  + \delta_n + C_8\beta_n
\end{align*}
with a probability $1-1 /\left(2 n^4\right)-1 / n-3 \nu_n - v_n=1-o(1)$ with some constant $C_8$ that only depends on $b_1, b_2$.

The lower bound can be proven using nearly identical arguments and setting a threshold of $t-C_7\Delta_n$; because these arguments hold for any $t$, we thus conclude, for fixed $k$, with probability at least $1-1 /\left(2 n^4\right)-1 / n-3 \nu_n - v_n$,
\[
\sup _{t \in \mathbb{R}}\left|P^B\left( \tilde S_{n,[k]}^B \leq t \right)-P\left(N(0,\Sigma)_{[k]} \leq t\right)\right| \leq \delta_n + C_8\beta_n.
\]
Application of Lemma \ref{lem_conp} then completes the proof.

\end{proof}

\Smallkaif*
\begin{proof}
We consider Algorithm 2.1 from \cite{romano_control_2007}, which at each stage takes as inputs critical values $\hat{c}_{n, K}(1-\alpha,k)$ and a vector of test statistics $T_n$; we denote this as $\texttt{RW-2.1}(T_n, \hat{c}_{n, K}(1-\alpha,k))$.

The choice of critical values we will use is based on a bootstrap of the $1-\alpha$ quantile of the $\kmax$:
\[
\hat{c}_{n, K}(1-\alpha, k) := 1-\alpha \text{ quantile of } (\tilde S_{n}^B)_{K,[k]} \mid X.
\]
As in \cite{romano_control_2007}, define $I(P)$ to be the set of indices of true null hypotheses under $P$. Note that, for any $K \supset I(P)$, 
$$
\hat{c}_{n, K}(1-\alpha, k) \geq \hat{c}_{n, I(P)}(1-\alpha, k)
$$
because for any $K \supset I(P)$, and under any distribution, $ (\tilde S_{n}^B)_{K,[k]} \geq  (\tilde S_{n}^B)_{I(P),[k]}$ almost surely (i.e., the $\kmax$ statistic can only get larger if we include more test statistics without dropping the others).  

As such, Theorem 2.1 (i) holds, and we conclude that $\texttt{RW-2.1}(\sqrt{n}\hat\theta, \hat{c}_{n, K}(1-\alpha,k))$ delivers:
\[
k\text{-FWER}_P \leq P\left\{\kmax \left(\sqrt{n}\hat\theta_j: j \in I(P)\right)>\hat{c}_{n, I(P)}(1-\alpha, k)\right\}.
\]
Specifically, then, it is sufficient to show that
\[
\limsup_{n,p} P\left\{\kmax \left(\sqrt{n}\hat\theta_j: j \in I(P)\right)> \hat{c}_{n, I(P)}(1-\alpha, k)\right\} \leq \alpha.
\]
Since $\theta_j(P) \leq 0$ for $j \in I(P)$, it follows that, almost surely,
$$
\begin{aligned}
\kmax \left(\sqrt{n}\hat\theta_j: j \in I(P)\right) 
& \leq \kmax \left(\sqrt{n}\left[\hat{\theta}_{ j}-\theta_j(P)\right]: j \in I(P)\right)
\end{aligned}
$$
and therefore
\begin{align*}
    &P\left\{\kmax \left(\sqrt{n}\hat\theta_j: j \in I(P)\right)> \hat{c}_{n, I(P)}(1-\alpha, k)\right\} \\
    \leq &P\left\{\kmax \left(\sqrt{n}(\hat\theta_j-\theta_j(P)): j \in I(P)\right)>\hat{c}_{n, I(P)}(1-\alpha, k)\right\} \\
    = &P\left\{\kmax \left(\tilde S_{n,j}: j \in I(P)\right)>\hat{c}_{n, I(P)}(1-\alpha, k)\right\} 
\end{align*}
where the last equality is simply notation. If we can show that the limit of the quantity on the right-hand side is no greater than $\alpha$, the proof is complete. 

Recall, however, that Theorem \ref{lem_hdkapprox} gives us that, for any viable $K$,
\[
\sup _{t \in \mathbb{R}}\left|P\left( \tilde S_{n,K,[k]} \leq t\right)-P\left(N(0,\Sigma)_{K,[k]} \leq t\right)\right| \leq C_2\beta_n + \rho_n 
\]
and Theorem \ref{thm_hdbootinfluence} that, with probability at least $1-\tilde\nu_n$ (where $\tilde\nu_n:=1 /\left(2 n^4\right)+1 / n+3 \nu_n + v_n$),
\[
\sup _{t \in \mathbb{R}}\left|P^B\left( \tilde S_{n,K,[k]}^B \leq t \right)-P\left(N(0,\Sigma)_{K,[k]} \leq t\right)\right| \leq \delta_n + C_8\beta_n
\]
and so by the triangle inequality, with probability at least $1-\tilde\nu_n$,
\[
\sup _{t \in \mathbb{R}}\left| P\left( \tilde S_{n,K,[k]} \leq t\right)-P^B\left( \tilde S_{n,K,[k]}^B \leq t \right)\right| \leq \delta_n + \rho_n + C_9 \beta_n.
\]


Now let $c_{1-\gamma}$ be the $1-\gamma$ quantile of $\tilde S_{n,K,[k]}$, and use $c_{1-\gamma}^B = \hat{c}_{n, K}(1-\gamma, k)$ for shorthand. Then, letting $\check\beta_n:=\delta_n + \rho_n + C_9\beta_n$, with probability at least $1-\tilde \nu_n$ both
\begin{align*}
    P^B \left(\tilde S_{n,K,[k]}^B  \leq c_{1-\alpha+\check\beta_n}\right) \geq P\left(\tilde S_{n,K,[k]} \leq c_{1-\alpha+\check\beta_n}\right) - (\delta_n + \rho_n + C_9 \beta_n) \geq 1-\alpha
\end{align*}
and letting $\tilde \beta_n := (2C_2+2C_9)\beta_n + 4\rho_n  + 2\delta_n$,
\begin{align*}
    P^B \left(\tilde S_{n,K,[k]}^B \leq c_{1-\alpha-\tilde\beta_n}\right) &\leq P \left(\tilde S_{n,K,[k]} \leq c_{1-\alpha-\tilde\beta_n }\right) + \delta_n + \rho_n + C_9 \beta_n \\
    &\leq (1-\alpha -\tilde\beta_n) + 2C_2\beta_n + 2\rho_n + \delta_n + \rho_n + C_9 \beta_n \\
    &= (1-\alpha -\tilde\beta_n) + C_{10}\beta_n + 3\rho_n  + \delta_n \\
    &< 1-\alpha
\end{align*}
where again the second inequality comes from the fact that $$P \left(\tilde S_{n,K,[k]} \leq t + \epsilon\right) \leq P \left(\tilde S_{n,K,[k]} \leq t\right) + C_1 k \epsilon \sqrt{1 \vee \log (p / \epsilon)} + 2(C_2\beta_n + \rho_n)$$ using Lemma A.7 of \cite{ding_gaussian_2025} via Theorem \ref{lem_hdkapprox} (and then, specifically, letting $t=c_{1-\alpha-\tilde\beta_n}-\epsilon$ and then taking $\epsilon \searrow 0$).

This in turn implies that
\[
P(c_{1-\alpha-\tilde\beta_n} < c^B_{1-\alpha} \leq c_{1-\alpha+\check\beta_n}) \geq 1- \tilde\nu_n.
\]
This in turn implies that both
\begin{align*}
    P(\tilde S_{n,K,[k]} > c^B_{1-\alpha}) &\leq P(\tilde S_{n,K,[k]} > c_{1-\alpha-\tilde\beta_n}) + \tilde\nu_n \\
    &\leq \alpha + \tilde\beta_n + \tilde\nu_n 
\end{align*}
and, using Lemma A.7 of \cite{ding_gaussian_2025} via Theorem \ref{lem_hdkapprox} again,
\begin{align*}
    P(\tilde S_{n,K,[k]} > c^B_{1-\alpha}) &\geq P(\tilde S_{n,K,[k]} > c_{1-(\alpha-\check\beta_n)}) - \tilde\nu_n  \\
    &\geq \alpha - \check\beta_n - 2(C_2\beta_n + \rho_n) - \tilde\nu_n.
\end{align*}
These two statements then jointly imply that 
\[
|P(\tilde S_{n,K,[k]}> \hat c_{n,K}(1 - \alpha)) - \alpha| \leq o(1).
\]
So, chaining inequalities, we arrive at
\begin{align*}
     k\text{-FWER}_P \leq P\left\{\tilde S_{n,I(P),[k]}> \hat{c}_{n,I(P)}\left(1-\alpha, k\right)\right\}
    \leq \alpha +o(1)
\end{align*}
and so
\begin{align*}
    \limsup_{n,p\to\infty}
     k\text{-FWER}_P 
    \leq \alpha 
\end{align*}
which is precisely what we wanted to show.

To prove the second statement, consider the $\tilde H_{0,j}$ corresponding to all $\theta_j(P)>0$. 
Note that the critical values are based on $\tilde S_n^B$, which is built
from \textit{estimated} influence functions; by Lemma \ref{lem_sbb_if} (which
uses Assumptions \ref{ass_mom}, \ref{ass_rate}, and \ref{ass_approx}), we
have $$\hat c_{n, [p]}\left(1-\alpha, k\right) \leq \hat c_{n,
[p]}\left(1-\alpha, 1\right) = O_P\left(B_n \sqrt{\log p} \right)$$
because the two conditional bounds from this lemma imply that for every $\varepsilon>0$ there exists $M<\infty$ such that, with unconditional probability tending to one, $P^B(\|\tilde S_n^B\|_\infty > M B_n\sqrt{\log p}) \le \varepsilon$ (and taking $\varepsilon<\alpha$ yields $\hat c_{n,[p]}(1-\alpha,1) \le M B_n\sqrt{\log p}$ with probability tending to one, i.e., $\hat c_{n,[p]}(1-\alpha,k) = O_P(B_n\sqrt{\log p})$).

Furthermore, each $\tilde S_{n,j} = \sqrt{n} \left[\hat{\theta}_{n,j}-\theta_j(P)\right]$ is tight, and so
$$
\sqrt{n} \hat{\theta}_{j} =\sqrt{n} (\hat{\theta}_{j}-\theta_j(P)) + \sqrt{n} \theta_j(P)\xrightarrow{P} \infty.
$$
However, because $\sqrt{n}$ grows faster than $B_n \sqrt{\log p}$ by Assumption \ref{ass_rate}, we have that also 
$$
\frac{\sqrt{n}\hat\theta_{j}}{ B_n\sqrt{\log p} \,\,} \xrightarrow{P} \infty.
$$
Therefore, with probability tending to one, $\sqrt{n}\hat\theta_j >\hat{c}_{n,[p]}\left(1-\alpha, k\right)$,  resulting in the rejection of $\tilde H_{0,j}$ in the first step of Algorithm 2.1, so long as $\theta_j(P)$ is fixed or approaches zero slower than a rate of $B_n \sqrt{\log p/n}$.

To prove the third statement, the asymptotic validity of the streamlined Algorithm 2.2 follows from the separation assumption 
\[
\min_{j \notin I(P)}\theta_j(P) \gg B_n \sqrt{\log p /n}
\]
using the same logic as for the proof of the second statement and the argument of proof of Theorem 3.3 of \cite{romano_control_2007}, i.e., $\min \left(\sqrt{n}\hat\theta_j: j \notin I(P)\right)$ is diverging at rate $\sqrt{n}$, and if any $\theta_j(P)=0$ then $
\max \left(\sqrt{n}\hat\theta_j: j \in I(P)\right) = O_P(B_n \sqrt{\log p})
$.
\end{proof}

\Ascoreest*
\begin{proof}
To show the conditional unbiasedness of the estimator, observe that
\begin{align*}
    E_P\left[\hat \theta_j^\textrm{acc} \mid \hat\eta_j \right]  &= \frac{1}{m} \sum_{i \in\mathcal{I}^\textrm{eval}} E_P[S_{ij}\mid \hat\eta_j]  =\theta_j^\textrm{acc}(\hat\eta_j;P)
\end{align*}
because $S_{ij}$ is drawn i.i.d. under the $\hat\eta_j$-conditional law.

To show $\sqrt{m}$-consistency and asymptotic normality, we make use of Lemma 1 from \cite{modarressi_causal_2025}, a clean statement of a conditional Berry-Esseen theorem for i.n.i.d. data. Noticing that for all $i \in \mathcal{I}^\textrm{eval}$ that $\zeta_{ij} := S_{ij}-\theta_j^\textrm{acc}(\hat\eta_j;P)$ are mean zero and independent, and that their third moments are bounded under the conditional law almost surely (as $S_{ij} \in \{0,1\}$ and $\theta_j^\textrm{acc}(\hat\eta_j;P) \in [0,1]$), using the assumption that $P(S_{ij}=1 \mid \hat\eta_j) \in [\delta, 1-\delta]$, this lemma provides that 
\begin{align*}
    \sup_{t \in \mathbb{R}} \left| P\left(m^{-1/2} \sum_{i \in \mathcal{I}^\textrm{eval}}\zeta_{ij} \leq t \mid \hat\eta_j \right) - P(N(0, \operatorname{Var}(S_{ij} \mid \hat\eta_j)) \leq t \mid \hat\eta_j) \right| \leq C_0 \frac{\sum_{i \in \mathcal{I}^\textrm{eval}} E\left[\left|\zeta_{ij}\right|^3 \mid \hat\eta_j \right]}{\left(\sum_{i \in \mathcal{I}^\textrm{eval}} E\left[\zeta_{ij}^2 \mid \hat\eta_j\right]\right)^{3 / 2}}
\end{align*}
for some universal constant $C_0 < \infty$, almost surely. The left-hand side of this inequality is exactly the left-hand side of the convergence statement in the proposition being proved. As for the right-hand side, note that then, given our previous observations and assumptions, $E\left[\left|\zeta_{ij}\right|^3 \mid \hat\eta_j \right] \leq 1 $ almost surely, and $E\left[\zeta_{ij}^2 \mid \hat\eta_j\right] =\operatorname{Var}(S_{ij} \mid \hat\eta_j) \geq \delta(1-\delta) = C_1 > 0$, almost surely. Thus we have that
\begin{align*}
    C_0 \frac{\sum_{i \in \mathcal{I}^\textrm{eval}} E\left[\left|\zeta_{ij}\right|^3 \mid \hat\eta_j \right]}{\left(\sum_{i \in \mathcal{I}^\textrm{eval}} E\left[\zeta_{ij}^2 \mid \hat\eta_j\right]\right)^{3 / 2}} \leq C_0\frac{ m}{(C_1m)^{3/2}} = \frac{C_0}{C_1^{3/2}} m^{-1/2} \asymp m^{-1/2}.
\end{align*}
Taking $m \to\infty$ then proves the stated claim.
\end{proof}

\Hdreg*
\begin{proof}
Let $\|\cdot\|_{p}$ denote the usual $\ell^p$ norm for a finite-dimensional vector, and let
$\|v\|_{p,n}:=E_n[|v_i|^p]^{1/p}$ denote the $L^p$ norm under the empirical measure.
Lastly, let the $L^p$ norm under probability measure $P$ be $\|\cdot\|_{p,P}$.
We note two inequalities used repeatedly: for $n$-dimensional vectors $v,u$, H\"older's inequality gives
\[
\|vu\|_{2,n}\le \|v\|_{4,n}\|u\|_{4,n},
\]
and for conformable matrix $A$ and vector $v$, $\|Av\|_2\le \|A\|_{\mathrm{op}}\|v\|_2$.

\noindent
\textbf{Preliminaries.}
Write $\gamma_T:=\Sigma_D^{-1}E_P[DT]$, so
$\widetilde T=T-\gamma_T^{\mathtt T}D$. Since $T$ and $D$ are subgaussian and $d$ is fixed,
$\widetilde T$ is subgaussian. Hence $E_P[\widetilde T^4]<\infty$ and $\|\widetilde T\|_{4,n}=O_P(1)$ by the WLLN.

Let $\gamma_{Y_j}:=\Sigma_D^{-1}E_P[DY_j]$. Then $\widetilde Y_j=Y_j-\gamma_{Y_j}^{\mathtt T}D$ and is subgaussian by the same argument, and $E_P[\widetilde Y_j^4]<\infty$.

Since $d$ is fixed and $\max_{j\le p}\|Y_j\|_{\psi_2}\le C_3$, we have $\sup_{j\le p}E_P[Y_j^2]\le C<\infty$.
By Cauchy-Schwarz,
\[
\|E_P[DY_j]\|_2 \le \sqrt{E_P[\|D\|_2^2]}\,\sqrt{E_P[Y_j^2]} \le C,
\]
uniformly in $j$ because $\|D\|_2^2$ is subexponential. Therefore, using (ND),
\[
\sup_{j\le p}\|\gamma_{Y_j}\|_2
\le \|\Sigma_D^{-1}\|_{\mathrm{op}}\sup_{j\le p}\|E_P[DY_j]\|_2
\le C<\infty.
\]
Similarly, by Cauchy-Schwarz and (ND),
\[
|\theta_j|
=\Omega^{-1}|E_P[\widetilde T\widetilde Y_j]|
\le \Omega^{-1}\sqrt{E_P[\widetilde T^2]E_P[\widetilde Y_j^2]}
\le C,
\]
uniformly in $j$. Hence $\widetilde U_j=\widetilde Y_j-\theta_j\widetilde T$ is also uniformly subgaussian.

Also, because $d$ is fixed and each $D_\ell$ is subgaussian, $E_P[\|D\|_2^4]<\infty$, and thus
\[
\big\|\|D_i\|_2\big\|_{4,n} = \Big(E_n[\|D_i\|_2^4]\Big)^{1/4}=O_P(1)
\]
by the WLLN (this will be used to control $\|e_T\|_{4,n}$ and $\|e_{Y_j}\|_{4,n}$).

Lastly, we will use that $\max_{j\le p}\|\widetilde Y_j\|_{4,n}=O_P(1)$ and $\max_{j\le p}\|\widetilde U_j\|_{4,n}=O_P(1)$. Concretely, $\|\widetilde Y_{ij}^4\|_{\psi_{1/2}} = \|\widetilde
Y_{ij}\|_{\psi_2}^4 \le C$ uniformly in $j$ by (SG) and Corollary 4 of
\cite{zhang_sharper_2021}, and hence $\|\widetilde Y_{ij}^4 - E[\widetilde
Y_{ij}^4]\|_{\psi_{1/2}} \le C$ by their Proposition 1. Thus Theorem 1(b)
of \cite{zhang_sharper_2021} (applied with $\theta = 1/2$ and equal
weights $w_i = n^{-1}$, so that $\|b\|_2 \lesssim n^{-1/2}$ and $\|b\|_2
L_n(\theta,b) \lesssim n^{-1}$ there), evaluated at $t = s + \log p$
together with a union bound over $j\in[p]$, yields
\[
\max_{j\le p}\left|\frac1n\sum_{i=1}^n \widetilde Y_{ij}^4 -
E_P[\widetilde Y_{ij}^4]\right|
=O_P\left( \sqrt{\frac{\log p}{n}}+\frac{(\log p)^2}{n}\right).
\]
Since $\sup_j E_P[\widetilde Y_j^4]<\infty$, and using the H\"older continuity bound
$|a^{1/4}-b^{1/4}|\le |a-b|^{1/4}$ for $a,b\ge0$, and applying Assumption \ref{ass_rate}, we obtain
\[
\max_{j\le p}\|\widetilde Y_j\|_{4,n}=O_P(1).
\]
Then $\widetilde U_j=\widetilde Y_j-\theta_j\widetilde T$ with $\sup_j|\theta_j|<\infty$ and $\|\widetilde T\|_{4,n}=O_P(1)$ gives
\[
\max_{j\le p}\|\widetilde U_j\|_{4,n}=O_P(1).
\]

\noindent
\textbf{Decomposition.}
We decompose
\[
\widehat\psi_{ij}-\psi_{ij}
=\big(\widehat\Omega^{-1}-\Omega^{-1}\big)\widehat{\widetilde T}_i\widehat{\widetilde U}_{ij}
+\Omega^{-1}e_{T,i}\widehat{\widetilde U}_{ij}
+\Omega^{-1}\widetilde T_i\big(\widehat{\widetilde U}_{ij}-\widetilde U_{ij}\big),
\]
where $e_{T,i}:=\widehat{\widetilde T}_i-\widetilde T_i=-D_i^{\mathtt T}(\widehat\gamma_T-\gamma_T)$.
By Minkowski's inequality,
\[
\max_{j\le p}\|\widehat\psi_{\cdot j}-\psi_{\cdot j}\|_{2,n}
\le (I)+(II)+(III),
\]
where
\[
\begin{aligned}
(I)&:=\max_{j\le p}\big\|\big(\widehat\Omega^{-1}-\Omega^{-1}\big)\widehat{\widetilde T}_i\widehat{\widetilde U}_{ij}\big\|_{2,n},\\
(II)&:=\max_{j\le p}\big\|\Omega^{-1}e_{T,i}\widehat{\widetilde U}_{ij}\big\|_{2,n},\\
(III)&:=\max_{j\le p}\big\|\Omega^{-1}\widetilde T_i(\widehat{\widetilde U}_{ij}-\widetilde U_{ij})\big\|_{2,n}.
\end{aligned}
\]

\noindent
\textbf{(I).}
By H\"older,
\[
(I)\le \big|\widehat\Omega^{-1}-\Omega^{-1}\big|\,\|\widehat{\widetilde T}\|_{4,n}\,\max_{j\le p}\|\widehat{\widetilde U}_j\|_{4,n}.
\]

We may write
\[
\widehat\Omega-\Omega = \big(E_n[\widetilde T^2]-\Omega\big) + E_n[\widehat{\widetilde T}^2-\widetilde T^2].
\]
The first term is $O_P(n^{-1/2})$ (by Bernstein/CLT since $\widetilde T^2$ is subexponential).
For the second term, note
\[
\widehat{\widetilde T}^2-\widetilde T^2 = (\widehat{\widetilde T}+\widetilde T)e_T,
\]
so by Cauchy-Schwarz and the identity $\widehat{\widetilde T}=\widetilde T+e_T$,
\[
\big|E_n[(\widehat{\widetilde T}+\widetilde T)e_T]\big|
\le \|\widehat{\widetilde T}+\widetilde T\|_{2,n}\|e_T\|_{2,n}
\le (\|e_T\|_{2,n}+2\|\widetilde T\|_{2,n})\|e_T\|_{2,n}.
\]

Recall $e_{T,i}=-D_i^{\mathtt T}(\widehat\gamma_T-\gamma_T)$, where
\[
\widehat\gamma_T:=(E_n[DD^{\mathtt T}])^{-1}E_n[DT], \qquad \gamma_T:=\Sigma_D^{-1}E_P[DT].
\]
Then
\[
\widehat\gamma_T-\gamma_T
=\big((E_n[DD^{\mathtt T}])^{-1}-\Sigma_D^{-1}\big)E_n[DT]
+\Sigma_D^{-1}\big(E_n[DT]-E_P[DT]\big).
\]
Using $A^{-1}-B^{-1}=A^{-1}(B-A)B^{-1}$ and operator-norm submultiplicativity,
\[
\big\|(E_n[DD^{\mathtt T}])^{-1}-\Sigma_D^{-1}\big\|_{\mathrm{op}}
\le \big\|(E_n[DD^{\mathtt T}])^{-1}\big\|_{\mathrm{op}}
\big\|\Sigma_D-E_n[DD^{\mathtt T}]\big\|_{\mathrm{op}}
\big\|\Sigma_D^{-1}\big\|_{\mathrm{op}}.
\]
By (ND), $\|\Sigma_D^{-1}\|_{\mathrm{op}}<\infty$.
Since $d$ is fixed and each $D_aD_b$ is subexponential, $\|\Sigma_D-E_n[DD^{\mathtt T}]\|_{\mathrm{op}}=O_P(n^{-1/2})$ (e.g.\ via entrywise Bernstein and $\|\cdot\|_{\mathrm{op}}\le \|\cdot\|_F$, the Frobenius norm).
Thus $E_n[DD^{\mathtt T}]\to_P \Sigma_D$ in operator norm, hence $\|(E_n[DD^{\mathtt T}])^{-1}\|_{\mathrm{op}}=O_P(1)$.
Finally, $\|E_n[DT]\|_2=O_P(1)$ and $\|E_n[DT]-E_P[DT]\|_2=O_P(n^{-1/2})$ since $d$ is fixed and $D_\ell T$ is subexponential.
Therefore,
\[
\|\widehat\gamma_T-\gamma_T\|_2 = O_P(n^{-1/2}).
\]
Then for each $i$, $|e_{T,i}|\le \|D_i\|_2\|\widehat\gamma_T-\gamma_T\|_2$, so
\[
\|e_T\|_{4,n}\le \|\widehat\gamma_T-\gamma_T\|_2\big\|\|D_i\|_2\big\|_{4,n} = O_P(n^{-1/2}),
\]
and hence $\|e_T\|_{2,n}\le \|e_T\|_{4,n}=O_P(n^{-1/2})$. Using also $\|\widetilde T\|_{2,n}=O_P(1)$, we conclude
\[
|\widehat\Omega-\Omega|=O_P(n^{-1/2}).
\]

By the WLLN, $E_n[\widetilde T^2]\to \Omega$ in probability. Since $\widehat\Omega=E_n[\widehat{\widetilde T}^2]
=E_n[\widetilde T^2]+o_P(1)$ (because $E_n[\widehat{\widetilde T}^2-\widetilde T^2]=O_P(n^{-1/2})$),
we have $\widehat\Omega\to \Omega$ in probability. In particular,
\[
P(\widehat\Omega\ge C_1/2)\to 1.
\]
On the event $\{\widehat\Omega\ge C_1/2\}$,
\[
\big|\widehat\Omega^{-1}-\Omega^{-1}\big|
=\frac{|\widehat\Omega-\Omega|}{\widehat\Omega\,\Omega}
\le \frac{2}{C_1^2}|\widehat\Omega-\Omega|
=O_P(n^{-1/2}).
\]

To control $\|\widehat{\widetilde T}\|_{4,n}$, note $\widehat{\widetilde T}=\widetilde T+e_T$, so by Minkowski,
\[
\|\widehat{\widetilde T}\|_{4,n}\le \|\widetilde T\|_{4,n}+\|e_T\|_{4,n}=O_P(1).
\]

Finally, consider $\max_{j\le p}\|\widehat{\widetilde U}_j\|_{4,n}$:
\[
\max_{j\le p}\|\widehat{\widetilde U}_j\|_{4,n}
\le \max_{j\le p}\|\widehat{\widetilde Y}_j\|_{4,n}
+ \max_{j\le p}|\widehat\theta_j|\,\|\widehat{\widetilde T}\|_{4,n}.
\]
Write $\widehat{\widetilde Y}_j=\widetilde Y_j+e_{Y_j}$ with
\[
e_{Y_j,i}:=\widehat{\widetilde Y}_{ij}-\widetilde Y_{ij}=-D_i^{\mathtt T}(\widehat\gamma_{Y_j}-\gamma_{Y_j}),
\qquad
\widehat\gamma_{Y_j}:=(E_n[DD^{\mathtt T}])^{-1}E_n[DY_j].
\]
Then $|e_{Y_j,i}|\le \|D_i\|_2\|\widehat\gamma_{Y_j}-\gamma_{Y_j}\|_2$ implies
\[
\max_{j\le p}\|e_{Y_j}\|_{4,n}
\le \big\|\|D_i\|_2\big\|_{4,n}\,\max_{j\le p}\|\widehat\gamma_{Y_j}-\gamma_{Y_j}\|_2.
\]

Decompose
\[
\widehat\gamma_{Y_j}-\gamma_{Y_j}
=(E_n[DD^{\mathtt T}])^{-1}\big(E_n[DY_j]-E_P[DY_j]\big)
+\big((E_n[DD^{\mathtt T}])^{-1}-\Sigma_D^{-1}\big)E_P[DY_j].
\]
Since $d$ is fixed, $(E_n[DD^{\mathtt T}])^{-1}=O_P(1)$ in operator norm, and
$\big\|(E_n[DD^{\mathtt T}])^{-1}-\Sigma_D^{-1}\big\|_{\mathrm{op}}=O_P(n^{-1/2})$ as shown above.
Also $\sup_j\|E_P[DY_j]\|_2<\infty$ from the earlier uniform bound.
For the empirical term, each component of $DY_j$ is a product of subgaussians and hence subexponential uniformly in $j$;
thus entrywise Bernstein plus a union bound over $j\le p$ (and $\ell\le d$) yields
\[
\max_{j\le p}\|E_n[DY_j]-E_P[DY_j]\|_2 = O_P\!\Big(\sqrt{\frac{\log p}{n}}\Big).
\]
Therefore,
\[
\max_{j\le p}\|\widehat\gamma_{Y_j}-\gamma_{Y_j}\|_2
=O_P\!\Big(\sqrt{\frac{\log p}{n}}\Big).
\]
Consequently, since $\big\|\|D_i\|_2\big\|_{4,n}=O_P(1)$,
\[
\max_{j\le p}\|e_{Y_j}\|_{4,n}=O_P\!\Big(\sqrt{\frac{\log p}{n}}\Big).
\]

Then, using $\max_j\|\widetilde Y_j\|_{4,n}=O_P(1)$ from the preliminaries and Minkowski,
\[
\max_{j\le p}\|\widehat{\widetilde Y}_j\|_{4,n}
\le \max_{j\le p}\|\widetilde Y_j\|_{4,n}+\max_{j\le p}\|e_{Y_j}\|_{4,n}
=O_P(1).
\]

Lastly, consider $\widehat\theta_j-\theta_j$:
\[
\widehat\theta_j-\theta_j
=\widehat\Omega^{-1}\big(E_n[\widehat{\widetilde T}\widehat{\widetilde Y}_j]-E_P[\widetilde T\widetilde Y_j]\big)
+(\widehat\Omega^{-1}-\Omega^{-1})E_P[\widetilde T\widetilde Y_j].
\]
We already established $|\widehat\Omega^{-1}-\Omega^{-1}|=O_P(n^{-1/2})$ and $\sup_j|E_P[\widetilde T\widetilde Y_j]|<\infty$.
For the remaining term,
\[
E_n[\widehat{\widetilde T}\widehat{\widetilde Y}_j]-E_P[\widetilde T\widetilde Y_j]
=
\big(E_n[\widetilde T\widetilde Y_j]-E_P[\widetilde T\widetilde Y_j]\big)
+E_n[e_T\widehat{\widetilde Y}_j]+E_n[\widetilde T e_{Y_j}],
\]
and (as previously done) Bernstein/union bounds for $\max_j|E_n[\widetilde T\widetilde Y_j]-E_P[\widetilde T\widetilde Y_j]|$,
together with H\"older and the established bounds $\|e_T\|_{4,n}=O_P(n^{-1/2})$,
$\max_j\|\widehat{\widetilde Y}_j\|_{4,n}=O_P(1)$, and $\max_j\|e_{Y_j}\|_{4,n}=O_P(\sqrt{\log p/n})$, yield
\[
\max_{j\le p}|\widehat\theta_j-\theta_j|
=O_P\!\Big(\sqrt{\frac{\log p}{n}}\Big).
\]
Since $\sup_j|\theta_j|<\infty$, this implies $\max_{j\le p}|\widehat\theta_j|=O_P(1)$.
Therefore,
\[
\max_{j\le p}\|\widehat{\widetilde U}_j\|_{4,n}=O_P(1).
\]

Combining all of these results, we conclude
\[
(I)=O_P(n^{-1/2}).
\]

\noindent
\textbf{(II).}
By H\"older,
\[
(II)\le \Omega^{-1}\|e_T\|_{4,n}\max_{j\le p}\|\widehat{\widetilde U}_j\|_{4,n}.
\]
Recalling $\|e_T\|_{4,n}=O_P(n^{-1/2})$ and $\max_{j\le p}\|\widehat{\widetilde U}_j\|_{4,n}=O_P(1)$, we immediately get
\[
(II)=O_P(n^{-1/2}).
\]

\noindent
\textbf{(III).}
By H\"older,
\[
(III)\le \Omega^{-1}\|\widetilde T\|_{4,n}\max_{j\le p}\|\widehat{\widetilde U}_j-\widetilde U_j\|_{4,n}.
\]
We have already concluded $\|\widetilde T\|_{4,n}=O_P(1)$. Also,
\[
\widehat{\widetilde U}_{ij}-\widetilde U_{ij}
= e_{Y_j,i}-(\widehat\theta_j-\theta_j)\widehat{\widetilde T}_i-\theta_j e_{T,i}.
\]
Thus by Minkowski,
\[
\max_{j\le p}\|\widehat{\widetilde U}_j-\widetilde U_j\|_{4,n}
\le
\max_{j\le p}\|e_{Y_j}\|_{4,n}
+\max_{j\le p}|\widehat\theta_j-\theta_j|\,\|\widehat{\widetilde T}\|_{4,n}
+\|e_T\|_{4,n}\max_{j\le p}|\theta_j|.
\]
Using the previously established bounds
\[
\max_{j\le p}\|e_{Y_j}\|_{4,n}=O_P\!\Big(\sqrt{\frac{\log p}{n}}\Big), \quad
\max_{j\le p}|\widehat\theta_j-\theta_j|=O_P\!\Big(\sqrt{\frac{\log p}{n}}\Big),
\]
\[
\|\widehat{\widetilde T}\|_{4,n}=O_P(1),\quad
\|e_T\|_{4,n}=O_P(n^{-1/2}),\quad
\max_{j\le p}|\theta_j|=O(1),
\]
we conclude
\[
(III)=O_P\!\Big(\sqrt{\frac{\log p}{n}}\Big).
\]

\noindent
\textbf{Rates.}
The preceding analyses yield
\[
\max_{j\le p}\|\widehat\psi_{\cdot j}-\psi_{\cdot j}\|_{2,n}
=O_P\!\Big(\sqrt{\frac{\log p}{n}}\Big).
\]

Under Assumption \ref{ass_rate}, we have
\[
\sqrt{\frac{\log p}{n}} = o\!\left(\frac{1}{\log(pn)}\right).
\]
Therefore,
\[
\max_{j\le p}\sqrt{E_n[(\widehat\psi_{ij}-\psi_{ij})^2]}
=
\max_{j\le p}\|\widehat\psi_{\cdot j}-\psi_{\cdot j}\|_{2,n}
=
o_P\!\left(\frac1{\log(pn)}\right).
\]
\end{proof}

\Hdlin*
\begin{proof}
By the triangle inequality,
\[
\|R_n\|_\infty
\le
\big|\widehat\Omega^{-1}-\Omega^{-1}\big|
\max_{j\le p}\big|\sqrt{n}E_n[\widetilde T\,\widetilde U_j]\big|
+
\big|\widehat\Omega^{-1}\big|
\max_{j\le p}\big|\sqrt{n}E_n[(\widehat{\widetilde T}-\widetilde T)\widetilde U_j]\big|.
\]
From the proof of Proposition \ref{thm_hdreg}, we know that
\[
\big|\widehat\Omega^{-1}-\Omega^{-1}\big|=O_P(n^{-1/2}),
\quad
\big|\widehat\Omega^{-1}\big|=O_P(1), \quad \max_{j\le p}\big|E_n[\widetilde T\,\widetilde U_j]\big|
=
O_P\!\left(\sqrt{\frac{\log p}{n}}\right)
\]
(specifically recalling that $\widetilde T\widetilde U_j$ is subexponential from the work done in Proposition \ref{thm_hdreg} and that $E[\widetilde{T} \widetilde{U}_j]=0$).
Also, since
\[
\widehat{\widetilde T}_i-\widetilde T_i
=
-D_i^{\mathtt T}(\widehat\gamma_T-\gamma_T),
\]
we have for each \(j\),
\[
E_n[(\widehat{\widetilde T}-\widetilde T)\widetilde U_j]
=
-(\widehat\gamma_T-\gamma_T)^{\mathtt T}E_n[D\,\widetilde U_j].
\]
Hence using Cauchy-Schwarz,
\[
\max_{j\le p}\big|E_n[(\widehat{\widetilde T}-\widetilde T)\widetilde U_j]\big|
\le
\|\widehat\gamma_T-\gamma_T\|_2
\max_{j\le p}\|E_n[D\,\widetilde U_j]\|_2.
\]
Again by the proof of Proposition \ref{thm_hdreg},
\[
\|\widehat\gamma_T-\gamma_T\|_2 = O_P(n^{-1/2}),
\]
while entrywise Bernstein and a union bound over \(j\le p\) and the fixed coordinates of \(D\) give
\[
\max_{j\le p}\|E_n[D\,\widetilde U_j]\|_2
=
O_P\!\left(\sqrt{\frac{\log p}{n}}\right),
\]
since each coordinate of \(D\widetilde U_j\) is subexponential uniformly in \(j\) and its expectation is zero (and then using the same logic as used in the proof of Proposition \ref{thm_hdreg}). Therefore,
\[
\max_{j\le p}\big|E_n[(\widehat{\widetilde T}-\widetilde T)\widetilde U_j]\big|
=
O_P\!\left(\frac{\sqrt{\log p}}{n}\right).
\]
Combining the preceding results then,
\[
\|R_n\|_\infty
\le
O_P(n^{-1/2})\sqrt{n}O_P\!\left(\sqrt{\frac{\log p}{n}}\right)
+
O_P(1)\sqrt{n}O_P\!\left(\frac{\sqrt{\log p}}{n}\right)
=
O_P\!\left(\sqrt{\frac{\log p}{n}}\right).
\]
Under Assumption \ref{ass_rate},
$\sqrt{\log p / n}\cdot \sqrt{\log(pn)} \le \log(pn)/\sqrt{n} = o(1)$, so
$\|R_n\|_\infty = O_P(\sqrt{\log p/n}) = o_P\big(1/\sqrt{\log(pn)}\big)$,
proving the stated claim.
\end{proof}

\end{document}